%% file: main.tex
\RequirePackage{luatex85}

\documentclass[11pt,a4paper]{article}
\usepackage{iftex}

\ifPDFTeX
\usepackage[utf8]{inputenc}
\usepackage[T1]{fontenc}
\fi

\usepackage{amsmath}
\usepackage{amssymb}
\usepackage{amsthm}
\usepackage{mathtools}
\usepackage{thmtools}

\ifPDFTeX
\usepackage{libertine}
\usepackage[libertine]{newtxmath} 
\usepackage[scaled=0.96]{zi4} 
\else
\usepackage{fontspec}
\usepackage{unicode-math}
\fi

\usepackage[DIV=11]{typearea} 

\usepackage{microtype}

\usepackage{xcolor}
\usepackage{enumerate}

\usepackage[algo2e,ruled,vlined,linesnumbered]{algorithm2e}

\usepackage{dsfont}

\usepackage[
	style=alphabetic,
	backref=true,
	doi=true,
	url=true,
	maxcitenames=3,
	mincitenames=3,
	maxbibnames=10,
	minbibnames=10,
	backend=biber
]{biblatex}

\usepackage[ocgcolorlinks]{hyperref} 

\input{commands}

\colorlet{DarkRed}{red!50!black}
\colorlet{DarkGreen}{green!50!black}
\colorlet{DarkBlue}{blue!50!black}

\hypersetup{
	linkcolor = DarkRed,
	citecolor = DarkGreen,
	urlcolor = DarkBlue,
	bookmarksnumbered = true,
	linktocpage = true
}

\allowdisplaybreaks[3]

\DontPrintSemicolon
\SetKwInOut{Input}{Input}\SetKwInOut{Output}{Output}

\SetCommentSty{mycommfont}

\newlength{\leftstackrelawd} 
\newlength{\leftstackrelbwd}
\def\leftstackrel#1#2{\settowidth{\leftstackrelawd}%
{${{}^{#1}}$}\settowidth{\leftstackrelbwd}{$#2$}%
\addtolength{\leftstackrelawd}{-\leftstackrelbwd}%
\leavevmode\ifthenelse{\lengthtest{\leftstackrelawd>0pt}}%
{\kern-.5\leftstackrelawd}{}\mathrel{\mathop{#2}\limits^{#1}}}

\makeatletter
\newcommand{\mylabel}[2]{#2\def\@currentlabel{#2}\label{#1}}
\makeatother

\declaretheorem[numberwithin=section]{theorem}
\declaretheorem[numberlike=theorem]{lemma}

\declaretheorem[numberlike=theorem]{corollary}
\declaretheorem[numberlike=theorem]{definition}

\declaretheorem[numberlike=theorem]{observation}

\DeclareMathOperator{\polylog}{polylog}
\DeclareMathOperator{\poly}{poly}
\DeclareMathOperator{\lse}{lse}
\DeclareMathOperator{\Exp}{\mathbb{E}}
\DeclareMathOperator{\diag}{diag}

\providecommand{\qmaxstretch}[2][]{\ensuremath{(W^{-1} A_{#1}^T #2)_{\max}}}
\providecommand{\pobjvalue}[1]{\ensuremath{\ones^TW #1}}

\providecommand{\potfun}{\Phi}

\providecommand{\nablaalg}{\texttt{gradient\_transshipment}\xspace}

\providecommand{\softmax}[2][]{\lse_{#1}\left(#2\right)}
\providecommand{\h}{\tilde{h}}

\providecommand{\ssspalgo}{\texttt{sssp}\xspace}
\providecommand{\treealgo}{\texttt{primal\_tree}\xspace}

\providecommand{\bt}{\tilde b}
\providecommand{\xt}{\tilde x}
\providecommand{\ft}{\tilde f}
\providecommand{\sumarcs}{\mathfrak{s}}

\providecommand{\maxarcs}{\mathfrak{m}}

\newcommand{\p}{\pi}
\newcommand{\ones}{\mathds{1}}
\newcommand{\ZZ}{\ensuremath{\mathbb{Z}}}
\newcommand{\RR}{\ensuremath{\mathbb{R}}}

\renewcommand{\epsilon}{\varepsilon}

\title{Near-Optimal Approximate Shortest~Paths and Transshipment in Distributed and Streaming Models\thanks{Accepted to \emph{SIAM Journal on Computing}. A preliminary version of this paper was presented at the \emph{31st International Symposium on Distributed Computing (DISC 2017)}.}}

\author{
  Ruben Becker\thanks{Gran Sasso Science Institute, L'Aquila, Italy}
  \and
  Sebastian Forster\thanks{Department of Computer Sciences, University of Salzburg, Austria; this author previously published under the name Sebastian Krinninger}
  \and
  Andreas Karrenbauer\thanks{Max Planck Institute for Informatics, Saarland Informatics Campus, Saarbrücken, Germany}
  \and
  Christoph Lenzen\thanks{Max Planck Institute for Informatics, Saarland Informatics Campus, Saarbrücken, Germany}
}

\date{}

\hypersetup{
	pdftitle = {Near-Optimal Approximate Shortest~Paths and Transshipment in Distributed and Streaming Models},
	pdfauthor = {Ruben Becker, Sebastian Forster, Andreas Karrenbauer, Christoph Lenzen}
}

\begin{document}
\maketitle
\begin{abstract}
\input{abstract}
\end{abstract}

\newpage

\tableofcontents

\newpage

\input{introduction}
\input{transshipment_asymmetric}
\input{sssp}
\input{distributed}
\input{streaming}
\input{spanner}

\printbibliography[heading=bibintoc] 

\end{document}

%% file: commands.tex
\makeatletter
\g@addto@macro\bfseries{\boldmath}
\makeatother

\makeatletter
\g@addto@macro\mdseries{\unboldmath}
\g@addto@macro\normalfont{\unboldmath}
\g@addto@macro\rmfamily{\unboldmath}
\g@addto@macro\upshape{\unboldmath}
\makeatother

\DeclareCiteCommand{\citem}
    {}
    {\mkbibbrackets{\bibhyperref{\usebibmacro{postnote}}}}
    {\multicitedelim}
    {}

\renewcommand*{\labelalphaothers}{\textsuperscript{+}}

\renewcommand*{\multicitedelim}{\addcomma\space}

\AtEveryCitekey{\clearfield{note}}

\makeatletter
\AtBeginDocument{%
    \newlength{\temp@x}%
    \newlength{\temp@y}%
    \newlength{\temp@w}%
    \newlength{\temp@h}%
    \def\my@coords#1#2#3#4{%
      \setlength{\temp@x}{#1}%
      \setlength{\temp@y}{#2}%
      \setlength{\temp@w}{#3}%
      \setlength{\temp@h}{#4}%
      \adjustlengths{}%
      \my@pdfliteral{\strip@pt\temp@x\space\strip@pt\temp@y\space\strip@pt\temp@w\space\strip@pt\temp@h\space re}}%
    \ifpdf
      \typeout{In PDF mode}%
      \def\my@pdfliteral#1{\pdfliteral page{#1}}
      \def\adjustlengths{}%
    \fi
    \ifxetex
      \def\my@pdfliteral #1{}
      \def\adjustlengths{\setlength{\temp@h}{-\temp@h}\addtolength{\temp@y}{1in}\addtolength{\temp@x}{-1in}}%
    \fi%
    \def\Hy@colorlink#1{%
      \begingroup
        \ifHy@ocgcolorlinks
          \def\Hy@ocgcolor{#1}%
          \my@pdfliteral{q}%
          \my@pdfliteral{7 Tr}
        \else
          \HyColor@UseColor#1%
        \fi
    }%
    \def\Hy@endcolorlink{%
      \ifHy@ocgcolorlinks%
        \my@pdfliteral{/OC/OCPrint BDC}%
        \my@coords{0pt}{0pt}{\pdfpagewidth}{\pdfpageheight}%
        \my@pdfliteral{F}
        \my@pdfliteral{EMC/OC/OCView BDC}%
        \begingroup%
          \expandafter\HyColor@UseColor\Hy@ocgcolor%
          \my@coords{0pt}{0pt}{\pdfpagewidth}{\pdfpageheight}%
          \my@pdfliteral{F}
        \endgroup%
        \my@pdfliteral{EMC}%
        \my@pdfliteral{0 Tr}
        \my@pdfliteral{Q}%
      \fi
      \endgroup
    }%
}
\makeatother

%% file: abstract.tex
We present a method for solving the \emph{transshipment} problem---also known as uncapacitated minimum cost flow---up to a multiplicative error of $ 1 + \varepsilon $ in undirected graphs with non-negative edge weights using a tailored gradient descent algorithm.
Using $\widetilde{O}(\cdot)$ to hide polylogarithmic factors in $n$ (the number of nodes in the graph), our gradient descent algorithm takes $ \widetilde O(\varepsilon^{-2}) $ iterations, and in each iteration it solves an instance of the transshipment problem up to a multiplicative error of $ \polylog n $.
In particular, this allows us to perform a single iteration by computing a solution on a sparse spanner of logarithmic stretch.
Using a randomized rounding scheme, we can further extend the method to finding approximate solutions for the single-source shortest paths (SSSP) problem.
As a consequence, we improve upon prior works by obtaining the following results:
\begin{enumerate}
  \item \textit{Broadcast CONGEST model:} $ (1 + \varepsilon) $-approximate SSSP using $ \widetilde{O} ((\sqrt{n} + D) \varepsilon^{-3}) $ rounds, where $ D $ is the (hop) diameter of the network.
  \item \textit{Broadcast congested clique model:} $ (1 + \varepsilon) $-approximate transshipment and SSSP using $ \widetilde{O} (\varepsilon^{-2}) $ rounds.
  \item \textit{Multipass streaming model:} $ (1 + \varepsilon) $-approximate transshipment and SSSP using $ \widetilde{O} (n) $ space and $ \widetilde{O} (\varepsilon^{-2}) $ passes.
\end{enumerate}
The previously fastest SSSP algorithms for these models leverage sparse hop sets.
We bypass the hop set construction; computing a spanner is sufficient with our method.
The above bounds assume non-negative edge weights that are polynomially bounded in $n$; for general non-negative weights, there is an additional multiplicative overhead equal to the logarithm of the maximum ratio between non-zero weights.
Our algorithms can also handle asymmetric costs for traversing edges in opposite directions. In this case, we obtain an additional multiplicative dependence of the maximum ratio between the two costs on some edge.

%% file: introduction.tex
\section{Introduction}

Single-source shortest paths (SSSP) is a fundamental and well-studied problem in computer science.
Thanks to sophisticated algorithms and data structures~\cite{FredmanT87,Thorup99,HansenKTZ15}, it has been known for a long time how to obtain (near-)optimal running time in the RAM model.
This is not the case in non-centralized models of computation, which become more and more relevant in a big-data world.
Despite progress for \emph{exact} SSSP algorithms~\cite{Spencer97,KleinS97,Cohen97,BrodalTZ98,ShiS99,CensorHillelKKLPS15,Gall16,Elkin17,GhaffariL18,ForsterN18}, there remain large gaps to the strongest known lower bounds.
Close-to-optimal running times have so far only been achieved by efficient approximation schemes~\cite{Blelloch0ST16,Cohen00,MeyerS03,MillerPVX15,HenzingerKN16,ElkinN16}.
For instance, in the CONGEST model of distributed computing, the state of the art is as follows:
exact SSSP on weighted graphs can be computed in $ \widetilde O (\sqrt{nD}) $\footnote{Throughout this paper we use $\widetilde{O}(\cdot)$-notation to hide polylogarithmic factors in $n$.} or $\widetilde O(\sqrt{n}D^{1/4} + n^{3/5} + D)$ rounds~\cite{ForsterN18}, where $ D $ is the (hop) diameter of the graph, and $ (1 + \varepsilon) $-approximate SSSP can be computed in $ (\sqrt{n} + D) \cdot 2^{O (\sqrt{\log{n} \log{(\varepsilon^{-1} \log{n})}})}$ rounds~\cite{HenzingerKN16}.\footnote{Note that these running times refer to weighted graphs. In unweighted graphs, the SSSP problem can easily be solved in $ O(D) $ rounds by performing a BFS tree computation.}
Even for constant $\varepsilon$, the latter exceeds the strongest known lower bound of $ \Omega (\sqrt{n / \log n} + D) $ rounds~\cite{Elkin06} by a super-polylogarithmic factor.
The techniques developed in the present paper allow us to make a qualitative algorithmic improvement for $ (1 + \varepsilon) $-approximate SSSP in this model: we solve the problem in $ \widetilde O ( (\sqrt{n} + D) \cdot \varepsilon^{-2}) $ rounds.
We thus narrow the gap between upper and lower bound significantly and additionally improve the dependence on $ \varepsilon $.
Our new approach achieves its superior running time by leveraging techniques from continuous optimization.

Our approach is based on tackling a problem that is more general than SSSP.
In the \emph{transshipment} problem, we seek to find a cheapest routing for sending units of a single good from sources to sinks along the edges of a graph meeting the nodes' demands.
Equivalently, we want to find the minimum-cost flow in a graph where edges have unlimited capacity.
The special case of SSSP can be modeled as a transshipment problem by setting the demand of the source to $ -n + 1 $ (thus supplying $ -n + 1 $ units) and the demand of every other node to~$ 1 $.
Unfortunately, this relation breaks when we consider approximation schemes:
A $(1+\varepsilon)$-approximate solution to the transshipment problem merely yields $(1+\varepsilon)$-approximations to the distances \emph{on average.}
In the special case of SSSP, however, one is interested in obtaining a $(1+\varepsilon)$-approximation to the distance for \emph{each single node} and we show how to extend our algorithm to provide such a guarantee as well.

Techniques from continuous optimization have been key to recent breakthroughs in the combinatorial realm of graph algorithms~\cite{DaitchS08,ChristianoKMST11,Sherman13,KelnerLOS14,Madry13,LeeS14,CohenMSV16}.
In this paper, we apply this paradigm to computing $ (1 + \varepsilon) $-approximate primal and dual solutions to the transshipment problem in undirected graphs with non-negative edge weights.
Accordingly, we perform projected gradient descent for a suitable norm-minimization formulation of the problem, where we approximate the infinity norm by a differentiable soft-max function.
To make this general approach work in our setting, we need to add significant problem-specific tweaks.
In particular, we develop a gradient descent algorithm that reduces the problem of computing a $ (1+\varepsilon) $-approximation to the more relaxed problem of computing a crude (e.g.\ $ O (\log{n}) $ factor) approximation.
We then exploit that an $ O(\log{n}) $-approximation can be computed very efficiently by solving the problem on a sparse spanner, and that it is well-known how to compute sparse spanners efficiently.
To obtain the aforementioned per-node guarantee in the approximate SSSP problem, we provide an intricate randomized rounding scheme that yields an approximate shortest-path tree (i.e., a primal solution) in addition to the dual solution (i.e., estimated distances to the source).

Our method is widely applicable among a plurality of non-centralized models of computation in rather straightforward ways.
We obtain the first non-trivial algorithms for approximate undirected transshipment in the Broadcast CONGEST, Broadcast Congested Clique, and Multipass Streaming models.
As a further, arguably more important, consequence, we improve upon prior results for computing approximate SSSP in these models.
Our approximate SSSP algorithms are the first to be provably optimal up to polylogarithmic factors.

\subparagraph*{Our Contributions and Results}

We summarize our technical and conceptual contributions as follows:
\begin{enumerate}
\item[\mylabel{contribution gradient descent}{(C1)}] We give a problem-specific gradient descent algorithm for approximating the transshipment problem, which requires access to an oracle computing an $ \alpha $-approximate dual solution for any given demand vector.\footnote{Note that dual feasibility is crucial here. In particular, this rules out an oracle based on tree embeddings~\cite{Bartal98,ElkinEST08}, as such trees might have stretch $ \Omega (n) $ on individual edges.}
To compute a $ (1+\varepsilon) $-approximation, the algorithm performs $ \widetilde O (\varepsilon^{-2} \alpha^2) $ oracle calls.
If the oracle returns primal solutions, so does our algorithm.
\item[\mylabel{contribution analysis of gradient}{(C2)}] We provide a randomized rounding scheme that, given the output of the gradient descent algorithm for an SSSP problem, can be used to obtain distance estimates that satisfy a per-node approximation guarantee.
\item[\mylabel{contribution solving on spanner}{(C3)}] We observe that spanners can be used to obtain an efficient transshipment oracle with approximation guarantee $\alpha \in O(\log n)$.
\end{enumerate}
By implementing our method in specific models of computation, we obtain the following algorithmic results in graphs with non-negative polynomially bounded\footnote{For general non-negative weights, running times scale by a multiplicative factor of $\log R$, where $R$ is the maximum ratio between non-zero edge weights.} edge weights:
\begin{enumerate}
  \item[\mylabel{result SSSP}{(R1)}] We give faster Las Vegas algorithms\footnote{The conference version of this paper claimed deterministic bounds for computing approximate single-source distances. However, this claim was supported with inaccurate arguments. The deterministic guarantee can still be obtained for computing the $s$-$t$ distance between two nodes.} for computing $ (1 + \varepsilon) $-approximate SSSP:
  \begin{enumerate}
  \item \emph{Broadcast CONGEST model:} We obtain an algorithm for computing $ (1 + \varepsilon) $-approximate SSSP in $ \widetilde{O} ((\sqrt{n} + D) \cdot \varepsilon^{-3}) $ rounds.
  This improves upon the previous best upper bound of $ (\sqrt{n} + D) \cdot 2^{O (\sqrt{\log{n} \log{(\varepsilon^{-1} \log{n})}})} $ rounds~\cite{HenzingerKN16}.
  For $\varepsilon^{-1}\in O (\polylog n)$, we match, up to polylogarithmic factors in $ n $, the lower bound of $ \widetilde\Omega (\sqrt{n} + D) $~\cite{Elkin06,DasSarmaHKKNPPW12}, which applies to any ($ \poly{n} $)-approximation of the distance between two fixed nodes in a weighted undirected graph.
  If a Monte-Carlo guarantee suffices, we obtain a slightly better guarantee of $ \widetilde{O} ((\sqrt{n} + D) \cdot \varepsilon^{-3/2}) $.
  \item \emph{Broadcast Congested Clique model:} We obtain an algorithm for computing $ (1 + \varepsilon) $-approximate SSSP using $ \widetilde{O} (\varepsilon^{-2}) $ rounds.
  This improves upon the previous best upper bound of $ 2^{O (\sqrt{\log{n} \log{(\varepsilon^{-1} \log{n})}})} $ rounds~\cite{HenzingerKN16}.
  \item \emph{Multipass Streaming model:} We obtain an algorithm for computing $ (1 + \varepsilon) $-approximate SSSP using $ \widetilde{O} (\varepsilon^{-2}) $ passes and $ O (n \log^2{n}) $ space.
This improves upon the previous best upper bound of $ (2+1/\varepsilon)^{O(\sqrt{\log{n} \log \log n})} $ passes and $ O (n \log^2 n) $ space~\cite{ElkinN16}.
By setting $ \varepsilon $ small enough, we can compute distances up to the value $ \log n $ exactly in integer-weighted graphs using $\polylog n $ passes and $ O (n \log^2{n}) $ space.
Thus, up to polylogarithmic factors in $ n $, our result matches a lower bound of $ \frac{n^{1 + \Omega(1/p)}}{\poly{p}} $ space for all algorithms that decide in $ p $ passes if the distance between two fixed nodes in an unweighted undirected graph is at most $ 2 (p + 1) $ for any $ p = O (\log{n} / \log{\log{n}}) $~\cite{GuruswamiO13}.
  \end{enumerate}
\item[\mylabel{result transshipment}{(R2)}] We give fast algorithms for computing $ (1 + \varepsilon) $-approximate transshipments. No non-trivial upper bounds were known in each of these models before.
  \begin{enumerate}
  \item \emph{Broadcast CONGEST model:} A deterministic algorithm using $ \widetilde{O} (n \varepsilon^{-2}) $ rounds.
  \item \emph{Broadcast Congested Clique model:} A deterministic algorithm running in $ \widetilde{O} (\varepsilon^{-2}) $ rounds.
  \item \emph{Multipass Streaming model:} A deterministic algorithm using $ \widetilde{O} (\varepsilon^{-2}) $ passes and $ O (n \log{n}) $ space.
  \end{enumerate}
\end{enumerate}

In the case of SSSP, we can compute a $ (1 + \varepsilon) $-approximation to the distance of the source to every node together with a tree such that the length of the path from the source to any node is within a factor of $ (1 + \varepsilon) $ of its true distance.

In the case of the transshipment problem, we can (deterministically) return $ (1 + \varepsilon) $-approximate primal and dual solutions.
We can further extend the results to bidirected graphs, where each edge can be used in either direction with potentially asymmetric weights.
Denoting by $ \lambda \geq 1 $ the maximum over all edges of the weight ratio between traversing the edge in different directions, our algorithms give the same guarantees if the number of rounds or passes, respectively, is increased by a factor of $ \lambda^2 \log{\lambda} $.

\subparagraph*{Related Work on the Transshipment Problem}

Transshipment is a classic problem in combinatorial optimization~\cite{KorteV00,Schrijver03}.
The classic algorithms for directed graphs with non-negative edge weights in the RAM model run in time $ O (n (m + n \log{n}) \log{n}) $~\cite{Orlin93} and $ O ((m + n \log{n}) B) $~\cite{EdmondsK72}, respectively, where $ B $ is the sum of the nodes' demands (when they are given as integers) and the term $ m + n \log{n} $ comes from SSSP computations.
If the graph contains negative edge weights, then these algorithms require an additional preprocessing step to compute SSSP in presence of negative edge weights, for example in time $ O (m n) $ using the Bellman-Ford algorithm~\cite{Bellman58,Ford56} or in time $ O (m \sqrt{n} \log N) $ using Goldberg's algorithm~\cite{Goldberg95}.\footnote{Goldberg's running time bound holds for integer-weighted graphs with most negative weight $ -N $.}
The weakly polynomial running time was first improved to $ \widetilde O (m^{3/2} \polylog R) $~\cite{DaitchS08} and then to $ \widetilde O (m \sqrt{n} \polylog R) $ in a recent breakthrough for minimum-cost flow~\cite{LeeS14}, where $ R $ is the maximum ratio between edge weights.
Recently, Sherman~\cite{Sherman17} obtained a randomized algorithm for computing a $ (1+\varepsilon) $-approximate transshipment in weighted undirected graphs in time $ O(\varepsilon^{-2}m^{1+o(1)}) $ using a preconditioning approach.
Note that Sherman's preconditioner gives a condition number of $ n^{o(1)} $ and thus is not suitable for obtaining a polylogarithmic iteration count as in our approach.
We are not aware of any non-trivial algorithms for computing (approximate) transshipments in non-centralized models of computation, such as distributed or streaming models.

\subparagraph*{Comparison to Hop Set Based SSSP Algorithms}
The state-of-the art SSSP algorithms in the distributed CONGEST model follow the framework developed in~\cite{Nanongkai14}, where (1) the problem of computing SSSP is reduced to an overlay network of size $ N = \widetilde O (\sqrt{n}) $ and (2) a sparse hop set is constructed to speed up computing SSSP on the overlay network.
An $ (h, \varepsilon) $-hop set is a set of weighted edges that, when added to the original graph, provides sufficient shortcuts to approximate all pairwise distances using paths with only $ h $ edges (``hops'').
The algorithm by Nanongkai et al.~\cite{HenzingerKN16} achieves an upper bound of $ (\sqrt{n} + D) \cdot 2^{O (\sqrt{\log{n} \log{(\varepsilon^{-1} \log{n})}})} $ on the number of rounds by constructing an $ (h, \varepsilon) $-hop set of size $ O (N \rho) $ where $ h \leq 2^{O (\sqrt{\log{n} \log{(\varepsilon^{-1} \log{n})}})} $ and $ \rho \leq 2^{O (\sqrt{\log{n} \log{(\varepsilon^{-1} \log{n})}})} $.
Elkin's algorithm~\cite{Elkin17}, which takes $ O (D^{1/3} (n \log{n})^{2/3}) $ rounds, uses an exact $ (N / \rho, 0) $-hop set of size $ O (N \rho) $ similar to the one developed by Shi and Spencer for the PRAM model ~\cite{Spencer97}.
Elkin's main technical contribution lies in showing how to compute this hop set without constructing the overlay network explicitly.
Roughly speaking, in all these algorithms, \emph{both} $ h $ and $ \rho $ enter the running time of the corresponding SSSP algorithms, in addition to the time needed to construct the hop set.

The concept of hop sets has been introduced by Cohen in the context of PRAM algorithms for approximate SSSP~\cite{Cohen00}.
The increased interest in hop sets and their applications in the last years~\cite{Bernstein09,HenzingerKN18,MillerPVX15,HenzingerKN16,ElkinN16} has culminated in the construction of $ (h, \varepsilon) $-hop sets of size $ O (n^{1 + \frac{1}{2^{k + 1}} - 1}) $ for $ h = O \left((\tfrac{k}{\varepsilon})^{k}\right) $~\cite{HuangP19,ElkinN19}.
Recent lower bounds by Abboud et al.~\cite{AbboudBP18} show that this trade-off is essentially tight: for any integer~$ k, $ any construction of an $ (h, \varepsilon) $-hop set of size $ \leq n^{1 + \frac{1}{2^k - 1} - \delta} $ (for $ \delta > 0 $) must have $ h = \Omega (c_k / \epsilon^{k + 1}) $, where $ c_k $ is a constant depending only on $ k $.
This implies that the hop set based algorithms, as long as the factor $ \rho $ has to be paid in the running time for construction hop sets of size $ n \rho $, will never be able to achieve a running time comparable to our SSSP algorithm exclusively by finding better hop sets.

\subparagraph*{Further Related Work}\label{sec:further related work}
In the following we review further results on (approximate) SSSP in the non-centralized models considered in this paper.

In the \emph{CONGEST} model of distributed computing, SSSP on unweighted graphs can be computed exactly in $ O (D) $ rounds by distributed breadth-first search~\cite{Peleg:book}.
For weighted graphs, the only non-trivial algorithm known is the distributed version of Bellman-Ford~\cite{Bellman58,Ford56}, which uses $ O (n) $ rounds.
In terms of approximation, Elkin~\cite{Elkin06} showed that computing an $ \alpha $-approximate SSSP tree requires $ \Omega ((n/\alpha)^{1/2} / \log{n} + D) $ rounds.
Together with many other lower bounds, this was strengthened in~\cite{DasSarmaHKKNPPW12} by showing that computing a $ (\poly{n}) $-approximation of the distance between two fixed nodes in a weighted undirected graph requires $ \Omega (\sqrt{n} / \log{n} + D) $ rounds.
The lower bounds were complemented by two SSSP algorithms: a randomized $ O (\alpha \log \alpha) $-approximation using $ \widetilde O (n^{1/2+1/\alpha} + D) $ rounds~\cite{LenzenPS13} and a randomized $ (1 + o(1)) $-approximation using $ \widetilde O (n^{1/2} D^{1/4} + D) $ rounds~\cite{Nanongkai14}.
Both results were improved upon in~\cite{HenzingerKN16} by a deterministic algorithm that computes $ (1 + o(1)) $-approximate SSSP in $ n^{1/2 + o(1)} + D^{1 + o(1)} $ rounds.
A recent hop set construction of Elkin and Neiman~\cite{ElkinN16} improves the running time for computing shortest paths from multiple sources.

The \emph{Congested Clique} model~\cite{LotkerPPP03} has seen increasing interest in the past years as it highlights the aspect of limited bandwidth in distributed computing, yet excludes the possibility of explicit bottlenecks (e.g., a bridge that would limit the flow of information between the parts of the graph it connects to $O(\log n)$ bits per round in the CONGEST model).
For weighted graphs, SSSP can again be computed exactly in $ O (n) $ rounds.
The first approximation was given by Nanongkai~\cite{Nanongkai14} with a randomized algorithm for computing  $ (1+o(1)) $-approximate SSSP in $ \widetilde O (\sqrt{n}) $ rounds.
All-pairs shortest paths in the Congested Clique model can be computed deterministically in $ \widetilde O (n^{1/3}) $ rounds for an exact result and in $ O(n^{0.158}) $ rounds for a $ (1+o(1)) $-approximation~\cite{CensorHillelKKLPS15}.
The time for computing all-pairs shortest paths exactly has subsequently been improved to $ O (n^{0.2096}) $ rounds~\cite{Gall16}.
The hop set construction of~\cite{HenzingerKN16} gives a deterministic algorithm for computing $ (1+o(1)) $-approximate SSSP in $ n^{o(1)} $ rounds.
A recent hop set construction of Elkin and Neiman~\cite{ElkinN16} improves the running time for computing shortest paths from multiple sources.
We note that the approaches of~\cite{HenzingerKN16} and~\cite{ElkinN16} as well as our approach can actually operate in the more restricted \emph{Broadcast Congested Clique} and \emph{Broadcast CONGEST} models, in which in each round, each node sends the \emph{same} message to all other nodes or all its neighbors, respectively.
After the preliminary version of our paper appeared, Censor-Hillel et al.~\cite{Censor-HillelDK19} presented a deterministic algorithm that computes $ (1 + \epsilon) $-approximate multi-source shortest paths from up to $ O (\sqrt{n}) $ sources in $ O (\log^2 n / \epsilon) $ rounds in the Congested Clique model.

In the \emph{Streaming} model, two approaches were known for (approximate) SSSP before the algorithm of~\cite{HenzingerKN16}.
First, shortest paths up to distance $ d $ can be computed using $ d $ passes and $ O (n) $ space in unweighted graphs by breadth-first search.
Second, approximate shortest paths can be computed by first obtaining a sparse spanner and then computing distances on the spanner without additional passes~\cite{FeigenbaumKMSZ05,ElkinZ06,FeigenbaumKMSZ08,Baswana08,Elkin11}.
This leads, for example, to a randomized $ (2k-1) $-approximate all-pairs shortest paths algorithm using $ 1 $ pass and $ \widetilde O (n^{1+1/k}) $ space for any integer $ k \geq 2 $ in unweighted undirected graphs.
In unweighted undirected graphs, the spanner construction of~\cite{ElkinZ06} can be used to compute $ (1+o(1)) $-approximate SSSP using $ O (1) $ passes and $ O (n^{1+o(1)}) $ space.
The hop set construction of~\cite{HenzingerKN16} gives a deterministic algorithm for computing $ (1+o(1)) $-approximate SSSP in weighted undirected graphs using $ n^{o(1)} $ passes and $ n^{1+o(1)} $ space.
Using randomization, this was improved to $ n^{o(1)} $ passes and $ O (n \log^2 n) $ space by Elkin and Neiman~\cite{ElkinN16}.
These upper bounds are complemented by a lower bound of $ n^{1 + \Omega(1/p)} / (\poly{p}) $ space for all algorithms that decide in $ p $ passes if the distance between two fixed nodes in an unweighted undirected graph is at most $ 2 (p + 1) $ for any $ p = O (\log{n} / \log{\log{n}}) $~\cite{GuruswamiO13}.
Note that this lower bound in particular applies to all algorithms that provide $ 1 + \varepsilon $ approximations for $ \varepsilon < 1 / 2 (p + 1) $ in integer-weighted graphs, as this level of precision requires to compute shortest paths for small enough distances exactly.

%% file: transshipment_asymmetric.tex
\section{General Approach for Solving Transshipment and SSSP}\label{sec:transshipment}

We consider a bidirected graph $ G = (V, E, w) $ where $ V $ is the set of $ n $ nodes, $ E $ is the set of directed arcs, and $ w $ assigns a positive integer weight $ w_{(u,v)} $ to every arc $ (u, v) \in E $.
Being bidirected means that for every arc $ (u, v) \in E $ there also exists a reverse arc $ (v, u) \in E $.
We can thus view the arcs $ (u, v) \in E $ and $ (v, u) \in E $ as two orientations of an undirected edge $ \{u, v\} \in \bar{E} $, where $ \bar{E} = \{ \{u, v \} \in \binom{V}{2} \mid (u, v) \in E \} $ and $ m := | \bar{E} | $.
For every $ e \in \bar{E} $ we call one of the orientations the forward arc of weight~$ w^+_e $ and the other orientation the backward arc of weight $ w^-_e $.
Since the choice of names is arbitrary, let us, w.l.o.g., assume that $w^+\ge w^-$. Define $\lambda:=\max\left\{w^+_e/w^-_e: e\in \bar{E}\right\} \ge 1$ as the maximum ratio over all edges between the forward and the backward weight.
Every weighted undirected graph can be readily modeled as a weighted bidirected graph with $ \lambda = 1 $.
As $\lambda$ is easily determined and only a single value, we assume that it is a globally known parameter in the following. Furthermore, let $W=\diag(w^+, w^-)$ denote the $2m\times 2m$ diagonal matrix containing the forward and backward weights and let $b\in\ZZ^n$ be a vector of demands. W.l.o.g., we restrict to feasible and non-trivial instances, i.e., $b^T \ones = 0$ and $b\neq 0$.\footnote{Here $ \ones $ denotes the all-ones vector and thus $ b^T \ones = 0 $ simply means that the positive demands equal the negative demands (i.e., the supplies).}

As we will now briefly argue, it suffices to formulate our algorithms for positive integer arc weights from $ 1 $ to $ \| w \|_\infty $.
First, arbitrary positive arc weights can be reduced to positive integer arc weights by replacing each arc weight $ w_e $ by $ w_e' = \lceil \tfrac{w_e}{\epsilon w_{\min}} \rceil $ for $ w_{\min} = \min_{e \in E} \{w_e\} $, which conceptually corresponds to rounding up each arc weight to the next multiple of $ \epsilon w_{\min} $ and carrying out the computation with weights equal to the corresponding multiplicities.
This increases the weight of any path by a factor of at most $ (1 + \epsilon) $ and bounds the maximum arc weight $ \| w' \|_\infty $ by $ O (\epsilon^{-1} R) $, where $ R $ is the maximum ratio between non-zero arc weights in the original graph. As $ \epsilon^{-1} $ may be assumed to be polynomial in $ n $ -- otherwise exact algorithms will be faster than our approximation algorithms -- this bounds $ \log \| w' \|_\infty $ by $ O (\log n + \log R) $.
Second, given positive integer weights, we can deal with arc weight $ 0 $ by using weight $ w_e'' = 1 + (1 + \epsilon^{-1}) n w_e' $ for every arc $ e $, assuming without loss of generality that $ \epsilon^{-1} $ is integer (see, e.g.,~\cite{ForsterN18} for details). As this increases the maximum arc weight by a factor of $O(n \epsilon^{-1})$, also $ \log \| w'' \|_\infty $ is bounded by $ O (\log n + \log R) $.

A common approach to model the \emph{bidirected transshipment} problem as a linear program considers the node-arc-incidence matrix $A \in \{-1,0,1\}^{n \times 2m}$ of the bidirected graph.\footnote{The incidence matrix $ A $ of a directed graph contains a row for every node and a column for every arc and $ A_{ij} $ is $ 1 $ if the $j$-th arc enters the $i$-th node, $ -1 $ if it leaves the node, and $ 0 $ otherwise.}
For notational convenience, we list the columns of the forward arcs first, followed by the columns for the backward arcs in the same order.
The asymmetric transshipment problem can then be written as a primal/dual pair of linear programs:
\begin{equation}\label{formula:asymmetric}
	\min \{ \ones^TWx : Ax=b, x\ge 0 \}
	= \max\{b^Ty: (W^{-1} A^Ty)_{\max} \le 1\},
\end{equation}
where for $z\in \RR^d$ we denote by $(z)_{\max}:=\max\{z_i\colon i\in [d]\}$ the maximum entry of $z$. The primal (left) program asks to ``ship'' the flow given by $b$ from sources (negative demand) to sinks (positive demand) along the arcs of the graph, minimizing the cost of the flow, i.e., $\sum_{(u,v)\in E}w_{(u,v)} x_{uv}$. Note that if both $x_{uv}>0$ and $x_{vu}>0$, we can reduce both variables by $\min\{x_{uv},x_{vu}\}$ without changing whether $Ax=b$, but reducing the objective. Thus, an optimal solution sends flow only in one direction over any edge.

The dual (right) program asks for node potentials $y$ such that for each arc $(u,v)\in E$, $y_v-y_u\leq w_{(u,v)}$, maximizing $b^Ty$. Note that, because $b^T \ones = 0$, shifting the potential by $r \cdot \ones$ for any $r\in \RR$ does neither change $b^Ty$ nor $y_u-y_v$ for any $u,v\in V$. The goal of the dual is thus to maximize the differences in potential of sources and sinks (weighted according to $b$), subject to the constraint that the potential of $u$ does not exceed the potential of $v$ by more than $w_{(u,v)}$ for any neighbor $v$ of $u$.

In the special case of SSSP with source $s\in V$, we have that (i) $b_s=-n+1$ and $b_v=1$ for all $v\neq s$, (ii) an optimal primal solution $x^*$ is given by routing, for each $s\neq v\in V$, one unit of flow along a shortest path from $s$ to $v$, and (iii) optimal potentials $y^*$ are given by setting $y_v^*$ to the distance from $s$ to $v$.

\subsection{Overview of our Approach}\label{sec:overview}

We will employ gradient descent on a suitable potential function to converge to a near-optimal solution. If we can ensure large (relative) progress, few iterations of gradient descent suffice, and we can hope for the high parallelism that is crucial for good results in our target models of computation. Naturally, we will also require the iterations to allow for an efficient implementation. We will use a sparse representation of the distance structure of the graph, which can be ``kept available'' to determine approximate solutions to intermediate problems at any time.

A generic (unconstrained, adaptive step size, minimizing) gradient descent algorithm operates as follows:
\begin{enumerate}
  \item Pick a differentiable potential function that reflects the objective well and does not change too quickly. A convex potential function is desirable, as it guarantees the existence of an optimum solution and non-zero gradient at non-optimal solutions.
  \item Find a reasonable starting solution (a poor approximation typically suffices).
  \item Proceed in iterations. In each iteration:
  \begin{enumerate}
    \item Determine the gradient of the potential function at the current solution.
    \item Determine a direction for the update in which the gradient indicates that the potential reduces quickly.
    \item Choose a large step size, under the constraint that the gradient does not change too much along the way, and update the solution according to direction and step size.
    \item Check whether the solution is sufficiently close to the optimum. If yes, terminate and return the current solution. Otherwise, proceed with the next iteration.
  \end{enumerate}
\end{enumerate}
The key to minimizing the number of iterations is to find good update directions, in which the gradient is large \emph{and} does not change too quickly. In contrast, the key to fast iterations is to avoid too costly calculations for finding such update directions. The choice of the ``right'' potential function is crucial for being able to achieve a good compromise between these conflicting goals. It should also be noted that an efficient test for termination is required. Fortunately, our potential function will not only be convex, but being able to make significant progress will be equivalent to being far from the optimum; thus, the determined step size (which corresponds to the estimated reduction of the potential) can be used as a simple tool to bound the approximation ratio of the current solution.

Inconveniently, both the primal and dual program in \eqref{formula:asymmetric} have constraints. There are various techniques to handle this issue in gradient descent methods.
We will rephrase the dual problem so that a single equality constraint (defining an $(n-1)$-dimensional affine hyperplane) remains.
We then perform an unconstrained gradient descent on the subspace given by the hyperplane, i.e., we restrict to updates that do not affect this constraint.

This strategy requires to overcome a number of challenges, which we discuss before proceeding to presenting the high-level algorithm. First, we show that a sparse spanner gives rise to an efficient oracle yielding approximate solutions to transshipment problems in Section~\ref{sec:oracle}. This oracle provides a sufficiently good starting solution and, more importantly, is our tool for determining good update directions for the gradient descent steps. We then proceed to rephrasing the transshipment problem and defining a suitable potential function in Section~\ref{sec:potential}. We also analyze its gradient to establish the key properties needed for showing that the update steps guarantee fast progress and, at termination, a close-to-optimal solution. In Section~\ref{sec:algorithm}, we state the pseudocode of the algorithm and prove its correctness and give a bound on its number of iterations. The algorithm computes a dual solution only. However, we show how to also obtain a primal solution from an additional call to the oracle, under the condition that the oracle is capable of providing not only dual, but also primal solutions; the oracle from Section~\ref{sec:oracle} meets this criterion.

\subsection{An Efficient Oracle for \texorpdfstring{$O(\log n)$}{O(log n)}-Approximate Solutions}\label{sec:oracle}

Our approach is based on guiding gradient descent using approximate solutions to intermediate transshipment problems. To this end, an efficient way of obtaining such solutions is required. We phrase this in terms of calling an oracle that provides dual solutions to~\eqref{formula:asymmetric}. If a primal solution is required, too, the oracle must also be able to provide primal solutions.

For the sake of concreteness, we specify how to obtain such an oracle already at this point. It turns out that computing an exact solution on a \emph{sparse spanner} is a good choice for all applications of our method demonstrated in this article.
\begin{definition}[Spanner]
Given $G=(V,E,w)$ and $\alpha \geq 1$, an \emph{$\alpha$-spanner} of~$G$ is a subgraph $(V,E',w|_{E'})$, $E'\subseteq E$, in which distances are at most by factor $\alpha$ larger than in $G$.
\end{definition}
In other words, a spanner removes edges from $G$ while approximately preserving distances. It is well-known that for every undirected graph we can efficiently compute an $\alpha$-spanner of size $ O (n \log{n}) $ with $ \alpha = O (\log{n}) $. For the sake of completeness, we discuss derandomized implementations of the Baswana-Sen spanner construction~\cite{BaswanaS07} suitable for the computational models under consideration; this is deferred to Section~\ref{sec:spanner}\footnote{
    Note that while it is possible to construct spanners with $ O(n) $ edges, the fastest known such algorithms for weighted graphs are usually less efficient than those for constructing spanners with $ O(n \log n) $ edges.
}.

In all cases, the sparse representation of the approximate distance structure of the graph can be kept ``available,'' i.e., in the broadcast congested clique and broadcast congest models, we can make a spanner globally known, and in the multipass streaming model we may keep a spanner in memory, cf.~Section~\ref{sec:apps}. Note that, as the graph is bidirected with potentially asymmetric arc weights, an undirected $\alpha$-spanner construction cannot be directly applied to the bidirected graph~$G$. However, we can instead consider the symmetrized problem
\begin{equation}\label{formula:symmetric}
	\min \{ \ones^TW_-x : Ax=b, x\ge 0 \}
	= \max\{b^Ty: (W_-^{-1} A^Ty)_{\max} \le 1\},
\end{equation}
where $W_-=\diag(w^-,w^-)$.
\begin{observation}\label{obs:symmetric}
Feasible primal and dual solutions $x$ and $y$ of \eqref{formula:symmetric} are also feasible solutions of \eqref{formula:asymmetric}, where
\begin{enumerate}[(i)]
 \item $\ones^TW_-x\leq \ones^TWx\leq \lambda \ones^TW_-x$, and
 \item $\tfrac{1}{\lambda}(W_-^{-1} A^T y)_{\max} \le (W^{-1} A^T y)_{\max} \le (W_-^{-1} A^T y)_{\max}$.
\end{enumerate}
In particular, the objective value of $x$ increases by a factor between $1$ and $\lambda$ (while the objective value of $y$ is unaffected).
\end{observation}
We will not make direct use of this observation beyond obtaining a starting solution, as employing the symmetric weights $W_-$ streamlines our reasoning. However, this observation provides the intuition that the asymmetry of weights reduces the quality of approximation provided by the spanner by a factor of $\lambda$ and therefore should result in replacing $\alpha$ by $\alpha \lambda$ in time bounds; essentially, this intuition turns out to be correct.

Given an $\alpha$-spanner of the undirected graph $G_-=(V,E,w^-)$ and a demand vector~$b$, it is straightforward to compute an $ (\alpha \lambda) $-approximate \emph{primal-dual pair} for~\eqref{formula:asymmetric}. By this we mean feasible $x$ and $y$ satisfying that $\alpha \lambda b^Ty\geq \ones^TWx$. As we will fix the vector $b$ of the original problem in the following sections, but need to solve intermediate problems for different demand vectors $\bt$, the lemma formalizing this statement is phrased accordingly.
\begin{lemma}\label{lemma:spanner}
Given an $\alpha$-spanner $S$ of $ G^- = (V, \bar{E}, w^-) $ and any $0\neq \bt\in \RR^n$ with $\bt^T\ones = 0$, an $\alpha$-approximate primal-dual pair of solutions to~\eqref{formula:symmetric} with demands $\bt$ can be computed (without further knowledge of $G$).
\end{lemma}
\begin{proof}
Let $x$ and $y$ be optimal solutions to \eqref{formula:asymmetric} on the spanner. Observe that any feasible primal solution on $S$, padded with $0$-entries for edges not present in $S$, is feasible on $G$. We show that $\alpha^{-1} y$, whose objective is by factor $\alpha$ smaller than that of~$y$, is a feasible dual solution on $G$. As the objectives of $x$ and $y$ on the spanner and the padded $x$ on $G$ are identical, this proves the claim.

Let $ \{u, v\}=e \in E $ be arbitrary. By the definition of a spanner, there must be a path $ P_{uv} $ of weight at most $\alpha w_e^-$ from $ u $ to $ v $ in $S$. We have
\begin{align*}
|(W_-^{-1} A^T y)_{(u, v)}| &= \frac{| y_v - y_u|}{w^-_e}\\
&\leq \frac{\sum_{\{u',v'\} = e' \in P_{uv}}| y_{v'} - y_{u'}|}{w^-_e}\\
 &= \frac{\sum_{\{u',v'\} = e' \in P_{uv}} w^-_{e'} |((w^-_{e'})^{-1} A^T y)_{(u',v')}|}{w^-_e}\\
 &\leq \frac{\sum_{\{u',v'\} =e' \in P_{uv}} w^-_{e'}}{w^-_e}\\
 &\leq \alpha.
\end{align*}
Thus, $ (W_-^{-1} A^T (\alpha^{-1}y))_{\max} \leq 1$, i.e., $\alpha^{-1}y$ is feasible on $G$, completing the proof.
\end{proof}
As mentioned before, we need to restrict our updates to maintain a certain constraint. As we will see in Section~\ref{sec:potential}, this constraint is that update steps must be orthogonal to $b$. Thus, instead of \eqref{formula:symmetric}, we will consider the program
\begin{equation}\label{formula:symmetric_constrained}
	\min \{ \ones^TW_-x : Ax + zb = \bt, x\ge 0 \}
	= \max\{\bt^Ty: (W_-^{-1} A^Ty)_{\max} \le 1 \wedge b^Ty=0\}
\end{equation}
where $ z $ is a one-dimensional variable.
The additional constraint of $b^Ty=0$ in the dual enforces said orthogonality. It is reflected in the primal program by relaxing the equality constraints to $Ax = \bt - zb$, i.e., shifting the demands by an arbitrary multiple of $b$. Note that feasible primal solutions of \eqref{formula:symmetric_constrained} on a spanner are still feasible on $G$ (after padding), and scaling dual solutions does not affect whether $b^Ty=0$ or not. Hence the same arguments as before yield that an $\alpha$-approximate pair can be computed based on an $\alpha$-spanner of $G$.
\begin{corollary}\label{coro:spanner}
Given an $\alpha$-spanner $S$ of $ G^- = (V, E, w^-) $ and any $0\neq \bt\in \RR^n$ with $\bt^T\ones = 0$, an $\alpha$-approximate primal-dual pair of solutions to~\eqref{formula:symmetric_constrained} with demands $\bt$ can be computed (without further knowledge of $G$).
\end{corollary}

\subsection{Potential Function for the Gradient Descent}\label{sec:potential}

In what follows, we consider $G$ and $b$ to be fixed, and denote by $y^*$ an optimal solution of the dual program, i.e., $b^Ty^*=\max\{b^Ty:(W^{-1} A^Ty)_{\max} \le 1\}$.
We relate the dual program to another linear program that normalizes the objective to $1$ and seeks to minimize $(W^{-1} A^Ty)_{\max}$ instead:
\begin{equation}\label{formula:USP2_short_asymmetric}
	\min\{(W^{-1} A^T\pi)_{\max}:b^T\pi=1\}.
\end{equation}
Let us denote by $\pi^*$ an optimal solution to this problem. There is a straightforward correspondence between feasible solutions of~\eqref{formula:asymmetric} and feasible solutions of~\eqref{formula:USP2_short_asymmetric}.
\begin{lemma}\label{lemma:trafo}
    Consider the map $\psi$ on vertex potentials $\pi$ with $(W^{-1}A^T\pi)_{\max}> 0$ defined by $\psi(\pi):= \frac{\pi}{(W^{-1}A^T\pi)_{\max}}$ and the map $\chi$ on vertex potentials $y$ with $b^Ty>0$ defined by $\chi(y):= \frac{y}{b^Ty}$.
    \begin{enumerate}
        \item If $\pi$ is a feasible solution of~\eqref{formula:USP2_short_asymmetric} with $(W^{-1}A^T\pi)_{\max}> 0$, then $\psi(\pi)$ defines a feasible solution of the dual in~\eqref{formula:asymmetric}.
		If $y$ is a feasible dual solution of~\eqref{formula:asymmetric} with $b^Ty> 0$, then $\chi(y)$ defines a feasible solution of~\eqref{formula:USP2_short_asymmetric}.
        \item The map $\psi(\cdot)$ preserves the approximation ratio.
        Namely, for any $\gamma\ge 1$, if $\pi$ is a solution of~\eqref{formula:USP2_short_asymmetric} within factor $\gamma$ of the optimum, i.e., $(W^{-1}A^T\pi)_{\max}\le \gamma \cdot (W^{-1}A^T\pi^*)_{\max}$,
		then $\psi(\pi)$ is feasible for~\eqref{formula:asymmetric} and within factor $\gamma$ of the optimum, i.e., $b^T\psi(\pi)\ge \gamma^{-1}b^Ty^*$.
    \end{enumerate}
\end{lemma}
\begin{proof}
\begin{enumerate}
   \item Let $\pi$ be feasible for~\eqref{formula:USP2_short_asymmetric}, i.e., $b^T\pi=1$, with $(W^{-1}A^T\pi)_{\max}> 0$. Then $(W^{-1}A^T\psi(\pi))_{\max}= (W^{-1} A^T\frac{\pi}{(W^{-1}A^T\pi)_{\max}})_{\max} =  1$ and thus $\psi(\pi)$ is feasible for the dual in~\eqref{formula:asymmetric}.	Now, let $y$ be feasible for the dual in~\eqref{formula:asymmetric}, i.e., $(W^{-1} A^Ty)_{\max}\le 1$, and let $y$ satisfy $b^Ty>0$. Then $b^T\chi(y)=b^T\frac{y}{b^Ty} = 1$ and thus $\chi(y)$ is feasible for~\eqref{formula:USP2_short_asymmetric}.
    \item Recall that $b\neq 0$ and $b^T\ones = 0$. As $\frac{b}{(W^{-1}A^Tb)_{\max}}$ is feasible for \eqref{formula:asymmetric} and $b^Tb>0$, we have that $b^Ty^*> 0$. Thus, $y^*\neq r\cdot \ones$ for all $r\in \RR$, i.e., there are $u,v\in V$ so that $y^*_u\neq y^*_v$. As $G$ is connected, this entails that $(W^{-1}A^Ty^*)_{\max}>0$. Hence, $\chi(y^*)$ is feasible for \eqref{formula:USP2_short_asymmetric} and has positive objective, i.e., $(W^{-1}A^T\chi(y^*))_{\max}=\frac{(W^{-1}A^Ty^*)_{\max}}{b^Ty^*}>0$; in particular, we have that $(W^{-1}A^T\pi^*)_{\max}>0$. Accordingly, if $\pi$ is a $\gamma$-approximation of \eqref{formula:USP2_short_asymmetric}, $(W^{-1}A^T\pi)_{\max}>0$ and
        \begin{align*}
			b^T\psi(\pi) &= \frac{b^T\pi}{(W^{-1}A^T\pi)_{\max}}\\
			&\ge \frac{1}{(W^{-1}A^T\pi)_{\max}}\\
			&\ge \frac{1}{\gamma (W^{-1}A^T\pi^*)_{\max}}\\
			&\ge \frac{1}{\gamma (W^{-1} A^T\chi(y^*))_{\max}}\\
			&= \frac{b^Ty^*}{\gamma (W^{-1} A^Ty^*)_{\max}}\\
			&\ge \frac{b^Ty^*}{\gamma}. \qedhere
		\end{align*}
\end{enumerate}
\end{proof}
In other words, it is sufficient to determine a $(1+\varepsilon)$-approximation to~\eqref{formula:USP2_short_asymmetric} in order to obtain a $(1+\varepsilon)$-approximation to~\eqref{formula:asymmetric}.

We have now translated the dual maximization problem of~\eqref{formula:asymmetric}, which has inequality constraints, into a minimization problem with a single equality constraint $b^T\pi = 1$. This would enable us to employ projected gradient descent in the $b^T\pi = 1$ plane, if it were not for the fact that the objective $(W^{-1} A^T\pi)_{\max}$ is not differentiable. To overcome this issue, we use the standard approach of ``smoothing'' the gradient by approximating the objective by a differentiable function. For the maximum value of a vector, a suitable candidate is given by the log-sum-exponent (or softmax) function. For vectors $z\in\RR^d$, it is defined as
\begin{equation*}
	\softmax[\beta]{z}:= \frac{1}{\beta} \ln\left(\sum_{i\in [d]}e^{\beta z_i}\right),
\end{equation*}
where the parameter $\beta > 0$ determines the trade-off between (a) accuracy of approximation and (b) ``smoothness.'' To clarify what we mean by (a), observe that
\begin{equation}\label{formula:lse_approx}
(z)_{\max}\leq \softmax[\beta]{z}\leq \frac{1}{\beta} \ln\left(\sum_{i\in [d]}e^{\beta (z)_{\max}}\right)= \frac{1}{\beta} \ln\left(de^{\beta (z)_{\max}}\right)=(z)_{\max}+\frac{\ln d}{\beta},
\end{equation}
because both $\ln (\cdot) $ and $e^{(\cdot)}$ are increasing functions. To precise what we mean by (b), denote for $z\in \RR^d$ by $\|z\|_1=\sum_{i=1}^d |z_i|$ its $1$-norm and by $\|z\|_{\infty}=\max\{|z_i|:i\in [d]\}$ its $\infty$-norm. Then the $1$-norm of the gradient $\nabla \softmax[\beta]{\cdot}$ is $\beta$-Lipschitz continuous w.r.t.\ the $\infty$-norm (see, e.g., \cite{Sherman13}), i.e.,
\begin{equation}\label{formula:lipschitz}
\forall z,z'\in \RR^d\colon \|\nabla\softmax[\beta]{z}-\nabla\softmax[\beta]{z'}\|_1\leq \beta \|z-z'\|_{\infty}.
\end{equation}

Recall that our goal is to find $\pi\in \RR^n$ that is a $(1+\varepsilon)$-approximate feasible solution to~\eqref{formula:USP2_short_asymmetric}. Accordingly, we define the \emph{potential function}
\begin{equation*}
	\potfun_\beta(\pi):=\softmax[\beta]{W^{-1} A^T\pi}.
\end{equation*}
Note that $\potfun_\beta(\cdot)$ is convex for any $\beta$, as it is constructed by composing $\softmax[\beta]{\cdot}$, which is convex for any $\beta$, with linear functions. We will vary $\beta$ over the course of the algorithm to balance the requirement of a sufficiently accurate approximation of $(W^{-1}A^T\pi)_{\max}$ with the need for small $\beta$, i.e., a smooth gradient and thus fast progress.

Concerning the approximation guarantee, our target of a $(1+\varepsilon)$-approximation entails that the potential must be a more accurate approximation to the objective. Accordingly, we will ensure that
\begin{equation}\label{formula:potfun_approx}
\potfun_\beta(\pi)\leq \left(1+\frac{\varepsilon}{4}\right)(W^{-1} A^T\pi)_{\max}.
\end{equation}
By \eqref{formula:lse_approx}, $4\ln(2m)\leq \varepsilon \beta (W^{-1}A^T\pi)_{\max}$ implies \eqref{formula:potfun_approx}, which is an invariant the algorithm will maintain.

To understand the effect of $\beta$ on the progress of the gradient descent, we examine the gradient of the potential function. As multiplications with $W^{-1}$ and $A^T$ are linear functions,
\begin{equation}\label{formula:grad_potfun}
\nabla\potfun_\beta(\pi) = AW^{-1}\nabla\softmax[\beta]{W^{-1}A^T\pi}.
\end{equation}
Recall that $W_-=\diag(w^-,w^-)$. From \eqref{formula:lipschitz} and H\"older's inequality,\footnote{It holds that $|z^Tz'|\leq \|z\|_p \cdot \|z'\|_q$ for any $p,q\geq 1$ satisfying $p^{-1}+q^{-1}=1$ (where $\infty^{-1}=0$).} we can derive a key property of $\nabla\potfun_\beta(\cdot)$ that we will use for bounding the change of the potential for a given update step $h\in \RR^n$ that is due to the change of the gradient, namely
\begin{align}\label{formula:grad_potfun_smooth}
\forall \pi,h\in \RR^n\colon &\left|\nabla\potfun_\beta(\pi)^Th-\nabla\potfun_\beta(\pi-h)^Th\right|\nonumber\\
\stackrel{\eqref{formula:grad_potfun}}{=}&\left|AW^{-1}\nabla\softmax[\beta]{W^{-1}A^T\pi}^T h-AW^{-1}\nabla\softmax[\beta]{W^{-1}A^T(\pi-h)}^Th\right|\nonumber\\
=&\left|\left(\nabla\softmax[\beta]{W^{-1}A^T\pi}-\nabla\softmax[\beta]{W^{-1}A^T(\pi-h)} \right)^TW^{-1}A^T h\right|\nonumber\\
\stackrel{\mbox{\scriptsize{H\"older}}}{\leq}&\left\|\nabla\softmax[\beta]{W^{-1}A^T\pi}-\nabla\softmax[\beta]{W^{-1}A^T(\pi-h)}\right\|_1\left\|W^{-1}A^Th\right\|_{\infty}\nonumber\\
\stackrel{\eqref{formula:lipschitz}}{\leq}&\;\beta\|W^{-1}A^Th\|_{\infty}^2\nonumber\\
= &\hspace*{.1cm}\beta \|W_-^{-1}A^Th\|_{\infty}^2\nonumber\\
= &\hspace*{.1cm}\beta (W_-^{-1}A^Th)_{\max}^2\,,
\end{align}
where the last two steps use that $(A^Th)_{uv}=-(A^Th)_{vu}$ for all $u,v\in V$ and $w^+_e\geq w^-_e$ for all $e\in \bar{E}$. Intuitively, we decomposed the gradient into the $\beta$-Lipschitz continuous part given by the $\softmax[\beta]{\cdot}$ function and the derivative $AW^{-1}$ of the ``inner'' function. This means that the ``length'' of an update step $h$ will be measured in terms of $(W_-^{-1}A^Th)_{\max}$.
Using this bound, for a given step direction $h$ we can now determine the step size that maximizes the progress in direction of $h$ as a function of $\nabla\potfun_\beta(\pi)$ (under the worst-case assumption that \eqref{formula:grad_potfun_smooth} is tight).
\begin{lemma}[Additive Decrement of $\potfun_{\beta}$]\label{lemma:decrease}
		Suppose $\pi,h \in\RR^n$ satisfy $\nabla\potfun_\beta(\pi)^Th>0$ and $(W_-^{-1}A^Th)_{\max}\leq 1$.
	Then, for $\delta := \nabla\potfun_\beta(\pi)^T h$ it holds that
	\[
		\potfun_\beta\Big(\pi-\frac{\delta h}{2\beta}\Big) \le \potfun_\beta(\pi)-\frac{\delta^2}{4\beta}.
	\]
\end{lemma}
\begin{proof}
	Let us denote $\h:=\tfrac{\delta h}{2\beta}$. Recall that $\potfun_\beta(\cdot)$ is convex and thus $ \potfun_\beta(\pi) \geq \potfun_\beta(\pi-\h) + \nabla \potfun_\beta\big(\pi-\h\big)^T\h $. This gives
	\begin{align*}
		\potfun_\beta(\pi-\h) - \potfun_\beta(\pi)&\leq
		-\nabla \potfun_\beta\big(\pi-\h\big)^T\h
		+ \nabla \potfun_\beta(\pi)^T\h
		- \nabla \potfun_\beta(\pi)^T\h\\
		&\stackrel{\eqref{formula:grad_potfun_smooth}}{\leq}
		\beta (W_-^{-1}A^T\h)_{\max}^2- \nabla\potfun_\beta(\pi)^T \h\\
		&= \frac{\delta^2 (W_-^{-1}A^Th)_{\max}^2}{4\beta} - \frac{\delta^2}{2\beta}\\
		&\leq \frac{\delta^2}{4\beta}-\frac{\delta^2}{2\beta} =-\frac{\delta^2}{4\beta}. \qedhere
	\end{align*}
\end{proof}

As this lemma suggests, we will try to ensure large progress by making $ \delta = \nabla\potfun_\beta(\pi)^T h $ as large as possible and $ \beta $ as small as possible.
Now observe that, for a fixed~$ \beta $, maximizing $ \nabla\potfun_\beta(\pi)^T h $ under the constraint $(W_-^{-1}A^Th)_{\max}\leq 1$ is another instance of the transshipment problem with demand vector $ \nabla\potfun_\beta(\pi) $.
Our algorithm will determine its step direction~$ h $ by computing an $ \alpha $-approximation to this transshipment instance.
Regarding the incentive of minimizing $ \beta $, note that progress with respect to $\potfun_\beta(\cdot)$ is meaningless if it does not provide a sufficiently accurate approximation of the true objective $(W_-^{-1}A^T(\cdot))_{\max}$, so we will increase $\beta$ only when it becomes necessary to ensure \eqref{formula:potfun_approx}. Bounding $\beta$ from above turns the additive progress guarantee into a relative one.
\begin{corollary}[Multiplicative Decrement of $\potfun_{\beta}$]\label{coro:decrease}
Suppose $\pi,h \in\RR^n$ satisfy $\nabla\potfun_\beta(\pi)^Th>0$, $(W_-^{-1}A^Th)_{\max}\leq 1$, and $\varepsilon\beta\potfun_{\beta}(\pi)\leq 10 \ln(2m)$. Then, for $\delta := \nabla\potfun_\beta(\pi)^T h$ it holds that
	\[
		\potfun_\beta\Big(\pi-\frac{\delta h}{2\beta}\Big) \le \left(1-\frac{\varepsilon\delta^2}{40\ln(2m)}\right)\potfun_\beta(\pi).
	\]
\end{corollary}

\subsection{Generic Algorithm}\label{sec:algorithm}
We now present the pseudocode of the generic gradient descent algorithm for finding a $(1+\varepsilon)$-approximate solution to \eqref{formula:asymmetric}. As mentioned before, the algorithm actually computes a $(1+\varepsilon)$-approximate solution $\pi$ to \eqref{formula:USP2_short_asymmetric} and then returns $\psi(\pi)$, cf.~Lemma~\ref{lemma:trafo}. The algorithm is generic, as the performed computations need to be implemented in model-specific ways. This is discussed in Section~\ref{sec:apps}.

\begin{algorithm2e}[ht!]
	compute $(\alpha\lambda)$-approximation $\pi$ to \eqref{formula:USP2_short_asymmetric}\tcp*{use oracle, Obs.~\ref{obs:symmetric}, and Lemma~\ref{lemma:trafo}} \nllabel{line:initial oracle call}
	set $\varepsilon':=1$\;
	\Repeat{$\varepsilon'\leq \varepsilon$}{
	  set $\varepsilon'\leftarrow \frac{\varepsilon'}{2}$\;
	  set $\beta:=\tfrac{8\ln(2m)}{\varepsilon' (W^{-1}A^T\pi)_{\max}}$\nllabel{line:set_beta}\;
	  \Repeat{$\delta \le \tfrac{\varepsilon'}{6\alpha\lambda}$\tcp*[f]{current $\pi$ is a $(1+\varepsilon')$-approximate solution}}{
		compute gradient $\nabla\potfun_\beta(\pi)$\;\nllabel{line:compute gradient}
		compute $\alpha$-approximate primal-dual pair $ ((f, z), h) $ for $ \min \{ \ones^TW_-\ft : A\ft + zb = \nabla\potfun_{\beta}(\pi), \ft\ge 0 \} = \max \left\{\nabla\potfun_\beta(\pi)^T \h : (W_-^{-1}A^T\h)_{\max} \le 1 \wedge b^T\h = 0\right\}$\;\nllabel{line:oracle call in iteration}
        \tcp*[l]{use oracle with weights $w^-$ and demands $\bt = \nabla\potfun_\beta(\pi)$ (Cor.~\ref{coro:spanner})}
		set $\delta:= \nabla\potfun_\beta(\pi)^T h $\;
		\lIf{$\delta> \tfrac{\varepsilon'}{6\alpha\lambda}$}{
			$\pi\leftarrow \pi - \tfrac{\delta h}{2\beta}$\nllabel{line:update}
		}
        \lIf{$\beta<\tfrac{4\ln(2m)}{\varepsilon' (W^{-1}A^T\pi)_{\max}}$}{$\beta:=\tfrac{8\ln(2m)}{\varepsilon' (W^{-1}A^T\pi)_{\max}}$\nllabel{line:double_beta}\\\tcp*[f]{happens only if $\varepsilon'=\tfrac{1}{2}$}}
	  }
	}
	$ \xt := W^{-1}\nabla \softmax[\beta]{W^{-1}A^T\pi} $\;\nllabel{line:primal solution xt}
	$ \binom{f^+}{f^-} := f $ with $f^+,f^-\in \RR^m$\;\nllabel{line:primal decompose f}
	$ f' :=\binom{f^-}{f^+}$\;
	$ x:=\frac{\xt + f'}{z} $ \tcp*{Observation~\ref{obs:primal}}\nllabel{line:primal solution x}
	$ y:=\psi(\pi) $\tcp*{Lemma~\ref{lemma:trafo}}
	\Return $ (x, y) $\;
	\caption{\nablaalg$(G, b, \varepsilon)$}
	\label{alg:gradient}
\end{algorithm2e}

The code is given in Algorithm~\ref{alg:gradient}. The algorithm first computes a starting solution. This can be done by calling the oracle on the symmetrized problem \eqref{formula:symmetric}, which by Observation~\ref{obs:symmetric} yields an $(\alpha\lambda)$-approximation to \eqref{formula:asymmetric}, and then rescaling according to Lemma~\ref{lemma:trafo}. As trying to enforce a too precise approximation initially slows down progress, the algorithm initializes $\varepsilon'$ as a constant (the choice of $1$ is, neglecting technicalities, arbitrary). Then it starts an outer loop decreasing $\varepsilon'$ exponentially. It sets~$\beta$ to a suitable value and then starts the inner loop, which performs gradient descent steps until the step size becomes too small to ensure fast progress, i.e., until $\delta\leq \frac{\varepsilon'}{6\alpha \lambda^2}$. As in the first iteration of the outer loop, the potential may decrease rapidly (we go from an $(\alpha \lambda)$-approximation to a $\tfrac{1}{2}$-approximation), Line~\ref{line:double_beta} adjusts $\beta$ if necessary. As we will show, this implies that at termination of the inner loop, the current solution is a $(1+\varepsilon')$-approximation.

Each gradient descent step consists of the following steps:
\begin{enumerate}
  \item Compute $\nabla \potfun_{\beta}(\pi)$.
  \item Compute an $\alpha$-approximate solution $h$ to
  \begin{equation*}
  \max \left\{\nabla \potfun_{\beta}(\pi)^T \h : (W_-^{-1}A^T\h)_{\max} \le 1\wedge b^T\h = 0\right\},
  \end{equation*}
  i.e., of \eqref{formula:symmetric_constrained} with demands $\nabla \potfun_{\beta}(\pi)$.\footnote{As $ \nabla \potfun_{\beta}(\pi)^T \ones = 0 $, this instance is feasible.} This is done by calling the oracle (for the constrained problem). This optimization problem maximizes the rate of progress ``in units of $(W_-^{-1}A^T\h)_{\max}$,'' which in turn determines the step size, see Corollary~\ref{coro:decrease}.
  \item Determine the step size in accordance with Corollary~\ref{coro:decrease} and check whether $\delta$ is small enough to prove that the current $\pi$ is a $(1+\varepsilon')$-approximation. If yes, the inner loop terminates; if not, the step is executed.
  \item Increase $\beta$ if it becomes too small. (This will only happen when $ \epsilon' = \tfrac{1}{2} $).
\end{enumerate}
We note that the dual oracle problems occurring for different gradients of the potential function in the second step differ only in terms of the objective function. The constraints however are the same in all iterations. This allows us to use the same spanner in all of these iterations as, in particular, the edge weights involved are always the backward weights $w^-$ of the input graph.
\paragraph*{Returning a primal solution}
As specified so far, the algorithm relies on a dual oracle only and also returns only a dual solution. For a primal solution, we require the oracle to provide an $\alpha$-approximate primal-dual pair to the constrained problem~\eqref{formula:symmetric_constrained} (with demands $\nabla \potfun_{\beta}(\pi)$).\footnote{Actually, the primal solution is exclusively required on termination to generate a primal solution to \eqref{formula:asymmetric}.}
Consider the final iteration of Algorithm~\ref{alg:gradient}.
Denote by~$ \pi $ the final normalized node potentials, set $ \xt = W^{-1}\nabla \softmax[\beta]{W^{-1}A^T\pi} $ (for the final value of $ \beta $) and $ y = \psi(\pi) $.
We will show how to compute $ x $ such that $ (x, y) $ is a $ (1 + \epsilon) $-approximate primal-dual pair.

\begin{observation}\label{obs:xt}
The vector $\xt=W^{-1}\nabla \softmax[\beta]{W^{-1}A^T\pi}$ satisfies \normalfont{(a)} $\xt\ge 0$, \normalfont{(b)} $ A\xt = \nabla\potfun_{\beta}(\pi) $, and \normalfont{(c)} $ \ones^TW\xt= 1 $.
\end{observation}
\begin{proof}
By~\eqref{formula:grad_potfun}, $A\xt = \nabla\potfun_{\beta}(\pi)$ and thus (b) holds. For the other two properties, we set $z=W^{-1}A^T\pi$ and compute
\begin{equation*}
    \forall i\in \{1,\ldots,2m\}\colon \xt_i = \frac{1}{w_i} \nabla \softmax[\beta]{z}_i=\frac{1}{\beta w_i}\frac{\beta e^{\beta z_i}}{\sum_{j=1}^{2m} e^{\beta z_j}}>0,
\end{equation*}
hence (a) holds, and
\begin{equation*}
\ones^TW\xt = \ones^T \nabla \softmax[\beta]{z}=\frac{1}{\beta}\sum_{i=1}^{2m} \frac{\beta e^{\beta z_i}}{\sum_{j=1}^{2m} e^{\beta z_j}}=1,
\end{equation*}
showing (c).
\end{proof}
Conveniently, this $\xt$ is a byproduct of evaluating $\nabla \potfun_{\beta}(\pi)$ in the final iteration of the algorithm, by storing the intermediate result before the final multiplication with $A$, cf.~\eqref{formula:grad_potfun}.

Recall that in its final iteration, the algorithm computes an $\alpha$-approximate primal-dual pair $ ((f, z), h) $ for \eqref{formula:symmetric_constrained} with demands $\bt = \nabla\potfun_{\beta}(\pi)$, i.e., the program
\begin{align*}
	&\min \{ \ones^TW_-\ft : A\ft + zb = \nabla\potfun_{\beta}(\pi), \ft\ge 0 \}\\
	=\; & \max\{\nabla\potfun_{\beta}(\pi)^T\h: (W_-^{-1} A^T\h)_{\max} \le 1 \wedge b^Ty=0\}
\end{align*}
(where we use $\ft$ and $\h$ to denote the variables for the sake of distinction from \eqref{formula:asymmetric}). We can use $\xt$ as defined above to cancel out the demands in the primal program and then rescale according to $z$.
\begin{observation}\label{obs:primal}
For any given $\beta\in \RR_{>0}$, $\pi\in \RR^n$ with $b^T\pi=1$, primal solution $(f,z)\in \RR^{2m}\times \RR$ to~\eqref{formula:symmetric_constrained} with $\bt=\nabla\potfun_{\beta}(\pi)$ satisfying $z>0$, and $\xt\in \RR^{2m}$ satisfying~$ A\xt = \nabla\potfun_{\beta}(\pi) $, write $f=\binom{f^+}{f^-}$ (with $f^+,f^-\in \RR^m$). Then a primal solution to \eqref{formula:asymmetric} is given by
\begin{equation}\label{formula:primal}
x:=\frac{\xt + \binom{f^-}{f^+}}{z}\in \RR^{2m}.
\end{equation}
\end{observation}
\begin{proof}
As $\xt$, $f$, and $z$ are all positive, so is $x$. $Ax=b$ follows from observing that $A\binom{f^-}{f^+}=-Af=-\nabla\potfun_{\beta}(\pi)+zb$, as for each (undirected) edge we have a forward and a backward arc.
\end{proof}
The intuition regarding why this should be a ``good'' primal solution to \eqref{formula:asymmetric} is as follows. Any $\pi'$ minimizing $\potfun_{\beta}(\cdot)$ must satisfy that $\nabla\potfun_{\beta}(\pi')^T h=0$ for any $h$ with $b^T h = 0$. Thus, $\nabla\potfun_{\beta}(\pi')$ is proportional to $b$. For a near-optimal solution $\pi$ we have that $\nabla\potfun_{\beta}(\pi)$ is \emph{almost} parallel to $b$, as the gradient restricted to the $b^Th = 0$ plane is small. Accordingly, the ``cost'' of the correction, given by $\ones^TWf$ (which relates to $\ones^TW_-f$), is small. The role of the scaling factor $z$ is to undo the normalization of \eqref{formula:asymmetric} we performed by transitioning to problem \eqref{formula:USP2_short_asymmetric}.

We remark that it also possible to drop the constraint $b^T h = 0$ in the update problem at the expense of an increase in the approximation ratio. To this end, we consider the projector $P := I - \pi b^T$, which yields $b^T P h = 0$ for any $h \in \RR^n$. A primal-dual pair $f',h'$ of
\[
\max \{ \bt^T P h : (W_-^{-1} A^T h)_{\max} \le 1 \} = \min \{ \ones^T W_- f : A f = P^T \bt,\, f \ge 0 \}
\]
with $\ones^T W_- f' \le \alpha' \cdot \bt^T P h'$ yields a primal-dual pair $(f,z), h$ for the more restricted update problem with $f = f'$, $z = \bt^T \pi$, and $h = P h' / (W_-^{-1} A^T P h')_{\max}$. Moreover,
\[
 \ones^T W_- f
 = \ones^T W_- f'
 \le \alpha' \cdot \bt^T P h'
 = \alpha' \cdot (W_-^{-1} A^T P h')_{\max} \cdot \bt^T h.
\]
Furthermore,
\begin{align*}
 (W_-^{-1} A^T P h)_{\max} &\le 1 + |b^T h| \cdot (W_-^{-1} A^T \pi)_{\max} \le 1 + b^T y^* \cdot \lambda \cdot (W^{-1} A^T \pi)_{\max} \\ &= 1 + \lambda \frac{(W^{-1} A^T \pi)_{\max}}{(W^{-1} A^T \pi^*)_{\max}} \le 1 + \lambda^2 \alpha'
 \end{align*}
 since we start with an $\alpha'\lambda$-approximation.\footnote{Note that $(W^{-1} A^T \pi)_{\max}$ is always within a constant factor of $\potfun_\beta(\pi)$, which decreases in every iteration.} Thus, we obtain an $\alpha$-approximation of the restricted problem with $\alpha \in O(\alpha'^2\lambda^2)$. However, in the models of computation considered in this paper, solving the restricted problem is not harder than the unrestricted one, so we do not need to pay the additional computational cost resulting from the weaker approximationa guarantee.

\paragraph*{Proof of correctness}
We show the approximation guarantee of the algorithm by arguing that its return value $ (x, y) $ is a $(1+\varepsilon)$-approximate pair, i.e., $\ones^TWx\leq (1+\varepsilon)b^Ty$. By weak duality this in particular implies that both $x$ and $y$ are ${(1+\varepsilon)}$-approximate solutions. We stress, however, that Algorithm~\ref{alg:gradient} would not have to compute $ x $ to return the $(1+\varepsilon)$-approximate dual solution $ y $; we could exploit the mere existence of~$ x $ for proving the approximation guarantee of $ y $. To obtain $ y $ it would thus be sufficient for the oracle to return only an $ \alpha $-approximate dual solution $ h $ for~\eqref{formula:symmetric_constrained}. If the oracle additonally returns an $ \alpha $-approximate dual solution $ (f, z) $, then $ x $ can be determined explicitly. Note that, by Corollary~\ref{coro:spanner}, the oracle we use throughout this work has this capability.

We first prove a helper statement showing that the algorithm maintains suitable values of $\beta$, and how adjusting $\beta$ affects $\potfun_{\beta}(\pi)$.
\begin{lemma}\label{lemma:helper}
Except immediately after updates of $\varepsilon'$ or $\pi$ (i.e., before adjusting~$\beta$ in the subsequent lines), Algorithm \emph{\nablaalg} maintains the invariant that
\begin{equation*}
4\ln(2m)\leq \varepsilon'\beta (W^{-1}A^T\pi)_{\max}\leq \varepsilon'\beta \potfun_{\beta}(\pi)\leq 10\ln(2m).
\end{equation*}
In particular, $\potfun_{\beta}(\pi)\leq (1+\tfrac{\varepsilon'}{4})(W^{-1}A^T\pi)_{\max}$.
\end{lemma}
\begin{proof}
First, observe that $4\ln(2m)\leq \varepsilon'\beta (W^{-1}A^T\pi)_{\max}$ follows immediately from Lines~\ref{line:set_beta} and~\ref{line:double_beta}. By~\eqref{formula:lse_approx}, this implies \eqref{formula:potfun_approx}, i.e., the last statement of the lemma. As by~\eqref{formula:lse_approx} $(W^{-1}A^T\pi)_{\max} \leq \softmax[\beta]{W^{-1} A^T\pi} = \potfun_{\beta}(\pi)$ for any $\beta$, it remains to show that $\varepsilon'\beta \potfun_{\beta}(\pi)\leq 10\ln(2m)$. Note that $\varepsilon'\beta (W^{-1}A^T\pi)_{\max}\leq 8\ln(2m)$ implies that
\begin{equation*}
\varepsilon'\beta \potfun_{\beta}(\pi)\leq \varepsilon'\beta\left(1+\frac{\varepsilon'}{4}\right)(W^{-1}A^T\pi)_{\max}\stackrel{\varepsilon'\leq 1}{\leq}10\ln(2m).
\end{equation*}
Thus, the statement always holds after executing Line~\ref{line:set_beta} or adjusting $\beta$ according to Line~\ref{line:double_beta}. This leaves only the possibility that updates of $\pi$ cause a violation of the inequality $ \varepsilon'\beta \potfun_{\beta}(\pi)\leq 10\ln(2m) $ while $\beta$ is not adjusted. However, Lemma~\ref{lemma:decrease} shows that updates of $\pi$ may only decrease the potential, completing the proof.
\end{proof}

\begin{lemma}[Correctness of Algorithm~\ref{alg:gradient}]\label{lemma:correctness}
	Let $0<\varepsilon\leq 1$, and let $ (x, y) $ be the return value of $ \nablaalg (G, b, \epsilon) $.
	Then $(x, y)$ is a $(1+\varepsilon)$-approximate pair for \eqref{formula:asymmetric}, i.e., it holds that
	$Ax=b$, $x\ge 0$, $\qmaxstretch{y}\le 1$, and $\pobjvalue{x}\le (1+\varepsilon)b^Ty$.
\end{lemma}
\begin{proof}
Let $ \beta $, $ \pi $, $ f $, $ z $, and $ h $ denote the values of the respective variables upon termination of the algorithm.
Remember that $ ((f, z), h) $ is an $\alpha$-approximate pair for~\eqref{formula:symmetric_constrained} with $\bt=\nabla\potfun_{\beta}(\pi)$.
Due to termination of the algorithm, we have that $\delta \leq \frac{\varepsilon'}{6\alpha \lambda}$ and $\varepsilon'\leq\varepsilon$, yielding that
\begin{equation}\label{formula:progress_termination}
\nabla\potfun_{\beta}(\pi)^T h = \delta \leq \frac{\varepsilon'}{6\alpha \lambda}\leq \frac{\varepsilon}{6\alpha}.
\end{equation}
Recall that $\|W^{-1}_-A^T\pi\|_{\infty}=(W^{-1}_-A^T\pi)_{\max}$ and note that, as $f\geq 0$, $\|W_-f\|_1=\ones^TW_-f$.
By Lemma~\ref{lemma:helper} we have $ 4\ln(2m)\leq \varepsilon'\beta (W^{-1}A^T\pi)_{\max}\leq \varepsilon'\beta \leq \varepsilon\beta $.
As $\potfun_\beta(\cdot)$ is convex, its gradient restricted to a line is non-decreasing. Thus, the first-order Maclaurin approximation gives
\begin{equation}\label{formula:potfun_grad_lower}
\pi^T\nabla\potfun_{\beta}(\pi)\geq \potfun_{\beta}(\pi)-\potfun_{\beta}(0)=\potfun_{\beta}(\pi)-\frac{\ln(2m)}{\beta} \geq \left(1-\frac{\varepsilon}{4}\right)\potfun_{\beta}(\pi).
\end{equation}

By multiplying the set of equalities $Af+zb=\nabla \potfun_{\beta}(\pi)$ with $\pi^T$ from the left, we see that
\begin{align}
z&=\pi^T b z\nonumber\\
& = \pi^T \nabla \potfun_{\beta}(\pi)-\pi^TAf\nonumber\\
& \stackrel{\eqref{formula:potfun_grad_lower}}{\geq} \left(1-\frac{\varepsilon}{4}\right)\potfun_\beta(\pi) - \pi^TAW_-^{-1}W_-f\nonumber\\
& \stackrel{\eqref{formula:lse_approx}}{\geq} \left(1-\frac{\varepsilon}{4}\right)(W^{-1}_-A^T\pi)_{\max} - (W^{-1}_-A^T\pi)^TW_-f\nonumber\\
& \stackrel{\mbox{\scriptsize{H\"older}}}{\geq} \left(1-\frac{\varepsilon}{4}\right)(W^{-1}_-A^T\pi)_{\max} - \|W^{-1}_-A^T\pi\|_{\infty}\|W_-f\|_1\nonumber\\
& = \left(1-\frac{\varepsilon}{4}-\ones^TW_-f\right)(W^{-1}_-A^T\pi)_{\max}\nonumber\\
& \stackrel{\mbox{\scriptsize{$\alpha$-approx.\;pair}}}{\geq} \left(1-\frac{\varepsilon}{4}-\alpha \nabla\potfun_{\beta}(\pi)^T h\right)(W^{-1}_-A^T\pi)_{\max}\nonumber\\
& \stackrel{\eqref{formula:progress_termination}}{\geq} \left(1-\frac{5\varepsilon}{12}\right)(W^{-1}_-A^T\pi)_{\max}.\label{formula:z}
\end{align}
In particular, $z>0$ and, by Observation~\ref{obs:primal} and Lemma~\ref{lemma:trafo}, $x$ and $y = \psi(\pi)$ are indeed feasible primal and dual solutions of \eqref{formula:asymmetric}, respectively.

It remains to show that $\pobjvalue{x}\le (1+\varepsilon)b^Ty$. As $\ones^T W\xt = 1$ by Observation~\ref{obs:xt},
\begin{align*}
\ones^TWx &= \frac{\ones^TW\xt + \ones^T W\binom{f^-}{f^+}}{z}\\
&\leq \frac{1+\lambda\ones^TW_-\binom{f^-}{f^+}}{z}\\
&=\frac{1+\lambda\ones^TW_-f}{z}\\
&\stackrel{\mbox{\scriptsize{$\alpha$-approx.\;pair}}}{\leq} \frac{1+\lambda\alpha \nabla\potfun_{\beta}(\pi)^Th}{z}\\
&\stackrel{\eqref{formula:progress_termination}}{\leq} \frac{1+\tfrac{\varepsilon}{6}}{z}\\
&\stackrel{\eqref{formula:z}}{\leq}\frac{1+\tfrac{\varepsilon}{6}}{\left(1-\tfrac{5\varepsilon}{12}\right)(W^{-1}_-A^T\pi)_{\max}}\\
&\stackrel{\varepsilon\leq 1}{\leq}\frac{1+\varepsilon}{(W^{-1}_-A^T\pi)_{\max}}\\
&\stackrel{\mbox{\scriptsize{Obs.~\ref{obs:symmetric}}}}{\leq}\frac{1+\varepsilon}{(W^{-1}A^T\pi)_{\max}}\\
&=(1+\varepsilon)b^T\psi(\pi) \\
&=(1+\varepsilon)b^Ty. \qedhere
\end{align*}
\end{proof}
\begin{corollary}\label{coro:correctness}
Suppose $0<\varepsilon\leq 1$. Whenever the inner loop of Algorithm \emph{\nablaalg} terminates, the current $\pi$ is a $(1+\varepsilon')$-approximate dual solution to \eqref{formula:asymmetric}.
\end{corollary}
\begin{proof}
If we ran the algorithm for $\varepsilon=\varepsilon'$, it would perform exactly the same computations up to that point and then terminate. Noting that an optimal primal solution $(f,z)$ for the problem posed to the oracle in the last iteration and, by Observation~\ref{obs:primal}, $\xt=W^{-1}\nabla \softmax[\beta]{W^{-1}A^T\pi}$ satisfies the preconditions of Lemma~\ref{lemma:correctness}, the claim is shown by the lemma.
\end{proof}

\paragraph*{Bounding the number of iterations}
We now examine how many iterations Algorithm~\ref{alg:gradient} requires to terminate. This reduces the task of bounding the overall running time to determining the cost of implementing a single iteration, which depends on the specific model of computation.

\begin{lemma}[Number of iterations of Algorithm~\ref{alg:gradient}]\label{lemma:iterations_asym}
Suppose that $0<\varepsilon\leq 1$. Then Algorithm \emph{\nablaalg} performs $O((\varepsilon^{-2}+\log \alpha+ \log \lambda)\alpha^2\lambda^2 \log n)$ iterations of its inner loop.
\end{lemma}
\begin{proof}
Denote by $\pi^{(i)}$ and $\beta_i$, $0\leq i\leq i_{\max}:=\lceil \log \varepsilon^{-1}\rceil$, the values of $\pi$ and $\beta$ at the beginning of the $(i+1)$-th iteration of the outer loop, where $\pi^{(i_{\max})}$ and $\beta_{i_{\max}}$ denote the values at termination. Refer to the $i$-th iteration as \emph{phase $i$} and denote by $\varepsilon_i=2^{-i}$ the value of $\varepsilon'$ during this loop iteration.

Lemma~\ref{lemma:helper} shows that the preconditions of Corollary~\ref{coro:decrease} are satisfied by each update of $\pi$, as $\delta\leq \tfrac{\varepsilon_i}{6\lambda \alpha}$ implies termination of the inner loop. Thus, if the inner loop does not terminate, the corollary shows that the potential decreases by a factor of at least
\begin{equation*}
\left(1-\frac{\varepsilon_i\delta^2}{40\ln(2m)}\right)\leq \left(1-\frac{\varepsilon_i^3}{1440\alpha^2\lambda^2\ln(2m)}\right)=:q_i.
\end{equation*}
However, when setting $\beta$ in Line~\ref{line:double_beta}, the potential might increase. Recall that, for any $\beta$, it holds that $\potfun_{\beta}(\pi)\geq (W^{-1}A^T\pi)_{\max}$ by~\eqref{formula:lse_approx}. As after Line~\ref{line:double_beta} we have that $\potfun_{\beta}(\pi)\leq (1+\tfrac{\varepsilon_i}{4})(W^{-1}A^T\pi)_{\max}$ again by Lemma~\ref{lemma:helper}, the increase in potential due to an update of $\beta$ according to Line~\ref{line:double_beta} is multiplicatively bounded by $1+\tfrac{\varepsilon_i}{4}$.
Observe that $\beta$ is at least doubled by the statement in Line~\ref{line:double_beta} and that this can happen no more than
\begin{equation*}
\ell_i:=\log\left(\frac{(W^{-1}A^T\pi^{(i-1)})_{\max}}{(W^{-1}A^T\pi^{(i)})_{\max}}\right)\leq \log\left(\frac{(W^{-1}A^T\pi^{(i-1)})_{\max}}{(W^{-1}A^T\pi^*)_{\max}}\right)
\end{equation*}
times in phase $i$. By Corollary~\ref{coro:correctness}, we have for all $i>1$ that
\begin{equation*}
(W^{-1}A^T\pi^{(i-1)})_{\max}\leq (1+\varepsilon_{i-1})(W^{-1}A^T\pi^*)_{\max}<2(W^{-1}A^T\pi^*)_{\max}
\end{equation*}
and thus $\ell_i=0$, and as $\pi^{(0)}$ is an $(\alpha\lambda)$-approximation, we have that $\ell_1\leq \log (\alpha\lambda)$.

Denote by $\beta_i'$ the value of $\beta$ when the inner loop terminates at the end of phase~$i$ and by $k_i$ the number of iterations in this phase. Using the above bounds on $\ell_i$, that $\potfun_{\beta_{i-1}}(\pi^{(i-1)})\leq (1+\tfrac{\varepsilon_i}{4})(W^{-1}A^T\pi^{(i-1)})_{\max}$ by Lemma~\ref{lemma:helper}, and once more the approximation guarantee of the $\pi^{(i)}$, we can bound
\begin{align*}
(W^{-1}A^T\pi^*)_{\max}&\leq (W^{-1}A^T\pi^{(i)})_{\max}\\
&\leq \potfun_{\beta_i'}(\pi^{(i)})\\
&\leq q_i^{k_i}\left(1+\frac{\varepsilon_i}{4}\right)^{\ell_i}\potfun_{\beta_{i-1}}(\pi^{(i-1)})\\
&\leq q_i^{k_i}\left(1+\frac{\varepsilon_i}{4}\right)^{\ell_i+1}(W^{-1}A^T\pi^{(i-1)})_{\max}\\
&\leq q_i^{k_i}\left(1+\frac{\varepsilon_i}{4}\right)^{\ell_i+1} (W^{-1}A^T\pi^*)_{\max}\cdot
\begin{cases}
\alpha\lambda & \mbox{if }i=1\\
(1+\varepsilon_{i-1}) & \mbox{if }i>1
\end{cases}\\
&\leq q_i^{k_i} (W^{-1}A^T\pi^*)_{\max}\cdot
\begin{cases}
\left(1+\frac{\varepsilon_1}{4}\right)^{\ell_1+1}\alpha\lambda & \mbox{if }i=1\\
\left(1+\frac{\varepsilon_i}{4}\right)(1+\varepsilon_{i-1}) & \mbox{if }i>1
\end{cases}\\
&\stackrel{\varepsilon_i \leq \varepsilon_{i-1} \leq 1}{<}q_i^{k_i} (W^{-1}A^T\pi^*)_{\max}\cdot
\begin{cases}
\alpha^2\lambda^2 & \mbox{if }i=1\\
1+\frac{3\varepsilon_{i-1}}{2} & \mbox{if }i>1.
\end{cases}
\end{align*}
Dividing by $(W^{-1}A^T\pi^*)_{\max}$, taking the logarithm, and rearranging, we infer that
\begin{align*}
k_i&<\begin{cases}
\frac{\log(\alpha^2\lambda^2)}{-\log q_1} & \mbox{if }i=1\\[1ex]
\frac{\log\left(1+\frac{3\varepsilon_{i-1}}{2}\right)}{-\log q_i} & \mbox{if }i>1
\end{cases}\\
&\in \begin{cases}
O\left((\log\alpha+\log \lambda)\alpha^2\lambda^2\ln(2m)\right) & \mbox{if }i=1\\
O\left(\varepsilon_i^{-2}\alpha^2\lambda^2\ln(2m)\right) & \mbox{if }i>1.
\end{cases}
\end{align*}
Summing over all phases $i$, the total number of iterations is bounded by
\begin{align*}
\sum_{i=1}^{i_{\max}}k_i &\in O\left(\left(\log\alpha+\log \lambda+\sum_{i=2}^{i_{\max}}\varepsilon_i^{-2}\right)\alpha^2\lambda^2\ln(2m)\right)\\
&= O\left(\left(\varepsilon^{-2}+\log \alpha+ \log \lambda\right)\alpha^2\lambda^2\log(n)\right). \qedhere
\end{align*}
\end{proof}
The bound on the number of iterations of the algorithm readily translates to one on the specific operations the algorithm performs.
\begin{corollary}\label{coro:iterations_asym}
Algorithm \emph{\nablaalg} can be executed using a total of $O((\varepsilon^{-2}+\log \alpha+ \log \lambda)\alpha^2\lambda^2 \log n)$ operations of the following types:
\begin{itemize}
  \item oracle call,
  \item computation of $(W^{-1}A^T\pi)_{\max}$ for a given $\pi$,
  \item computation of $\nabla\potfun_{\beta}(\pi)$ for given $\beta$ and $\pi$,
  \item scalar product $\bt^T h$ for given $\bt$ and $h$,
  \item comparisons of and multiplications with scalar values.
\end{itemize}
\end{corollary}
\begin{proof}
Follows by checking the lines of the code, the statements referred to in the comments, and the bound on the number of iterations given in Lemma~\ref{lemma:iterations_asym}.
\end{proof}

We summarize the results of this section in the following theorem.
\begin{theorem}\label{theorem:asymmetric}
Given an oracle that computes $\alpha$-approximate dual solutions to~\eqref{formula:symmetric_constrained}, Algorithm \emph{\nablaalg} computes, for any $ 0 < \varepsilon \leq 1 $, a $(1+\varepsilon)$-approximate dual solution to the transshipment problem~\eqref{formula:asymmetric} on bidirected graphs with positive arc weights calling the oracle $O((\varepsilon^{-2}+\log \alpha + \log \lambda)\alpha^2\lambda^2\log n)$ times. If, on termination of the algorithm, the oracle provides an $\alpha$-approximate primal-dual pair of solutions to \eqref{formula:symmetric_constrained}, a $(1+\varepsilon)$-approximate primal-dual pair of solutions to \eqref{formula:asymmetric} can be readily determined (see Observation~\ref{obs:primal}).
\end{theorem}

%% file: sssp.tex
\section{Single-Source Shortest Paths}\label{sec:sssp}

For any instance of the transshipment problem, there is always an optimal primal solution whose arcs with non-zero flow form a forest. This is a well-known result of the structure of linear programs~\cite[Theorem 11.1]{AhujaMO93}; for completeness, we directly prove the statement here for the problem at hand.
\begin{lemma}\label{lemma:tree}
\eqref{formula:asymmetric} has an optimal primal solution that sends flow only along the arcs of a forest.
\end{lemma}
\begin{proof}
Suppose $x\in \RR^{2m}$ is an optimal primal solution and let $C$ be any cycle so that for each edge $ e = \{u,v\}\in C$, $x_{(u,v)}>0$ or $x_{(v,u)}>0$. Consistently direct $C$. Let $f\in \RR^{2m}$ (not necessarily satisfying $f\geq 0$) send $\min\{x_{(u,v)}\,|\,\{u,v\}\in C \wedge x_{(u,v)}\neq 0\}$ units of flow in positive direction ``around'' $C$, so that $f_{(u,v)}\neq 0$ implies $x_{(u,v)}\neq 0$. By construction, both $x-f\geq 0$ and $x+f\geq 0$ and, as $C$ is a cycle, $Af=0$, i.e., $A(x-f)=A(x+f)=Ax=b$. Thus, both $x-f$ and $x+f$ are feasible.

Note that $x-f$ or $x+f$ satisfies that there is one arc less on $C$ carrying flow; assume w.l.o.g.\ it is $x-f$. We claim that $\ones^T W(x-f)=\ones^T Wx$. Assuming for contradiction that this is false, either $\ones^T W(x-f)<\ones^T W x$ or $\ones^T W (x+f)< \ones^T W x$, i.e., either $x-f$ or $x+f$ has smaller objective than~$x$, contradicting its optimality. We conclude that $x-f$ is also an optimal solution. The claim now follows by inductively repeating the argument until no cycles remain.
\end{proof}

When considering the special case of the single-source shortest path (SSSP) problem, two aspects of the solution given in Section~\ref{sec:transshipment} are unsatisfactory. First, if a primal solution is computed, there is no guarantee that it induces a tree, and thus it is not clear which arc one should traverse from $v$ when searching for a short path to $s$. Second, approximating the optimal solution guarantees only that the computed upper bounds on the distances from the source $s$ are a $(1+\varepsilon)$-approximation \emph{on average,} i.e., for any specific node $v$ it may still be the case that $y_v-y_s\gg d(s,v)$, where $d(\cdot,\cdot)$ is the distance metric on $V$ induced by $G$.

The goal of this section is to resolve these issues. Our solution is based on a simple sampling procedure using Algorithm~\nablaalg as a subroutine. It results in an approximate shortest-path tree rooted at $s$, which guarantees a $(1+\varepsilon)$-approximation for each distance from $s$ to a $v\in V$.

\subsection{Sampling a Tree Solution}\label{sec:tree}

As the first step towards the randomized solution, we construct a (primal) tree solution that is good on average.
We assume in the following that there is a single source node with negative demand and the demand on every non-source node is either $ 0 $ or $ 1 $.
The idea, which relies on our specific choice of a spanner-based oracle, is as follows:
\begin{enumerate}
  \item Run Algorithm \nablaalg for $\varepsilon':=\tfrac{\varepsilon}{6}$ until termination and denote by $x$ the returned primal solution.
  \item For each node $v\in V$, partition its incoming arcs $(u,v)$ with $x_{(u,v)}>0$ into classes in which arc weights differ by factor at most $2$ (we will ensure that there are only $O(\log n)$ classes). Denoting by $f_{(u,v)}$ the sum of flows of arcs in the class of $(u,v)$, sample $(u,v)$ with (roughly) probability $\min\{f_{(u,v)}^{-1}(2\lambda\alpha+1)x_{(u,v)},1\}$.
  \item We show that an optimal solution using only sampled arcs and spanner edges is a $(1+\varepsilon)$-approximation with probability at least~$\tfrac{2}{3}$, and that we can bound the number of arcs sampled by the procedure by $O(\alpha\lambda n \log n)$ with probability at least~$\tfrac{2}{3}$.
  \item We abort the procedure if more arcs are selected. Otherwise, we compute and return an optimum tree solution on the selected arc set. By the union bound, we thus obtain a $(1+\varepsilon)$-approximation using only the spanner edges and $O(\alpha\lambda n \log n)$ sampled arcs with probability at least~$\tfrac{1}{3}$.
  \item To obtain a Las Vegas algorithm, we simply repeat the procedure until a good solution is obtained; this can be checked by comparing the objective value of the obtained tree solution to the one of the dual solution returned by Algorithm \nablaalg.
\end{enumerate}
The second last step exploits that in the computational models considered in this work, operations on sparse graphs are considerably cheaper; thus, this strategy is tied to using an oracle based on a sparse graph.

In order to prove the correctness of the above procedure, conceptually decompose $x=:x^\vartriangle+x^{\circ}$, where the arcs with non-zero flow in $x^\vartriangle$ form a directed acyclic graph (DAG). Such a decomposition is always feasible, which again is a standard property whose proof we give for completeness.
\begin{lemma}\label{lemma:decompose}
Any primal solution to \eqref{formula:asymmetric} can be decomposed into a feasible solution~$x^\vartriangle$ inducing a DAG satisfying $Ax^\vartriangle=0$ and $x^{\circ}\geq 0$ satisfying $Ax^{\circ}=0$.
\end{lemma}
\begin{proof}
The proof is very similar to the one of Lemma~\ref{lemma:tree}. Suppose $x\in \RR^n$ is a primal solution and let $C$ be any \emph{consistently oriented} cycle so that for each arc $(u,v)\in C$, $x_{(u,v)}>0$ (if no such~$C$ exists, $x^\vartriangle:=x$ and $x^{\circ}=0$ meet the claim of the lemma). Let $f\in \RR^{2m}$ send $\min\{x_{(u,v)}\,|\,(u,v)\in C \wedge x_{(u,v)}\neq 0\}$ units of flow in positive direction ``around'' $C$. By construction both $f$ and $x-f$ are feasible, and $x-f$ has one arc less carrying non-zero flow. The claim now follows by inductively repeating the argument until no directed cycles remain, where $x^{\circ}$ is the sum of the flows $f$ from the individual steps.
\end{proof}
We would like to sample from $x^\vartriangle$, but have to face the ``disturbance'' by $x^{\circ}$. We will exploit that the flow $x^{\circ}$ is ``cheap'' to bound its effect on the sampled solution. In order to formalize this, we fix some notation.
In the following, let $x\in \RR^{2m}$ be the primal solution for \eqref{formula:asymmetric} returned by Algorithm~\nablaalg when called with accuracy parameter $\varepsilon':=\tfrac{\varepsilon}{6}$. Decompose $x=:x^\vartriangle+x^{\circ}$ according to Lemma~\ref{lemma:decompose}. We define $f_v^\vartriangle:=\sum_{\{v,u\}\in E} x^\vartriangle_{(u,v)}$.
\begin{lemma}\label{lemma:sample_acyclic}
Suppose that for $v\in V\setminus \{s\}$, $b_v\geq 0$. Let each $v\in V$ with $f_v^\vartriangle>0$ sample one incoming arc, where the probability to choose arc $(u,v)\in E$ is $(f_v^\vartriangle)^{-1}x_{(u,v)}^\vartriangle$. Then the resulting arc set is a directed tree $T$ rooted at the source $s$ spanning all nodes with non-zero demand. Denoting by $t\in \RR^{2m}$ the (unique) flow with $At=b$ and $t_e=0$ for any arc $e\notin T$, it holds that $\Exp[\ones^TWt]=\ones^TWx^\vartriangle$.
\end{lemma}
\begin{proof}
Recall that $Ax^\vartriangle=b$. As $s$ is the only node with supply, i.e., $b_v\geq 0$ for all $v\in V\setminus \{s\}$, any such $v$ that has an outgoing arc $(v,u)\in E$ carrying non-zero flow or that satisfies $b_v>0$ must have an incoming arc $(u,v)\in E$ carrying non-zero flow. Thus, we can inductively construct a directed path ending at any such node by following the incoming arc of the (current) first node of the path, until either this first node is $s$ or we close a directed cycle. The latter is not possible, because $x^\vartriangle$ induces a DAG by construction (cf.\ Lemma~\ref{lemma:decompose}). Hence, the sampled graph is a DAG in which each node that sampled an arc is reachable from $s$. Thus, as each node sampled at most one incoming arc, the sampled graph must be a tree $T$ rooted at $s$. As we observed that each $v$ with $b_v>0$ must have an incoming arc carrying non-zero flow, each such $v$ sampled an arc, implying that $T$ spans the nodes of non-zero demand. Accordingly, there is a unique flow $t$ on $T$ such that $At=b$.

It remains to show that $\Exp[\ones^TWt]=\ones^TWx^\vartriangle$.
We show this by induction on the number of nodes in the DAG induced by $ x^\vartriangle $.
The base case where the DAG induced by $ x^\vartriangle $ has $ 0 $ nodes is trivial.
For the induction step, let $v$ be a node so that $x_{(v,u)}^\vartriangle=0$ for all $(v,u)\in E$; such a node must exist, as $x^\vartriangle$ induces a DAG.
If $b_v=0$, $f_v^\vartriangle=b_v=0$ and the claim is trivially true, so suppose $f_v^\vartriangle=b_v>0$. We interpret the random choice of $v$ as follows: $v$ picks an incoming arc $(u,v)$ and we route $b_v$ units of flow from $u$ to $v$, changing demands accordingly by setting $ b_v' = 0 $ and $ b_u' = b_u + b_v $. In expectation, this induces cost
\begin{align*}
K:=\sum_{(u,v)\in E}w_{(u,v)}b_v\cdot\frac{x_{(u,v)}^\vartriangle}{f_v^\vartriangle}=\sum_{(u,v)\in E}w_{(u,v)}x_{(u,v)}^\vartriangle,
\end{align*}
where we used that $f_v^\vartriangle=b_v$ for each $v$.
The modified demands $b'$ satisfy for $(u,v)\in E$ with $x_{(u,v)}>0$ that
\begin{align*}
\Exp[b_u']&=b_u+\frac{x_{(u,v)}^\vartriangle}{f_v^\vartriangle}\cdot b_v\stackrel{f_v^\vartriangle=b_v}{=}b_u+x_{(u,v)}^\vartriangle,
\end{align*}
that $b_v'=0$, and that $b_u'=b_u$ at all other nodes $u$.
Now $ x' $ by $ x'_{(u,v)} = 0 $ and $ x'_{(u',v')} = x_{(u',v')} $ for any other arc $ (u', v') \in E $.
Observe that $ x' $ can be decomposed into $ (x')^\vartriangle $ and $ (x')^\circ $ (as in Lemma~\ref{lemma:decompose}) such that the DAG induced by $ (x')^\vartriangle $ amounts to the DAG induced by $ x^\vartriangle $ with the node $ v $ removed.
Let $ T' $ denote the tree sampled from $ (x')^\vartriangle $.
Recall that for independent random variables $X$ and~$Y$, $\Exp[XY]=\Exp[X]\Exp[Y]$. As the random variable $b'$ does not depend on the random choices of any nodes but $v$, it is independent of the tree $T'$.
We denote by $X_u$ the cost of routing one unit of flow from $s$ to $u\in T\setminus\{v,s\}$ on $T'$. By linearity of expectation we then obtain
\begin{align*}
\Exp\left[\ones^TWt\right]&=K+\Exp\left[\sum_{u\in T\setminus\{v,s\}}b_u'X_u\right]\\
&=K+\sum_{u\in T\setminus\{v,s\}}\Exp\left[b_u'X_u\right]\\
&=K+\sum_{u\in T\setminus\{v,s\}}\Exp\left[b_u'\right]\Exp\left[X_u\right].
\end{align*}
Examining the sum in the final term, observe that when deleting $v$ from the graph, the respective restriction of $x^\vartriangle$ (namely $ (x')^\vartriangle $ routes exactly demands $\Exp[b_u']\geq b_u\geq 0$ from $s$ to each $u\in T\setminus\{v,s\}$. Thus, by the induction hypothesis and, again, linearity of expectation,
\begin{align*}
\sum_{u\in T\setminus\{v,s\}}\Exp\left[b_u'\right]\Exp\left[X_u\right]=\ones^TWx^\vartriangle-\sum_{(u,v)\in E}w_{(u,v)}x_{(u,v)}^\vartriangle=\ones^TWx^\vartriangle-K.
\end{align*}
We conclude that indeed $\Exp[\ones^TWt]=\ones^TWx^\vartriangle$, completing the proof.
\end{proof}
Our goal is now to apply Lemma~\ref{lemma:sample_acyclic} followed by Markov's bound in order to show that sampling arcs is sufficient for obtaining a good tree solution. Alas, the decomposition $x=x^\vartriangle+x^{\circ}$ is unknown, and Lemma~\ref{lemma:sample_acyclic} critically depends on the fact that $x^\vartriangle$ is acyclic. A naive solution would be to sample sufficiently often according to $x$ so that each arc has at least the same probability to be sampled when sampling from the ``correct'' distribution given by $x^\vartriangle$. Unfortunately, it is possible that an arc $(u,v)$ satisfies $x^\vartriangle_{(u,v)}\ll x^{\circ}_{(u,v)}$, conflicting with the requirement of sampling few arcs. We overcome this obstacle by replacing such ``bad'' arcs by a corresponding path in the spanner.

In the following, denote $w_{\min}:=\min_{(u,v)\in E}\{w_{(u,v)}\}$ and partition $E$ into sets $E_{v,k}:=\{{(u,v)\in E} \mid w_{(u,v)}\in [2^{k-1}w_{\min},2^kw_{\min})\wedge x_{(u,v)}>0\}$, where $v\in V$ and $k\in \ZZ_{>0}$. The set $E_{v,k}\neq \emptyset$ is \emph{bad,} if $2\alpha\lambda\sum_{e\in E_{v,k}}x^\vartriangle_e\leq \sum_{e\in E_{v,k}}x^{\circ}_e$. We say that $e\in E$ is bad if the set $E_{v,k}$ containing $e$ is bad.
\begin{lemma}\label{lemma:redirect}
Redirecting the flow $x_{(u,v)}^\vartriangle$ of each bad arc $(u,v)$ over a shortest $u$-$v$ path on an (undirected) $\alpha$-spanner for the graph with weights $w^-$ incurs an additional cost of at most $\varepsilon'\ones^TWx^*$.
\end{lemma}
\begin{proof}
As the maximum ratio between the weights of an edge in opposing directions is $\lambda$, an undirected $\alpha$-spanner for weights $w^-$ clearly allows, for any $(u,v)\in E$, finding a directed path of length at most $\alpha \lambda w_{(u,v)}$ from $u$ to $v$. The cost for rerouting the flow on bad arcs over the spanner is thus bounded by
\begin{align*}
\sum_{\mathrm{bad }(u,v)\in E}\alpha \lambda w_{(u,v)}x_{(u,v)}^\vartriangle&=\sum_{\mathrm{bad }E_{u,k}}\sum_{e\in E_{u,k}}\alpha \lambda w_e x_e^\vartriangle\\
&< \sum_{\mathrm{bad }E_{u,k}}2^k\alpha \lambda w_{\min}\sum_{e\in E_{u,k}} x_e^\vartriangle\\
&\leq \sum_{\mathrm{bad }E_{u,k}}2^{k-1}w_{\min}\sum_{e\in E_{u,k}} x_e^{\circ}\\
&\leq \sum_{\mathrm{bad }E_{u,k}}\sum_{e\in E_{u,k}} w_e x_e^{\circ}\\
&=\sum_{\mathrm{bad }(u,v)\in E}w_{(u,v)}x_{(u,v)}^{\circ}\\
&\leq \ones^T Wx^{\circ}.
\end{align*}
Furthermore, by Lemma~\ref{lemma:correctness} and the fact that $\ones^TWx^\vartriangle\geq \ones^TWx^*$ by feasibility of $x^\vartriangle$ and optimality of $x^*$,
\begin{align*}
\ones^TWx^{\circ}=\ones^TWx-\ones^TWx^\vartriangle\leq (1+\varepsilon')\ones^TWx^*-\ones^TWx^*=\varepsilon'\ones^TWx^*. \qedhere
\end{align*}
\end{proof}

We now are ready to show that sampling sufficiently many arcs according to~$x$ and adding the spanner is, in expectation, almost as good as directly sampling according to $x^\vartriangle$.
\begin{lemma}\label{lemma:sample}
Suppose that for $v\in V\setminus \{s\}$, $b_v\geq 0$. For each $E_{v,k}\neq \emptyset$, sample arc $e=(u,v)\in E_{v,k}$ with independent probability
\begin{align*}
p_e:=\min\left\{\frac{(2\alpha\lambda+1) x_e}{\sum_{e'\in E_{v,k}}x_{e'}},1\right\},
\end{align*}
and repeat this procedure until at least one arc from $E_{v,k}$ is selected. Add an arbitrary such arc, and add all spanner edges. Then the expected cost of an optimal solution on the induced graph is bounded by $(1+2\varepsilon')\ones^TWx^*$.
\end{lemma}
\begin{proof}
By Lemma~\ref{lemma:sample_acyclic}, sampling arc $(u,v)\in E$ with probability $(f_u^\vartriangle)^{-1}x_{(u,v)}^\vartriangle$ at $u\in V$ with $f_u^\vartriangle>0$ results in a flow $t$ on a tree of expected cost at most $\ones^TWx^\vartriangle$. Observe that if the probability to select $(u,v)\in E$ is \emph{at least} as large (and we are guaranteed to select at least one incoming arc for each $v\in V$), the expected cost can only become smaller. Consider an arc $e\in E_{v,k}$ with $p_e<1$ that is not bad. It satisfies that
\begin{align*}
p_e &\geq \frac{(2\alpha\lambda+1) x_e^\vartriangle}{\sum_{e'\in E_{v,k}}x_{e'}}\\
&=\frac{(2\alpha\lambda+1) x_e^\vartriangle}{\sum_{e'\in E_{v,k}}x_{e'}^\vartriangle+\sum_{e'\in E_{v,k}}x_{e'}^{\circ}}\\
&>\frac{(2\alpha\lambda+1) x_e^\vartriangle}{(2\alpha\lambda+1)\sum_{e'\in E_{v,k}}x_{e'}^\vartriangle} & \text{($ e $ is not bad)} \\
&=\frac{x_e^\vartriangle}{f_v^\vartriangle},
\end{align*}
i.e., the probability to select such an arc is sufficiently large. The same is trivially true if $p_e=1$. Hence, it remains to address bad arcs.

To this end, we apply Lemma~\ref{lemma:redirect} to see that we can reroute their flow over the spanner at an additive additional cost of at most $\varepsilon'\ones^T Wx^*$. Logically, we can reflect this by replacing the weight of such an arc by the weight of a respective shortest path in the spanner and adding all bad arcs to the sample (i.e., ``sampling'' each bad arc with probability $1$); denoting the new weights by $W'$, we are then guaranteed that the union of the spanner and the actually sampled arcs contains a solution of expected cost at most
\begin{align*}
\Exp\left[\ones^TW't\right]= \ones^TW'x\leq \ones^TWx+\varepsilon'\ones^T Wx^*\leq (1+2\varepsilon')\ones^TWx^*. \qedhere
\end{align*}
\end{proof}

\begin{corollary}\label{coro:sample}
Performing the sampling procedure of Lemma~\ref{lemma:sample} and returning an optimum primal tree solution on the graph induced by the sampled arcs and the spanner yields a $(1+\tfrac{2}{3}\varepsilon)$-approximation with probability at least~$\tfrac{1}{2}$.
\end{corollary}
\begin{proof}
Lemma~\ref{lemma:sample} shows that the expected cost of an optimal solution on the stated arc set is at most $(1+2\varepsilon')\ones^TWx^*=(1+\tfrac{\varepsilon}{3})\ones^TWx^*$. Applying Markov's bound to the (positive) random variable that is the amount by which this cost exceeds that of an optimal solution shows that the cost is at most $(1+\tfrac{2}{3}\varepsilon)\ones^TWx^*$ with probability at least~$\tfrac{1}{2}$. Lemma~\ref{lemma:tree} shows that there is an optimal solution that is a forest, and as there is only one source, it is a tree.
\end{proof}

To complete the discussion of the sampling procedure of Lemma~\ref{lemma:sample}, it remains to show that it will stop fast enough, i.e., that there is a constant probability of sampling at least one arc.
In order to show this, we use the following general observation.

\begin{lemma}\label{lem:arithmetic_geometric_bound}
Let $ p $ be a finite probability mass function for a sample space of size~$ t $, i.e., $ p_i \in [0, 1] $ for all $ 1 \leq i \leq t $ and $ \sum_{1 \leq i \leq t} p_i = 1 $.
Then $ \prod_{1 \leq i \leq t} (1 - p_i) \leq \tfrac{1}{e} $.
\end{lemma}

\begin{proof}
By the inequality of arithmetic and geometric means we have
\begin{equation*}
\left( \prod_{1 \leq i \leq t} (1 - p_i) \right)^{1/t} \leq \frac{1}{t} \cdot \sum_{1 \leq i \leq t} (1 - p_i) = \frac{1}{t} \cdot \left( t - \sum_{1 \leq i \leq t} p_i \right) = \frac{t-1}{t} = 1 - \frac{1}{t} \, .
\end{equation*}
It now follows that
\begin{equation*}
\prod_{1 \leq i \leq t} (1 - p_i) \leq \left( 1 - \frac{1}{t} \right)^t \leq \frac{1}{e} \, ,
\end{equation*}
where the latter inequality can be derived from the limit definiton of Euler's number.
\end{proof}

\begin{lemma}\label{lem:probability to sample at least one arc}
The probability that no arc of $ E_{v,k} $ is sampled in the procedure given in Lemma~\ref{lemma:sample} is at most $ \tfrac{1}{e} $.
\end{lemma}

\begin{proof}
For every arc $ e \in E_{v,k} $, set $ p_e' = \tfrac{x_e}{\sum_{e'\in E_{v,k}}x_{e'}} $.
Clearly, $ p_e' \leq p_e $ and $ p_e' \in [0, 1] $ for every arc $ e \in E_{v,k} $ and $ \sum_{e \in E_{v,k}} p_e' = 1 $.
With Lemma~\ref{lem:arithmetic_geometric_bound} we get that the probability that no arc is sampled is
\begin{equation*}
\prod_{e \in E_{v,k}} (1 - p_e) \leq \prod_{e \in E_{v,k}} (1 - p_e') \leq \frac{1}{e} \, . \qedhere
\end{equation*}
\end{proof}

We put all the steps for constructing a primal tree solution together in the pseudocode given in Algorithm~\ref{alg:tree}.
We summarize the guarantees of this procedure as follows.

\begin{algorithm2e}[ht!]
	\SetKw{Break}{break}
	$ (x, y) := \nablaalg(G, b, \frac{\varepsilon}{6})$ \tcp*{$ (1+\frac{\varepsilon}{6}) $-approximate pair} \nllabel{line:refined approximation to transshipment}
   	\ForEach{$ v \in V $ and $ k \in \{1, \dots, \lceil \log \| w \|_\infty \rceil \} $}{ \nllabel{line:begin sampling}
 	    $E_{v,k}:=\{(u,v)\in E\,|\,w_{(u,v)}\in [2^{k-1}, 2^k)\wedge x_{(u,v)}>0\}$\;
		\If{$ E_{v,k} \neq \emptyset $}{
			$ F_{v,k} := \emptyset $\;
			\For{$ \lceil \log (4 n \lceil \log \| w \|_\infty \rceil) \rceil $ repetitions}{
				\If{$ F_{v,k} = \emptyset $}{
					\ForEach{$ e \in E_{v,k} $}{
						Add $ e $ to $ F_{v,k} $ with probability $ p_e:=\min \left\{ \frac{(2\alpha\lambda+1) x_e}{\sum_{e'\in E_{v,k}}x_{e'}},1 \right\} $\;
						\If{$ |F_{v,k}| = 16 \alpha \lambda + 8 $}{
							$ F_{v,k} \gets \emptyset $\;
							\Break;
						}
					}
				}
			}
		}
	} \nllabel{line:end sampling}
	Construct $ \alpha $-spanner $ S $ of $ G^- = (V, E, w^-) $\nllabel{line:construct spanner S}\;
	Construct graph $ H $ consisting of all sampled arcs $ \bigcup_{v, k} F_{v, k} $ and $ S $\nllabel{line:construct auxiliary graph}\;
	Compute optimal primal tree solution $ t $ and dual solution $ y'$ on $ H $ \nllabel{line:compute optimal tree solution on auxiliary graph}\;
	\eIf(\tcp*[h]{$ (1 + \epsilon) $-approximate pair}){$ \ones W t \leq (1 + \epsilon) b^T y $}{
		\Return $ (t, y) $ \tcp*{algorithm succeeded}
	}{
		\Return $ (\bot, \bot) $ \tcp*{algorithm failed} \nllabel{line:last line of primal tree}
	}
	\caption{\treealgo$(G, b, \varepsilon)$}
	\label{alg:tree}
\end{algorithm2e}

\begin{theorem}\label{theorem:average tree}
If $b_v\in \{0,1\}$ for all $v\in V\setminus \{s\}$, then, given any $ 0 < \varepsilon \leq 1 $, with probability at least~$ \tfrac{1}{4} $, algorithm~\treealgo computes a $(1+\varepsilon)$-approximate primal-dual pair of solutions to~\eqref{formula:asymmetric} in a bidirected graph with positive integer arc weights, where the primal solution has non-zero flow only on the arcs of a tree.
\end{theorem}

\begin{proof}
By Theorem~\ref{theorem:asymmetric}, the primal-dual pair $ (x, y) $ computed by algorithm \nablaalg is a $ (1 + \tfrac{\epsilon}{6}) $-approximate pair, i.e., $ \ones W x \leq (1 + \tfrac{\epsilon}{6}) b^T y $.
For any $ v \in V $ and any $ k \in \{1, \dots, \lceil \log \| w \|_\infty \rceil \} $ we have that for each repetition of the sampling, no arc is sampled, i.e., $ F'_{v,k} = \emptyset $, with probability at most $ \tfrac{1}{e} $ by Lemma~\ref{lem:probability to sample at least one arc}.
Furthermore, by linearity of expectation, the expected number of sampled arcs is $ 2 \alpha \lambda + 1 $.
By the Markov inequality, this expectation is exceeded by a factor of~$ 8 $ with probability at most~$ \tfrac{1}{8} $.
The probability that $ | F_{v,k} | $ reaches size $ 16 \alpha \lambda + 8 $ is therefore at most~$ \tfrac{1}{8} $.
Thus, in each repetition of the sampling, the probability that $ F_{v,k} = \emptyset $ is at most $ \tfrac{1}{e} + \tfrac{1}{8} \leq \tfrac{1}{2} $.
If this happens for all repetitions of the sampling, the algorithm sets $ F_{v,k} = \emptyset $ and might then not be correct.
The probability of this event falls exponentially with the number of repetitions, i.e., is at most $ \tfrac{1}{2^{\lceil \log (4 n \lceil \log \| w \|_\infty \rceil) \rceil}} \leq \tfrac{1}{4 n \lceil \log \| w \|_\infty \rceil} $.
Thus, the probability that $ F_{v,k} = \emptyset $ after the sampling for some $ v \in V $ and some $ k \in \{1, \ldots, \lceil \log \| w \|_\infty \rceil \} $ is, again by the union bound, at most $ n \lceil \log \| w \|_\infty \rceil \cdot \tfrac{1}{4 n \lceil \log \| w \|_\infty \rceil} = \tfrac{1}{4} $.
Moreover, Corollary~\ref{coro:sample} states that computing an optimal tree solution $ t $ using only sampled and spanner arcs yields a $(1+\tfrac{2}{3}\varepsilon)$-approximation, i.e., $ \ones W t \leq (1 + \tfrac{2}{3} \epsilon) \ones W x^* $, with probability at least~$\tfrac{2}{3}$; we compute such a solution.
Applying the union bound, we conclude that with probability at least~$\tfrac{1}{2} - \tfrac{1}{4} = \tfrac{1}{4}$, both all the sets of sampled arcs are non-empty and the resulting solution is of sufficient quality.
In that case we have
\begin{equation*}
\ones W t \leq (1 + \tfrac{2}{3} \epsilon) \ones W x^* \leq (1 + \tfrac{2}{3} \epsilon) \ones W x \leq (1 + \tfrac{2}{3} \epsilon) (1 + \tfrac{1}{6} \epsilon) \ones b^T y \leq (1 + \epsilon) b^T y \, ,
\end{equation*}
i.e., $ (t, y) $ is a $ (1 + \epsilon) $-approximate primal-dual pair and the algorithm succeeds.
\end{proof}

\subsection{Computing an Approximate Shortest-Path Tree}\label{sec:per_node}

Next, we address the point that an on-average guarantee is typically insufficient when considering the single-source shortest path problem. Rather, we would like an \emph{approximate shortest-path tree,} i.e., a tree rooted at the source $s$ such that, for each $v\in V$, the distance from $s$ to $v$ in the tree is larger by a factor of at most $1+\varepsilon$ than the distance in $G$.

Our approach here is straightforward.
The algorithm (see Algorithm~\ref{alg:sssp-randomized} for pseudocode) starts with solving a ``single-source'' transshipment instance, in which $ s $ has demand $ -n + 1 $ and every other node has demand $ 1 $, up to increased precision $ \varepsilon' = \tfrac{\varepsilon}{12} $.
In particular, it computes a primal tree solution providing and ``on-average'' guarantee and a dual solution.
Next, it determines the distances from the source to every node in this tree.
By weak duality, we know that for all nodes whose distance in the tree exceeds their dual variable by a factor at most $1+12\varepsilon' \leq 1 + \varepsilon$ we have already found a suitable approximate shortest path in the current tree.
Our algorithm will consider these nodes ``done'', sets their demands to~$ 0 $ and repeats this process with the modified transshipment instance.
We can argue that the nodes whose distance in the tree exceeds their dual variable by a factor at most $1+12\varepsilon'$ make up a constant fraction of the optimal value of the transshipment instance.
Therefore, this algorithm will be finished after logarithmically many repetitions.
Now the desired approximate shortest path tree can be found by taking the union of all the intermediate trees with ``on-average'' guarantee and computing a shortest path tree in the respective graph.

\begin{algorithm2e}[ht]
	Set $E':=\emptyset$, $b:=\ones -n\ones_s$, and $\varepsilon':=\tfrac{\varepsilon}{12}$.\;
	\While{$b_s<0$}{
	    $ (T, y) := \treealgo(G, b, \varepsilon') $ \nllabel{line:compute primal tree solution}\;
	    \If{$ (T, y) \neq (\bot, \bot) $}{
			Set $E'\leftarrow E'\cup T$\;
			\ForEach{$v\in V$ with $b_v=1$}{
				\If{$d_T(v)\leq (1+12\varepsilon')(y_v-y_s)$}{
					Set $b_v \leftarrow 0$ and $b_s \leftarrow b_s +1$
				}
			}
		}
	}
	Compute shortest-path tree $T'$ on $(V,E',w|_{E'})$ rooted at $s$\;
	\Return{$T'$}
	\caption{\ssspalgo$(G, s, \varepsilon)$}
	\label{alg:sssp-randomized}
\end{algorithm2e}

In the following, assume w.l.o.g.\ that $y_s=y_s^*=0$, i.e., for each $v\in V$, $y_v^*$ is the distance from $s$ to $v$ and $y_v\leq y_v^*$ is an approximation thereof. Thus, $b^Ty^*$ is the sum of all distances from $s$, and $b^Ty$ is close to this sum, where no distance is overestimated. Assigning to each node $v\in V$ weight~$y^*_v$ (i.e., its distance to $s$), this entails that the sum of weights of nodes satisfying $y_v\geq \tfrac{y_v^*}{1+4\varepsilon'}$ is a constant fraction of the total weights.
\begin{lemma}\label{lemma:good_dual}
Assume that $b_v \ge 0$ for all $v \in V \setminus \{s\}$, $y \in \RR^n$ is a $(1+\varepsilon')$-approximate dual solution to \eqref{formula:asymmetric}, and $Y := \{ v \in V \setminus \{s\}: y_v < y_v^*/(1+4\varepsilon')\}$ for some $\varepsilon' \le \tfrac{1}{4}$. Then,
 $
  \sum_{v \in Y} b_v y_v^* \le \tfrac{1}{2} b^Ty^*.
 $
\end{lemma}
\begin{proof}
We have that
 \begin{align*}
  3\varepsilon' \sum_{v \in Y} b_v y_v &\le 3\varepsilon' \sum_{v \in Y} b_v y_v + (1+\varepsilon') b^T y   - b^T y^* \\
  &= (1+4\varepsilon')\sum_{v \in Y} b_v y_v + (1+\varepsilon') \sum_{v \in V \setminus Y} b_v y_v - b^T y^* \\
  &< \sum_{v \in Y} b_v y_v^*         + (1+\varepsilon') \sum_{v \in V \setminus Y} b_v y_v - b^T y^* \\
  &= (1+\varepsilon') \sum_{v \in V \setminus Y} b_v y_v - \sum_{v \in V \setminus Y} b_v y_v^* \\
  &\le \varepsilon' \sum_{v \in V \setminus Y} b_v y_v,
 \end{align*}
where the last inequality follows from the observation that $\qmaxstretch{y}=1$ guarantees that $y_v \le y_v^*$ for all $v \in V$, the fact that $b_v\ge 0$ for all $v\in V\setminus \{s\}$, and the assumption that $y_s=0$. It follows that $\sum_{v\in V\setminus Y}b_v y_v \ge \tfrac{3}{4}b^Ty$. We conclude that
 \begin{align*}
  \sum_{v \in Y} b_v y_v^*
  &=  b^T y^* - \sum_{v \in V \setminus Y} b_v y_v^*\\
  &\le (1+\varepsilon')b^T y - \sum_{v \in V \setminus Y} b_v y_v\\
  &\leq \left(\frac{1}{4}+\varepsilon'\right)b^T y\\
  &\stackrel{\varepsilon'\leq \tfrac{1}{4}}{\le} \frac{1}{2} b^T y^*. \qedhere
 \end{align*}
\end{proof}
In other words, with respect to node weights $y_v^*$, we have a good lower bound on distances for a large fraction of the total weight. A similar argument applies to a primal solution on a tree, which yields good upper bounds on distances for a large fraction of the weight.
\begin{lemma}\label{lemma:good_primal}
Assume that $\varepsilon'\leq \tfrac{1}{4}$ and $b_v\geq 0$ for all $v\in V\setminus \{s\}$. Moreover, suppose that $t\in \RR^{2m}$ is a $(1+\varepsilon')$-approximate primal solution to \eqref{formula:asymmetric} satisfying $t_e=0$ for all $e\in E\setminus T$, where $T$ is a tree rooted at~$s$. Denote by $d_T(v)$ the distances from $s$ to $v\in V$ in $T$. Then $X := \{ v \in V \setminus \{s\}: d_T(v) > (1+4\varepsilon')y^*_v\}$ satisfies that
 $
  \sum_{v \in X} b_v y_v^* \le \tfrac{1}{4} b^Ty^*.
 $
\end{lemma}
\begin{proof}
Recall that $y_s^*=0$, i.e., the (minimal) cost of routing $b_v$ units of flow from $s$ to $v$ in $G$ is $b_v y_v^*$, $b_v$ times the distance from $v$ to $s$. Similarly, routing $b_v$ units of flow from $s$ to $v$ in $T$ incurs cost $b_v d_T(v)$. Hence,
\begin{align*}
(1+\varepsilon')b^T y^* = (1+\varepsilon')\ones^TWx^* &\geq \ones^TWt\\
&=\sum_{v\in V}b_v d_T(v)\\
&=\sum_{v\in X}b_v d_T(v)+\sum_{v\in V\setminus X}b_v d_T(v)\\
&\geq \sum_{v\in X}b_v(1+4\varepsilon')y^*_v+\sum_{v\in V\setminus X}b_v y^*_v\\
&=4\varepsilon'\sum_{v\in X} b_vy^*_v +b^Ty^*.
\end{align*}
Subtracting $b^Ty^*$ and diving by $4\varepsilon'$, the claim follows.
\end{proof}
Together, these lemmas show that a tree solution and a dual solution provide us with the means to detect nodes for which the distances in $T$ are good approximations to their distances in $G$, while ensuring sufficient progress.
\begin{corollary}\label{coro:good}
Use the notation of Lemmas~\ref{lemma:good_dual} and~\ref{lemma:good_primal} and assume that their preconditions are satisfied. Then $Z:=\{v\in V\setminus \{s\}: d_T(v)>(1+12\varepsilon')y_v\}$ satisfies $\sum_{v \in Z} b_v y_v^* \leq \tfrac{3}{4} b^Ty^*$.
\end{corollary}
\begin{proof}
Any $v\in V\setminus (X\cup Y)$ satisfies
\begin{align*}
d_T(v)\leq (1+4\varepsilon')y^*_v\leq (1+4\varepsilon')^2y_v\stackrel{\varepsilon'\leq \tfrac{1}{4}}{\leq}(1+12\varepsilon')y_v.
\end{align*}
Hence, $Z\subseteq X\cup Y$ and the claim follows from the bounds provided by the lemmas.
\end{proof}

\begin{theorem}\label{theorem:tree}
Given any $ 0 < \varepsilon \leq 1 $, algorithm~\ssspalgo computes a $ (1 + \varepsilon) $-approximate single-source shortest path tree in bidirected graphs with positive polynomially bounded integer arc weights, i.e., for each $v\in V$, the distance from $s$ to~$v$ in the returned tree is at most factor $1+\varepsilon$ larger than in the input graph.
The arc set~$ E' $ for the final shortest path computation has size $ O (n \log \| w \|_\infty) $ and the algorithm performs $ O (c \log \| w \|_\infty) $ iterations with probability at least $1 - n^{-c}$ for every $c \geq 1 $.
\end{theorem}

\begin{proof}
Observe that $b$ remains integral and feasible throughout the algorithm and $b_v \in \{0, 1\}$ for all $v\in V\setminus\{s\}$.
Therefore, $b^T y^*<1$ (the minimum arc weight) implies that $b=0$ and thus that the algorithm terminates.
As $b^Ty^*$ is initially bounded from above by $n^2 \| w \|_\infty $, $n$ times the upper bound of $n \| w \|_\infty$ on the length of any shortest path, Corollary~\ref{coro:good} implies that the outer loop of the algorithm terminates after $ r := \lfloor \log_4 ( n^2 \| w \|_\infty) \rfloor + 1 = O (\log (n \| w \|_\infty)) $ iterations for which the call of procedure $ \treealgo $ was successful.
For each iteration this happens with probability at least~$ \tfrac{1}{4} $ by Theorem~\ref{theorem:average tree}.
The algorithm ensures that for each $v\in V\setminus \{s\}$, there is an iteration in which the tree $T$ added to $E'$ and the dual solution $y$ satisfy
\begin{align*}
d_T(v)\leq (1+12\varepsilon')(y_v-y_s)=(1+\varepsilon)(y_v-y_s)\leq (1+\varepsilon)(y_v^*-y_s^*),
\end{align*}
i.e., the distance from $s$ to~$v$ in~$T$ is at most factor $1+\varepsilon$ larger than the distance from $s$ to $v$ in $G$. It follows that the returned shortest-path tree on $(V,E',W|_{E'})$ satisfies the claimed approximation guarantee.
The probabilistic bound on the number of iterations is immediate from standard Chernoff bounds.
As arcs are added to $ E' $ only in the $ O (\log (n \| w \|_\infty)) $ iterations with a successful call of procedure $ \treealgo $ and in each such iteration we add a tree of $ O (n) $ arcs, the size of $ E' $ is $ O (n \log (n \| w \|_\infty)) $.
\end{proof}

%% file: distributed.tex
\section{Implementation in Distributed and Streaming Models}\label{sec:apps}

In the following we explain how to implement the gradient descent algorithm for computing $ (1+\varepsilon) $-approximate transshipments and SSSP in distributed and streaming models.
Common to all these implementations is the use of sparse \emph{spanners.}
We give deterministic algorithms for computing spanners in distributed and streaming models in Section~\ref{sec:spanner}.

\subsection{Broadcast Congested Clique}\label{sec:clique}

\subparagraph*{Model}
In the Broadcast Congested Clique model, the system consists of $n$ fully connected nodes labeled by unique $O(\log n)$-bit identifiers. Computation proceeds in synchronous rounds, where in each round, nodes may perform arbitrary internal computations, broadcast (send) an $O(\log n)$-bit message to all other nodes, and receive the messages from other nodes. The input is distributed among the nodes.
The first part of the input of every node consists of its incident arcs (given by their endpoints' identifiers) and their weights.
The weights are assumed to be polynomially bounded in $n$, so that they can be encoded using $O(\log n)$ bits.
The second part of the input is problem-specific:
for the transshipment problem, every node $ v $ knows its demand~$ b_v $ and for SSSP $ v $ knows whether or not it is the source $ s $.
In both cases, every node knows~$\varepsilon$ as well.
Each node needs to compute its part of the output.
For the transshipment problem, every node needs to know a $ (1+\varepsilon) $ approximation of the optimum value, and for SSSP every node needs to know a $ (1+\varepsilon) $-approximation of its distance to the source.
For a primal solution to transshipment, every node must know the flow over its incident arcs, and for a dual solution every node must know its potential.
In the case of SSSP, every node needs to know a $ (1+\varepsilon) $-approximation of its distance to the source.
In terms of primal solutions for (approximate) SSSP, only tree solutions are of interest.
Here, each node must know its parent, which is equivalent to knowing the next routing hop on an approximately shortest path to the source.
The complexity of the algorithm is measured in the worst-case number of rounds until the computation is complete.

\subparagraph*{Implementing Algorithm~\nablaalg}

In the following, we explain how to implement our gradient descent algorithm for approximating the transshipment problem.

\begin{theorem}\label{theorem:BCC transshipment_asym}
Given any $ 0 < \varepsilon \leq 1 $, in the Broadcast Congested Clique model a $(1+\varepsilon)$-approximate primal-dual pair of solutions to the transshipment problem in bidirected graphs with positive arc weights can be computed deterministically in $ \widetilde  O (\lambda^2 \varepsilon^{-2}) $ rounds.
\end{theorem}

\begin{proof}
We will show that each iteration of Algorithm~\ref{alg:gradient} can be implemented in a constant number of rounds and that the remainder of the algorithm can be implemented in $ O (\log^2 n) $ rounds.
By Theorem~\ref{theorem:asymmetric}, which bounds the number of iterations of Algorithm~\ref{alg:gradient}, a round complexity of $ O (\lambda^2 (\varepsilon^{-2} + \log{n} + \log{\lambda}) \log^3{n}) $ follows.
For $ \lambda \in \Omega (n) $ we can solve the problem in $ O (n) $ rounds by making all arcs global knowledge.
Therefore, this simplifies to $ \widetilde O (\lambda^2 \varepsilon^{-2}) $.

We will explain how to implement the following lines of the algorithm:
Lines \ref{line:initial oracle call} and~\ref{line:oracle call in iteration}, where the algorithm performs oracle calls, Lines \ref{line:set_beta} and~\ref{line:double_beta}, where the algorithm needs to compute $ (W^{-1}A^T\pi)_{\max} $, Line~\ref{line:compute gradient}, where the algorithm computes $ \nabla\potfun_\beta(\pi) $, and Line~\ref{line:primal solution xt}, where the algorithm computes $ \xt = W^{-1}\nabla \softmax[\beta]{W^{-1}A^T\pi} $.
We implement the latter by computing, for every node $ v $, $ \xt_e $ for every incident arc $ e $.
Once we have explained how to implement the lines mentioned above, it is clear that the rest of the algorithm can be implemented by performing only local computation.

We implement the oracle calls in Lines \ref{line:initial oracle call} and~\ref{line:oracle call in iteration} as follows.
We first make $ b \in \mathbb{R}^n $ known to every node in one round.
Then we compute an $ O (\log{n}) $-spanner of $ G^- = (V, E, w^-) $ with $ O (n \log n) $ edges and make it known to every node.
Using the randomized algorithm of Baswana and Sen~\cite{BaswanaS07}, computing the spanner takes $ O (\log^2 n) $ rounds.
In Appendix~\ref{sec:spanner}, we provide a deterministic version of this algorithm that performs $ O (\log^2 n) $ rounds and computes a spanner of size $ O (n \log n) $.
For the spanner computed in this way, each node knows a subset of its incident spanner edges of size $ O (\log n) $ such that the union of all nodes' subsets equals the whole set of spanner edges (i.e., the arboricity of the spanner is $ O (\log n) $).
Therefore the spanner can be made global knowledge in $ O (\log n) $ rounds.
The oracle calls in Lines \ref{line:initial oracle call} and~\ref{line:oracle call in iteration} require only internal computation.

To explain how Lines \ref{line:set_beta}, \ref{line:compute gradient}, \ref{line:double_beta}, and~\ref{line:primal solution xt} can be implemented in a constant number of rounds each, we introduce the following notation.
For every arc $ e = (u, v) \in E $ define the \emph{stretch} under potentials~$ \pi $ as
\begin{equation*}
s_e (\pi) := \frac{\pi_v - \pi_u}{w_e} \, .
\end{equation*}
Note that $ W^{-1}A^T\pi $ is the vector containing the stretches of the arcs under potential~$ \pi $.
For every node $ v $, let $ \sumarcs_v $ and $ \maxarcs_v $ be the following quantities over its incoming arcs:
\begin{align*}
\sumarcs_v &:= \sum_{e = (u, v) \in E} e^{\beta s_e (\pi)} \\
\maxarcs_v &:= \max \{ s_e (\pi) : e = (u, v) \in E \} \, .
\end{align*}
Thus, for every arc $ e \in E $,
\begin{equation*}
(W^{-1}\nabla \softmax[\beta]{W^{-1}A^T\pi})_e = \frac{e^{\beta s_e (\pi)}}{w_e \cdot \sum_{e' \in E} e^{\beta s_{e'} (\pi)}} = \frac{e^{\beta s_e (\pi)}}{w_e \cdot \sum_{v \in V} \sumarcs_v} \, ,
\end{equation*}
for every node $ v $,
\begin{align*}
\nabla \potfun_\beta(\pi)_v
  &= (A W^{-1} \nabla\softmax[\beta]{W^{-1}A^T\pi})_v \\
  &= \sum_{e = (u, v) \in E} \frac{e^{\beta s_e (\pi)}}{w_e \cdot \sum_{v \in V} \sumarcs_v} - \sum_{e = (v, u) \in E} \frac{e^{\beta s_e (\pi)}}{w_e \cdot \sum_{v \in V} \sumarcs_v} \, ,
\end{align*}
and
\begin{equation*}
(W^{-1}A^T\pi)_{\max} = \max \{ s_e (\pi) : e = (u, v) \in E \} = \max \{\maxarcs_v : v \in V \} \, .
\end{equation*}
The vector $ \pi \in \mathbb{R}^n $ can be made global knowledge in a single round.
As every node~$ v $ knows its incident arcs and their respective weights, it can then internally compute $ s_e (\pi) $ for every incident arc $ e $ as well as~$ \sumarcs_v $ and $ \maxarcs_v $.
In two rounds of communication, the values $ \sumarcs_v $ and $ \maxarcs_v $ of all nodes~$ v $ can be made global knowledge.
Once these values are known, each node can locally compute $ {\nabla\potfun_\beta(\pi)}_v $, $ (W^{-1}A^T\pi)_{\max} $, and $ (W^{-1}\nabla \softmax[\beta]{W^{-1}A^T\pi})_e $ for each incident arc $ e $, using the equations above.
The vector $ {\nabla\potfun_\beta(\pi)} \in \mathbb{R}^n $ can then be made globel knowledge in a single round.
\end{proof}

\subparagraph*{Implementing Algorithm~\treealgo}

In the following we explain how to implement our sampling-based algorithm to compute a primal tree solution for the approximate transshipment problem.

\begin{theorem}\label{theorem:BCC tree solution_asym}
If $b_v\in \{0,1\}$ for all $v\in V\setminus \{s\}$, then in the Broadcast Congested Clique model, given any $ 0 < \varepsilon \leq 1 $, with probability at least $ \tfrac{1}{4} $,
a $(1+\varepsilon)$-approximate primal-dual pair of solutions to the transshipment problem in bidirected graphs with positive polynomially bounded integer arc weights, where the primal solution has non-zero flow only on the arcs of a tree, can be computed in $ \widetilde O (\lambda^2 \varepsilon^{-2}) $ rounds.
\end{theorem}

\begin{proof}
By Theorem~\ref{theorem:BCC transshipment_asym}, the call to Algorithm~\nablaalg in Line~\ref{line:refined approximation to transshipment} takes $ O (\lambda^2 \varepsilon^{-2}) $ rounds.
Recall that the corresponding primal solution~$ x $ is known locally to the nodes and therefore each node $ v $ knows $ x_{(u, v)} $ for each incoming arc $ (u, v) $.
After determining $ \lceil \log \| w \|_\infty \rceil $ in a single round, the sampling in Lines~\ref{line:begin sampling} to~\ref{line:end sampling} can thus be performed without additional communication.
As $ |F_{v, k}| = O (\alpha \lambda) $ for each $ v \in V $ and each $ k \in \{ 1, \ldots, \log \| w \|_\infty \} $, the set $ \bigcup_{v, k} F_{v, k} $ can be made global knowledge in $ O (\alpha \lambda \log \| w \|_\infty ) $ rounds, which for $ \alpha = O (\log n) $ is $ O (\lambda \log n \log \| w \|_\infty) $.
Since $ \| w \|_\infty $ is polynomially bounded, this term is dominated by $ O (\lambda^2 \varepsilon^{-2}) $.
To implement Line~\ref{line:construct spanner S}, we run the algorithm of Corollary~\ref{coro:deterministic spanner} that performs $ O (\log^2 n) $ rounds to compute a spanner~$ S $ of stretch $ \alpha = O (\log n) $ and make it known to all nodes.
Once $ S $ and $ \bigcup_{v, k} F_{v, k} $ are global knowledge, Lines~\ref{line:construct auxiliary graph} to~\ref{line:last line of primal tree} require no additional communication.
\end{proof}

\subparagraph*{Implementing Algorithm~\ssspalgo}

Finally, we explain how to implement our adaptive-demands algorithm for computing an approximate single-source shortest path tree.

\begin{theorem}\label{theorem:BCC sssp_asym}
Given any $ 0 < \varepsilon \leq 1 $, in the Broadcast Congested Clique model a $(1+\varepsilon)$-approximation to the single-source shortest path problem in bidirected graphs with positive polynomially bounded integer arc weights can be computed in $ \widetilde O (c \lambda^2 \varepsilon^{-2}) $ rounds with probability at least $ 1 - n^{-c} $ for every $ c \geq 1 $.
\end{theorem}

\begin{proof}
By Theorem~\ref{theorem:tree}, the number of iterations of the while loop is $ O (c \log \| w \|_\infty) $ with probability $ 1 - n^{-c} $.
Observe that Line~\ref{line:compute primal tree solution}, which calls the algorithm for computing a primal tree solution, is the only place in the algorithm, where communication between the nodes is necessary (if the primal tree solution $ T $ and the dual solution~$ y $ are made global knowledge in $ O (1) $ rounds after each call).
Therefore, the round complexity of computing an approximate SSSP tree exceeds the round complexity of computing a primal tree solution to the transshipment problem (as stated in Theorem~\ref{theorem:BCC transshipment_asym}) by a factor of $ O (c \log \| w \|_\infty) $, which is $ O (c \log n) $ as $ \| w \|_\infty $ is polynomially bounded.
\end{proof}

As discussed at the beginning of Section~\ref{sec:transshipment} this result extends to graphs with arbitrary non-negative arc weights at the cost of mild overheads.
Note that in the literature the approximate SSSP problem has also been studied in a relaxation of the Broadcast Congested Clique that does not require the weights to be polynomially bounded by increaseing the message size to $ \Theta (\log n + \log \| w \|_\infty) $~\cite{ElkinN16}.
By a reduction of Klein and Subramanian~\cite{KleinS97}, that has been adapted to the Broadcast Congested Clique by Elkin and Neiman~\cite{ElkinN16}, we still may assume without loss of generality that our approximate SSSP algorithm is called on a graph with largest arc weight $ \| w \|_\infty = O (n \epsilon^{-1}) $.
Furthermore, we may assume that $ \epsilon^{-1} = O (n) $ as our running time anyway depends on $ \epsilon^{-1} $ and thus in the regime $ \epsilon^{-1} = \Omega (n) $ we might as well compute a solution with the Bellman-Ford algorithm, which takes $ O (n) $ rounds.
In this way we arrive at polynomially bounded arc weights again.

\subsection{Broadcast CONGEST Model}\label{sec:congest}

\subparagraph*{Model}
The broadcast CONGEST model differs from the broadcast congested clique in that communication is restricted to edges that are present in the input graph. That is, node $v$ receives the messages sent by node $u$ if and only if $\{u, v\}\in \bar{E}$. All other aspects of the model are identical to the broadcast congested clique. We stress that this restriction has significant impact, however: Denoting the hop diameter\footnote{That is, the diameter of the unweighted graph $G=(V,E)$.} of the input graph by $D$, it is straightforward to show that $\Omega(D)$ rounds are necessary to solve the transshipment problem. Moreover, it has been established that $\Omega(\sqrt{n}/\log n)$ rounds are required even on graphs with $D\in O(\log n)$~\cite{DasSarmaHKKNPPW12}. Both of these bounds apply to randomized approximation algorithms (unless the approximation ratio is not polynomially bounded in $n$).

\subparagraph*{Solving the Transshipment Problem}
A straightforward implementation of our algorithm in this model simply simulates the broadcast congested clique. A folklore method to simulate (global) broadcast is to use ``pipelining'' on a breadth-first-search (BFS) tree.
\begin{lemma}[cf.~\cite{Peleg:book}]\label{lemma:broadcast}
Suppose each $v\in V$ holds $m_v\in \ZZ_{\ge 0}$ messages of $O(\log n)$ bits each, for a total of $M= \sum_{v\in V} m_v$ strings. Then all nodes in the graph can receive these $M$ messages within $O(M+D)$ rounds.
\end{lemma}
\begin{proof}[Proof Sketch]
It is easy to construct a BFS tree in $O(D)$ rounds (rooted at, e.g., the node with smallest identifier) and obtain an upper bound $d\in [D,2D]$ on the hop diameter by determining the depth of the tree and multiplying it by $2$. By a convergecast, we can determine $|M|$, where each node in the tree determines the total number of messages in its subtree. We define a total order on the messages via lexicographical order on node identifier/message pairs.\footnote{W.l.o.g., we assume identifier/message pairs to be different.} Then nodes flood their messages through the tree, where a flooding operation for a message may be delayed by those for other messages which are smaller w.r.t.\ the total order on the messages. On each path, a flooding operation may be delayed once for each flooding operation for a smaller message. Hence, this operation completes within $O(D+|M|)$ rounds, and using the knowledge on $d$ and~$M$, nodes can safely terminate within $O(d+|M|)=O(D+|M|)$ rounds.
\end{proof}
We obtain the following corollary to Theorem~\ref{theorem:BCC transshipment_asym}.
\begin{corollary}\label{coro:bcc}
Given any $ 0 < \varepsilon \leq 1 $, in the Broadcast CONGEST model a $(1+\varepsilon)$-approximate primal-dual pair of solutions to the transshipment problem in bidirected graphs with positive arc weights can be computed deterministically in $ \widetilde  O (\lambda^2 n \varepsilon^{-2}) $ rounds.
\end{corollary}
\begin{proof}
Simulate each round on the broadcast congested clique using Lemma~\ref{lemma:broadcast}, i.e., with parameters $|M|=n$ and $D\leq n$. Applying Theorem~\ref{theorem:BCC transshipment_asym}, the claim follows.
\end{proof}

\subparagraph*{Special Case: Single-Source Shortest Paths}
The near-linear running time bound of Corollary~\ref{coro:bcc} is far from the $\widetilde{\Omega}(\sqrt{n}+D)$ lower bound, which also applies to the single-source shortest path problem. However, for this problem there is an efficient reduction to a smaller \emph{skeleton graph}, implying that we can match the lower bound up to polylogarithmic factors. The skeleton graph is given as an overlay network on a subset $V'\subseteq V$ of the nodes, where each node in $V'$ learns its incident arc and their weights.

\begin{theorem}[\cite{HenzingerKN16}]\label{theorem:reduction}
Given any $ 0 < \varepsilon \leq 1 $ and a source node $ s $, in the Broadcast CONGEST model an overlay network $G'=(V',E')$ of an undirected graph $ G $ with positive integer edge weights that has edge weights $w':E'\to \{1, \ldots, O(n \| w \|_\infty) \}$ together with some additional information for every node satisfying the following properties with high probability can be computed in $\widetilde{O}(\sqrt{n} \epsilon^{-1} \log^2 \| w \|_\infty + D)$ rounds:
\begin{itemize}
  \item $|V'|= \widetilde{O}(\sqrt{n} \varepsilon^{-1} \log \| w \|_\infty) $ and $s\in V'$.
  \item For every node $ u \in V $, as soon as $ u $ receives, for every $ v \in V' $, a $ (1 + \tfrac{\varepsilon}{3}) $-approximation of the distance from $ s $ to $ v $ in $ G'$, it can infer a $ (1 + \epsilon) $-approximation of the distance from $ s $ to $ u $ in $ G $ without any additional communication.
\end{itemize}
\end{theorem}

To see why this reduction to the skeleton graph also works in bidirected graphs with asymmetric arc weights at the expense of a factor of $ \lambda $, let us briefly review the construction of Henzinger et al.~\cite{HenzingerKN16}.
In the first step, a set of nodes $ V' $ is computed such that every path $ \pi $ consisting of $ \sqrt{n} $ edges contains a node $ u $ such that the $h^*$-hop distance from $ u $ to $ V' $ is at most $ \tfrac{\epsilon}{10} $ times the weight of~$ \pi $; the $h^*$-hop distance is the length of the shortest path with at most $h^*$ edges for some $h^* = \widetilde O (\sqrt{n}) $ chosen by the algorithm.
In the second step, the $ (1 + \tfrac{\epsilon}{10}) $-approximate $k$-hop distance from each node of~$ V' $ to each node of $ V $ is computed where $ k = \widetilde O (\sqrt{n}) $.
The edge weights of the overlay network are set to the approximate $k$-hop distances between nodes in~$ V' $.
The remaining approximate distances computed in this step are also stored and used in the end to determine a $ (1 + \epsilon) $-approximation of the distance from $ s $ for each node.
Furthermore, the construction guarantees that an implicit $ (1 + \epsilon) $-approximate shortest path tree can be obtained if every node $ v $ chooses the neighbor $ u $ as its parent that minimizes the sum of the approximate distance from $ s $ to~$ u $ and the weight of the arc from $ u $ to~$ v $.

In bidirected graphs, we simply run the first step of the algorithm with $ \epsilon' = \tfrac{\epsilon}{\lambda 10} $ on $ G^+ $, the graph in which every bidirected edge $ e $ with edge weights $ w^+_e $ and $ w^-_e $ for the forward and backward arc, respectively, is replaced by an undirected edge of weight $ w^+_e $ (recall that $ w^+_e \geq w^-_e $).
Observe that for every path, the total weight in~$ G^+ $ exceeds its total weight in $ G $ by a factor of at most~$ \lambda $.
Let $ \pi $ be a path consisting of $ \sqrt{n} $ arcs.
The overlay network construction guarantees that $ \pi $ contains a node~$ u $ such that the $h$-hop distance from $ u $ to $ V' $ in $ G^+ $ is at most $ \epsilon' $ times the weight of~$ \pi $ in~$ G^+ $.
As the $h$-hop distance from $ u $ to $ V' $ in $ G $ is at most the $h$-hop distance from $ u $ to $ V' $ in $ G^+ $ and the weight of $ \pi $ in $ G^+ $ exceeds its weight in $ G $ by a factor of at most~$ \lambda $, we get the following:
the $h$-hop distance from $ u $ to $ V' $ in $ G^+ $ is at most $ \epsilon' \lambda = \tfrac{\epsilon}{10} $ times the weight of~$ \pi $ in~$ G $.
We can now proceed with the second step of the overlay network construction, which is based on a distance approximation technique that pipelines multiple (single-source) Moore-Bellman-Ford instances that also works in directed graphs and thus does not need to be adapted.
In this way, we can construct an overlay network with $\widetilde{O}(\lambda \sqrt{n} \varepsilon^{-1} \log \| w \|_\infty) $ nodes in $\widetilde{O}(\lambda \sqrt{n} \epsilon^{-1} \log^2 \| w \|_\infty + D)$ rounds.
Observe that the derived overlay network is guaranteed to have ratio at most $(1 + \epsilon) \lambda = O (\lambda) $ between forward and backward arcs.
We can thus apply the previous simulation approach as follows.

\begin{corollary}\label{coro:reduction_asym}
Given any $ 0 < \varepsilon \leq 1 $, in the Broadcast CONGEST model a $(1+\varepsilon)$-approximation to the single-source shortest path problem in bidirected graphs with positive polynomially bounded integer arc weights can be computed in $ \widetilde O (c (\sqrt{n} + D) \log \| w \|_\infty \lambda^2 \varepsilon^{-3}) = \widetilde O (c (\sqrt{n} + D) \lambda^2 \varepsilon^{-3}) $ rounds with probability at least $ 1 - n^{-c} $ for every $c \geq 1 $.
\end{corollary}
\begin{proof}
We apply Theorem~\ref{theorem:reduction}. Subsequently, we use Lemma~\ref{lemma:broadcast} to simulate rounds of the Broadcast Congested Clique on the overlay network, taking $ O (|V'| + D) = \widetilde{O}(\sqrt{n} \varepsilon^{-1} \log \| w \|_\infty + D) $ rounds per simulated round. Using Theroem~\ref{theorem:BCC sssp_asym} for $\varepsilon' = \tfrac{\varepsilon}{3} $, we obtain a $(1+\varepsilon')$-approximation to the distances from $ s $ to each $v\in V'$ in the overlay network. After broadcasting these distances using Lemma~\ref{lemma:broadcast} again, all nodes can locally compute a $(1+\varepsilon)$-approximation of their distance from~$s$. The total running time is $\widetilde{O}((\sqrt{n} \varepsilon^{-1} \log \| w \|_\infty + D) c \lambda^2 \varepsilon^{-2} + \sqrt{n} \epsilon^{-1} \log^2 \| w \|_\infty ) = \widetilde O (c (\sqrt{n} + D) \lambda^2 \varepsilon^{-3} \log^2 \| w \|_\infty) $, which is $ \widetilde O (c (\sqrt{n} + D) \lambda^2 \varepsilon^{-3}) $ since $ \| w \|_\infty $ is polynomially bounded.
\end{proof}
As discussed at the beginning of Section~\ref{sec:transshipment} this result extends to graphs with arbitrary non-negative arc weights at the cost of mild overheads.

If a Monte-Carlo guarantee suffices, then we can get a better dependence on $ \epsilon $ by using the overlay network construction of Nanongkai~\cite{Nanongkai14}.
\begin{theorem}[\cite{Nanongkai14}]\label{theorem:randomized overlay}
Given any $ 0 < \varepsilon \leq 1 $ and any $ 1 \leq h \leq n $, in the Broadcast CONGEST model an overlay network $G'=(V',E')$  of a bidirected graph $ G $ with positive integer arc weights that has arc weights $w':E'\to \{1, \ldots, O(n \| w \|_\infty) \}$ together with some additional information for every node satisfying the following properties with high probability can be computed in $\widetilde{O}((\tfrac{h}{\epsilon} + \tfrac{n}{h}) \log \| w \|_\infty + D)$ rounds:
\begin{itemize}
  \item $|V'|= \widetilde{O}(\tfrac{n}{h}) $ and $s\in V'$.
  \item For every node $ u \in V $, as soon as $ u $ receives, for every $ v \in V' $, a $ (1 + \tfrac{\varepsilon}{3}) $-approximation of the distance from $ s $ to $ v $ in $ G'$, it can infer a $ (1 + \epsilon) $-approximation of the distance from $ s $ to $ u $ in $ G $ without any additional communication. This approximate distances are correct with high probability.
\end{itemize}
\end{theorem}

Note that this overlay network construction readily works on directed graphs.
By setting $ h = \lambda \sqrt{n} \epsilon^{-1/2} $, we obtain an overlay network with $ |V'| = \widetilde O (\sqrt{\epsilon n} \lambda^{-1}) $ that can be constructed in $ \widetilde O (\lambda \sqrt{n} \varepsilon^{-3/2} \log \| w \|_\infty + D) $ rounds.
Using the same simulation approach as before, we get the following result.
\begin{corollary}\label{coro:reduction_asym Monte Carlo}
Given any $ 0 < \varepsilon \leq 1 $, in the Broadcast CONGEST model a $(1+\varepsilon)$-approximation to the single-source shortest path problem in bidirected graphs with positive polynomially bounded integer arc weights that is correct with high probability can be computed in $ \widetilde O ((\sqrt{n} + D) \lambda \varepsilon^{-3/2} \log \| w \|_\infty) = \widetilde O ((\sqrt{n} + D) \lambda \varepsilon^{-3/2}) $ rounds.
\end{corollary}

If approximating the distance between two node suffices, then both of these nodes can be made part of the overlay network and we only need to compute the approximate distance between these two nodes on the overlay network.

\begin{theorem}[Implicit in \cite{HenzingerKN16}]
Given any $ 0 < \varepsilon \leq 1 $ and two nodes $ s $ and~$ t $, in the Broadcast CONGEST model an overlay network $G'=(V',E')$ of an undirected graph $ G $ with positive integer edge weights that has edge weights $w':E'\to \{1, \ldots, O(n \| w \|_\infty) \}$ satisfying the following properties can be computed deterministically in $\widetilde{O}(\sqrt{n} \epsilon^{-1} \log^2 \| w \|_\infty + D)$ rounds:
\begin{itemize}
  \item $|V'|= \widetilde{O}(\sqrt{n} \varepsilon^{-1} \log \| w \|_\infty) $ and $\{s, t \} \subseteq V'$.
  \item A $ (1 + \tfrac{\varepsilon}{3}) $-approximation of the distance from $ s $ to~$ t $ in $ G'$ is a $ (1 + \epsilon) $-approximation of the distance from $ s $ to~$ t $ in $ G $
\end{itemize}
\end{theorem}

As we can compute the $s$-$t$ distance on the overlay network by running our transshipment algorithm with demand $ -1 $ on $ s $ and demand $ 1 $ on $ t $ (and demand $ 0 $ on all other nodes), we obtain a deterministic algorithm for approximating the $s$-$t$ distance.

\begin{corollary}
Given any $ 0 < \varepsilon \leq 1 $ and two nodes $ s $ and $ t $, in the Broadcast CONGEST model a $(1+\varepsilon)$-approximation to the distance from $ s $ to $ t $ in a bidirected graph with positive polynomially bounded integer arc weights can be computed deterministically in $ \widetilde O ((\sqrt{n} + D) \log \| w \|_\infty \lambda^2 \varepsilon^{-3}) = \widetilde O ((\sqrt{n} + D) \lambda^2 \varepsilon^{-3}) $ rounds.
\end{corollary}

%% file: streaming.tex
\subsection{Multipass Streaming}

\subparagraph*{Model}

In the \emph{Streaming} model the input graph is presented to the algorithm edge by edge as a ``stream'' without repetitions and the goal is to design algorithms that use as little space as possible. Space is measured in memory words, where we assume that each word consists of $ O (\log n) $ bits.
The weights are assumed to be polynomially bounded in~$n$, so that they can be encoded using $O(\log n)$ bits.
In the \emph{Multipass Streaming} model, the algorithm is allowed to make several such passes over the input stream and the goal is to design algorithms that need only a small number of passes (and again little space).
For graph algorithms, the usual assumption is that the edges of the input graph are presented to the algorithm in arbitrary order.
For bidirected graphs, we assume that both arcs of each edge are presented consecutively in the stream.

\subparagraph*{Implementing Algorithm~\nablaalg}
In the following we explain how to implement our gradient descent algorithm for approximating the shortest transshipment.

\begin{theorem}\label{theorem:streaming transshipment_asym}
Given any $ 0 < \varepsilon \leq 1 $, in the Multipass Streaming model a $(1+\varepsilon)$-approximate dual solution to the transshipment problem in bidirected graphs with positive arc weights can be computed deterministically in $ \widetilde O (\lambda^2 \varepsilon^{-2}) $ passes with a space of $ O (n \log n) $ words.
\end{theorem}

Note that we will not explicitly compute a primal solution $ x \in \mathbb{R}^{2 m} $ because storing this solution might take too much space, i.e., we will not implement Lines \ref{line:primal solution xt} and~\ref{line:primal solution x} of Algorithm~\ref{alg:gradient}.

\begin{proof}[Proof of Theorem~\ref{theorem:streaming transshipment_asym}]
We will use $ O (n \log {n}) $ space to store the following information.
\begin{itemize}
\item An $n$-dimensional vector $ b $ for the input demands.
\item An $ O (\log{n}) $-spanner $ S $ of $ G^- = (V, E, w^-) $ of size $ O (n \log n) $.
\item An $ n $-dimensional vector $ \pi $ for the dual solution maintained by the gradient descent algorithm.
\item An $ O (n \log n) $-dimensional vector $ f $ for the primal solution on the spanner.
\item An $ n $-dimensional vector $ h $ for the dual solution on the spanner.
\item Scalars $ m $, $ n $, $ \varepsilon $, $ \lambda $, $ \alpha $, $ \varepsilon' $, $ \beta $, and $ \delta $.
\end{itemize}
Additionally, we use ``temporary space'' of size $ O (n \log {n}) $ for implementing individual lines of the algorithm.
At any time, the total space will be $ O (n \log {n}) $.

We will show that each iteration of Algorithm~\ref{alg:gradient} can be implemented in a single pass and that the remainder of the algorithm can be implemented in $ O (\log n) $ passes.
By Theorem~\ref{theorem:asymmetric}, which bounds the number of iterations of Algorithm~\ref{alg:gradient}, the bound on the number of passes is then $ O (\lambda^2 (\varepsilon^{-2} + \log{n} + \log{\lambda}) \log^3{n}) $.
If $ \lambda $ is as large as polynomial in $ n $ we can solve the problem in a polynomial number of passes, i.e., $ O (\lambda) $ passes, by using any of the polynomial-time algorithms.
Therefore, this simplifies to $ \widetilde O (\lambda^2 \varepsilon^{-2}) $.

At the beginning of the algorithm, we first compute and store $n$, $m$, and $ \lambda $ in a single pass.
We then construct and store an $ \alpha $-spanner $ S $ of $ G^- = (V, E, w^-) $, where $ \alpha = 2 \lceil \log n\rceil - 1 $, with $ O (n \log{n}) $ edges in $O(\log n)$ passes with $ O (n \log{n}) $ temporary space using the algorithm of Corollary~\ref{coro:spanner streaming}.
The spanner can be used to implement the oracle calls in Lines \ref{line:initial oracle call} and~\ref{line:oracle call in iteration} using only internal computation as on the spanner the $ O(n \log n)$-dimensional primal solution and the $ n $-dimensional dual solution can be determined in $ O (n \log n) $ space by enumerating all potential solutions and checking for optimality.\footnote{Alternatively, we could internally run one of the classic linear-space algorithms for the transshipment problem.}
It remains to describe how to implement Lines \ref{line:set_beta} and~\ref{line:double_beta}, where the algorithm needs to compute $ (W^{-1}A^T\pi)_{\max} $, and Line~\ref{line:compute gradient}, where the algorithm computes $ \nabla\potfun_\beta(\pi) $.
It is clear that the rest of the algorithm can be implemented by performing only local computation.

To implement Lines \ref{line:set_beta}, \ref{line:compute gradient}, and~\ref{line:double_beta}, first observe that, using the definition $ s_e (\pi) := \tfrac{\pi_v - \pi_u}{w_e} $ for every arc $ e \in E $, both $ \sum_{e \in E} e^{\beta s_e (\pi)} $ and $ (W^{-1}A^T\pi)_{\max} = \max \{ s_e (\pi) : e = (u, v) \in E \} $ can be computed in a single pass with $ O (1) $ temporary space as summation and maximum are associative and commutative operators:
Before the pass, we create temporary variables $ \sumarcs $ and $ \maxarcs $, both initialized to $ 0 $.
During the pass, every time we read an arc $ e $ of weight $ w_e $ from the stream, we first compute $ s_e (\pi) $ and then update $ \sumarcs $ to $ \sumarcs + e^{\beta s_e (\pi)} $ and $ \maxarcs $ to $ \max \{ \maxarcs, s_e (\pi) \} $.
After the pass, we have $ \sumarcs = \sum_{e \in E} e^{\beta s_e (\pi)} $ and $ \maxarcs = (W^{-1}A^T\pi)_{\max} $.
Since each component of $ \nabla \potfun_\beta(\pi) $ is given by
\begin{align*}
\nabla \potfun_\beta(\pi)_v
  &= (A W^{-1} \nabla\softmax[\beta]{W^{-1}A^T\pi})_v \\
  &= \sum_{e = (u, v) \in E} \frac{e^{\beta s_e (\pi)}}{w_e \cdot \sum_{e' \in E} e^{\beta s_{e'} (\pi)}} - \sum_{e = (v, u) \in E} \frac{e^{\beta s_e (\pi)}}{w_e \cdot \sum_{e' \in E} e^{\beta s_{e'} (\pi)}} \, ,
\end{align*}
the same idea can be used to compute $ \nabla \potfun_\beta(\pi) $ in a single pass with $ O (1) $ temporary space, once $ \sumarcs = \sum_{e \in E} e^{\beta s_e (\pi)} $ is known.
\end{proof}

\subparagraph*{Implementing Algorithm~\treealgo}

In the following we explain how to implement our sampling-based algorithm to compute a primal tree solution for the approximate transshipment problem.

\begin{theorem}\label{theorem:streaming tree solution_asym}
In the Multipass Streaming model one can, with probability at least~$ \tfrac{1}{4} $, compute in $ \widetilde O (\lambda^2 \varepsilon^{-2}) $ passes with a space of $ O (\lambda n \log^2 {n}) $ words a $(1+\varepsilon)$-approximate primal-dual pair of solutions to the transshipment problem given a parameter $ 0 < \varepsilon \leq 1 $, a bidirected graph with positive polynomially bounded integer arc weights, and a demand vector~$ b $ such that $b_v\in \{0,1\}$ for all $v\in V\setminus \{s\}$ and with the additional property that the primal solution has non-zero flow only on the arcs of a tree.
\end{theorem}

\begin{proof}
Again, we do not want to explicitly store the primal solution $ x $ computed at the end of Algorithm~\nablaalg to save space.
We will be able to retrieve $ x_e $ for every arc $ e $ ``on the fly'' upon reading $ e $ from the stream by modifying Algorithm~\nablaalg to additionally return $ f' $, $ \pi $, $ z $, and $ \beta $.

We will use space $ O (n \lambda \log n \log \| w \|_\infty \rceil) $, which is $ O (n \lambda \log^2 {n}) $ since $ \| w \|_\infty $ is polynomially bounded, to store the following information:
\begin{itemize}
\item An arc set $ \bigcup_{v, k} F_{v, k} $ of size $ O (\lambda \log n) $ for every $ v \in V $ and every $ k \in \{1, \ldots,  \lceil \log \| w \|_\infty \rceil \} $.
\item An $ O (\log{n}) $-spanner $ S $ of $ G $ of size $ O (n \log n) $.
\item An $ (n - 1) $-dimensional vector $ t $ for the primal solution on $ S $.
\item An $ n $-dimensional vector $ y' $ for the dual solution on $ S $.
\item A scalar to store $ \lceil \log \| w \|_\infty \rceil $
\item An $ O (n \log n) $-dimensional vector $ f' $ as computed by the call to the transshipment algorithm.
\item An $ n $-dimensional vector $ \pi $ as computed by the call to the transshipment algorithm.
\item Scalars $ \beta $ and $ z $ as computed by the call to the transshipment algorithm.
\end{itemize}

First, observe that Lines \ref{line:construct auxiliary graph} to~\ref{line:last line of primal tree} of Algorithm~\treealgo can be performed without performing additional passes.
Furthermore, by Theorem~\ref{theorem:streaming transshipment_asym}, the call to Algorithm~\nablaalg in Line~\ref{line:refined approximation to transshipment} takes $ O (\lambda^2 (\varepsilon^{-2} + \log{n} + \log{\lambda}) \log^3{n}) $ passes and space $ O (n \log {n}) $.
As before, we can implement Line~\ref{line:construct spanner S} by constructing and storing an $ \alpha $-spanner $ S $, where $ \alpha = O (\log n) $, with $ O (n \log{n}) $ edges in $O(\log n)$ passes with $ O (n \log{n}) $ temporary space using the algorithm of Corollary~\ref{coro:spanner streaming}.

It remains to explain the implementation of the sampling process in Lines~\ref{line:begin sampling} to~\ref{line:end sampling}.
Before the sampling starts, we compute $ \lceil \log \| w \|_\infty \rceil $ in a single pass.
We first explain how to compute $ x_e $ (the primal solution for $ e $ in the call to Algorithm~\nablaalg in Line~\ref{line:refined approximation to transshipment}) upon reading $ e $ from the stream.
Remember that
\begin{equation*}
\xt_e = (W^{-1}\nabla \softmax[\beta]{W^{-1}A^T\pi})_e = \frac{e^{\beta s_e (\pi)}}{w_e \sum_{e' \in E} e^{\beta s_{e'} (\pi)}} \, ,
\end{equation*}
and $ x_e = \tfrac{\xt_e + f'_e}{z} $.
We can therefore first compute $ \sumarcs = \sum_{e \in E} e^{\beta s_e (\pi)} $ in an initial pass and later on instantly compute $ x_e $ using the formulas above.

We will perform the sampling in parallel for each $ v \in V $ and for each $ k \in \{ 1, \ldots \lceil \log \| w \|_\infty \rceil \} $ in $ O (\log (n \| w \|_\infty)) $ passes, which is $ O (\log n) $ since $ \| w \| $ is polynomially bounded.
Every time we read an arc $ e $ from the stream, we compute the unique $ k $ for which $ w_e \in [2^{k-1}, 2^k) $.
We spend the first pass to compute $ \sum_{e \in E_{v,k}} x_e $  for each $ v \in V $ and each $ k \in \{ 1, \ldots \lceil \log \| w \|_\infty \rceil  \} $.
We spend the remaining $ \lceil \log (4 n\lceil \log \| w \|_\infty \rceil ) \rceil $ passes to repeatedly sample a set $ F_{v, k} $ for each $ v \in V $ and each $ k \in \{ 1, \ldots \lceil \log \| w \|_\infty \rceil \} $.
Note that the sampling probability $ p_e $ can be computed instantly upon reading $ e $ from the stream.
\end{proof}

\subparagraph*{Implementing Algorithm~\ssspalgo}

Finally, we explain how to implement our adaptive-demands algorithm for computing an approximate single-source shortest path tree.

\begin{theorem}\label{theorem:streaming SSSP_asym}
Given any $ 0 < \varepsilon \leq 1 $, in the Multipass Streaming model a $(1+\varepsilon)$-approximation to the single-source shortest path problem in bidirected graphs with positive polynomially bounded integer arc weights can be computed in $ \widetilde O (c \lambda^2 \varepsilon^{-2}) $ passes with probability $ 1 - n^{-c} $ for every $ c \geq 1 $ and with a space of $ O (n \lambda \log^2 {n} )) $.
\end{theorem}

\begin{proof}
To implement Algorithm~\ssspalgo, we will, in addition to the streaming implementation of Algorithm~\treealgo, use space $ O (n \log n \log \| w \|_\infty) $, which is $ O (n \log^2 n) $ since $ \| w \|_\infty $ is polynomially bounded, to store the following information:
\begin{itemize}
\item An arc set $ E' $ of size $ O (n \log \| w \|_\infty) $.
\item A tree $ T $ consisting of $ n - 1 $ arcs.
\item An $n$-dimensional vector $ y $.
\item An $n$-dimensional vector $ b $.
\end{itemize}

By Theorem~\ref{theorem:tree}, the number of iterations of the while-loop is $ O (c \log \| w \|_\infty) $ with probability $ 1 - n^{-c} $.
Observe that Line~\ref{line:compute primal tree solution}, which calls the algorithm for computing a primal tree solution, is the only place in the algorithm where reading from the stream is necessary.
Thus, the number of passes for implementing Algorithm~\ssspalgo exceeds those for computing a primal tree solution (Theorem~\ref{theorem:streaming tree solution_asym}) by a factor of $ O (c \log \| w \|_\infty) $ with probability $ 1 - n^{-c} $.
Furthermore, the additional space we need is dominated by the size of $ E' $ which is $ O (n \log \| w \|_\infty) $ (as in each iteration a tree of size $ n - 1 $ is added to $ E' $).
\end{proof}

As discussed at the beginning of Section~\ref{sec:transshipment} this result extends to graphs with arbitrary non-negative arc weights at the cost of mild overheads.
Note that in the literature the approximate SSSP problem has also been studied in a relaxation of the Multipass Streaming model that does not require the weights to be polynomially bounded by increasing the word size to $ \Theta (\log n + \log \| w \|_\infty) $~\cite{ElkinN16}.
By a reduction of Klein and Subramanian~\cite{KleinS97}, that has been adapted to the Multipass Streaming model by Elkin and Neiman~\cite{ElkinN16}, we still may assume without loss of generality that our approximate SSSP algorithm is called on a graph with largest arc weight $ \| w \|_\infty = O (n \epsilon^{-1}) $.
Furthermore, we may assume that $ \epsilon^{-1} = O (n) $ as our running time anyway depends on $ \epsilon^{-1} $ and thus in the regime $ \epsilon^{-1} = \Omega (n) $ we might as well compute a solution with the Bellman-Ford algorithm, which takes $ O (n) $ passes and $ O (n) $ space.
In this way we arrive at polynomially bounded arc weights again.

%% file: spanner.tex
\section{Deterministic Spanner Computation in Congested Clique and Multipass Streaming Model}\label{sec:spanner}

For $k\in \ZZ_{>0}$, a simple and elegant randomized algorithm computing a $(2k-1)$-spanner with $O (kn^{1+1/k})$ edges in expectation was given by Baswana and Sen~\cite{BaswanaS07}. For the sake of completeness, we restate it here.

\vbox{\vspace*{1ex}\hrule\vspace{1mm}
\begin{enumerate}
\item Initially, each node is a singleton \emph{cluster}: $R_1:=\{\{v\}\mid v\in V\}$.
\item For $i=1,\ldots,k-1$ do:
\begin{enumerate}
\item Each cluster from $R_i$ is \emph{marked} independently with probability $n^{-1/k}$. $R_{i+1}$ is defined to be the set of clusters marked in phase $i$.
\item If $v$ is a node in an unmarked cluster:
\begin{enumerate}
\item Define $Q_v$ to be the set of edges that consists of the lightest edge from $v$ to each cluster in $R_i$ it is adjacent to.
\item If $v$ is not adjacent to any marked cluster, all edges in $Q_v$ are added to the spanner.
\item Otherwise, let $u$ be the closest neighbor of $v$ in a marked cluster. In this case, $v$ adds to the spanner the edge $\{v,u\}$ and all edges $\{v,w\}\in Q_v$ with $w_{(v,w)}<w_{(v,u)}$ (break ties by neighbor identifiers). Also, let $X$ be the cluster of $u$. Then $X:=X\cup \{v\}$, i.e., $v$ \emph{joins} the cluster of $u$.
\end{enumerate}
\end{enumerate}
\item Each $v\in V$ adds, for each $X\in R_k$ it is adjacent to, the lightest edge connecting it to $X$ to the spanner.
\end{enumerate}\smallskip
\hrule\vspace{2mm}
}

It is easy to see that the expected number of edges that the algorithm selects into the spanner is $O(kn^{1+1/k})$: In each iteration, each node $v$ sorts its incident clusters in order of ascending weight of the lightest edge to them and elects for each cluster, up to the first sampled one, the respective lightest edge into the spanner. Because this order is independent of the randomness used in this iteration, $v$ selects $O(n^{1/k})$ edges in expectation and $O(n^{1/k}\log n)$ edges with high probability.\footnote{That is, for any fixed constant choice of $c>0$, the number of selected edges is bounded by $O(n^{1/k}\log n)$ with probability at least $1-1/n^c$.} The same bound applies to the final step, as $|R_k|\in O(n^{1/k})$ in expectation and $|R_k|\in O(n^{1/k}\log n)$ with high probability. Moreover, this observation provides a straightforward derandomization of the algorithm applicable in our model: Instead of picking $R_{i+1}$ in iteration $i$ randomly, we consider the union $E_i$ over all nodes $v$ of the lightest $O(n^{1/k}\log n)$ edges in $Q_v$. By a union bound, with high probability we can select $R_{i+1}$ such that (i) $|R_{i+1}|\leq n^{-1/k}|R_i|$ and (ii) each node selects only $O(n^{1/k}\log n)$ edges into the spanner in this iteration. In particular, such a choice must exist, and it can be computed from $R_i$ and $E_i$ alone.
With this argument we deterministically obtain a spanner of size $ O (k n^{1+1/k} \log{n}) $.

\vbox{\vspace*{1ex}\hrule\vspace{1mm}
\begin{enumerate}
\item Initially, each node is a singleton \emph{cluster}: $R_1:=\{\{v\}\mid v\in V\}$.
\item For $i=1,\ldots,k-1$ do for each node $v$:
\begin{enumerate}
\item Define $Q_v$ to be the set of edges that consists of the lightest edge from $v$ to each cluster in $R_i$ it is adjacent to.
\item Broadcast the set $Q_v'$ of the lightest $O(n^{1/k}\log n)$ edges in $Q_v$.
\item For $w\in V$, denote by $X_w\in R_i$ the cluster so that $v\in X_w$. Locally compute $R_{i+1}\subseteq R_i$ such that (i) $|R_{i+1}|\leq n^{-1/k}|R_i|$ and (ii) for each $w\in V$ for which $Q_v'\neq Q_v$, it holds that $X_w\in R_{i+1}$ or $Q_v'$ contains an edge connecting to some $X\in R_{i+1}$.
\item Update clusters and add edges to the spanner as the original algorithm would, but for the computed choice of $R_{i+1}$.
\end{enumerate}
\item Each $v\in V$ adds, for each $X\in R_k$ it is adjacent to, the lightest edge connecting it to $X$.
\end{enumerate}\smallskip
\hrule\vspace{2mm}
}

A slightly stronger bound of $ O (k n^{1+1/k}) $ on the size of the spanner, matching the expected value from above up to constant factors, can be obtained in this framework by using the technique of deterministically finding \emph{early hitting sets} developed by Roditty, Thorup, and Zwick~\cite{RodittyTZ05}.

Note that, as argued above, the selection of $R_{i+1}$ in Step 2(c) is always possible and can be done deterministically, provided that $R_i$ is known. Because $R_1$ is simply the set of singletons and each node computes $R_{i+1}$ from $R_i$ and the $Q_v'$, this holds by induction. We arrive at the following result.
\begin{corollary}[of~\cite{BaswanaS07} and~\cite{RodittyTZ05}]\label{coro:deterministic spanner}
For $k\in \ZZ_{>0}$, in the broadcast congested clique a ${(2k-1)}$-spanner of size $ O (kn^{1+1/k}) $ can be deterministically computed and made known to all nodes in $O(kn^{1/k}\log n)$ rounds.
\end{corollary}
\begin{proof}
The bound of $\alpha=2k-1$ follows from the analysis in~\cite{BaswanaS07}, which does not depend on the choice of the $R_i$. For the round complexity, observe that all computations can be performed locally based on knowing $Q_v'$ for all nodes $v\in V$ in each iteration. As $|Q_v'|\in O(n^{1/k} \log n)$ for each iteration and each $v$, the claim follows.
\end{proof}

By similar arguments, we can also get an algorithm in the multipass streaming algorithm with comparable guarantees.
We remark that, apart from the property of being deterministic, this was already stated in~\cite{Baswana08} as a simple consequence of~\cite{BaswanaS07}.

\begin{corollary}[of~\cite{BaswanaS07} and~\cite{RodittyTZ05}]\label{coro:spanner streaming}
For $k\in \ZZ_{>0}$, in the multipass streaming model a ${(2k-1)}$-spanner of size $ O (kn^{1+1/k}) $ can be deterministically computed in $ k $ passes using $ O ((k + \log{n}) n^{1+1/k}) $ space.
\end{corollary}

%% file: references.bib
@inproceedings{Censor-HillelDK19,
  author    = {Keren Censor{-}Hillel and
               Michal Dory and
               Janne H. Korhonen and
               Dean Leitersdorf},
  title     = {Fast Approximate Shortest Paths in the Congested Clique},
  booktitle = {Proc.\ of the {ACM} Symposium on Principles of Distributed Computing (PODC)},
  pages     = {74--83},
  year      = {2019},
  doi       = {10.1145/3293611.3331633},
  archivePrefix = {arXiv},
  eprint    = {1903.05956}
}

@inproceedings{ForsterN18,
  author    = {Sebastian Forster and
               Danupon Nanongkai},
  title     = {A Faster Distributed Single-Source Shortest Paths Algorithm},
  booktitle = {Proc.\ of the {IEEE} Symposium on Foundations of Computer Science (FOCS)},
  pages     = {686--697},
  year      = {2018},
  doi       = {10.1109/FOCS.2018.00071},
  archivePrefix = {arXiv},
  eprint    = {1711.01364}
}

@inproceedings{GhaffariL18,
  author    = {Mohsen Ghaffari and
               Jason Li},
  title     = {Improved distributed algorithms for exact shortest paths},
  booktitle = {Proc.\ of the {ACM} {SIGACT} Symposium on Theory of Computing (STOC)},
  pages     = {431--444},
  year      = {2018},
  doi       = {10.1145/3188745.3188948},
  archivePrefix = {arXiv},
  eprint    = {1712.09121},
}

@Book{Peleg:book,
  author    = {David Peleg},
  title     = {Distributed Computing: A Locality-Sensitive Approach},
  publisher = {SIAM},
  address   = {Philadelphia, PA},
  year      = {2000}
}

@inproceedings{HenzingerKN16,
  author    = {Monika Henzinger and
               Sebastian Krinninger and
               Danupon Nanongkai},
  title     = {A deterministic almost-tight distributed algorithm for approximating
               single-source shortest paths},
  booktitle = {Proc.\ of the {ACM} {SIGACT} Symposium on Theory of Computing (STOC)},
  pages     = {489--498},
  year      = {2016},
  doi       = {10.1145/2897518.2897638},
  archivePrefix = {arXiv},
  eprint    = {1504.07056}
}

@inproceedings{Sherman13,
  author    = {Jonah Sherman},
  title     = {Nearly Maximum Flows in Nearly Linear Time},
  booktitle = {Proc.\ of the {IEEE} Symposium on Foundations of Computer Science (FOCS)},
  pages     = {263--269},
  year      = {2013},
  doi       = {10.1109/FOCS.2013.36},
  archivePrefix = {arXiv},
  eprint    = {1304.2077}
}

@article{FredmanT87,
  author    = {Michael L. Fredman and
               Robert Endre Tarjan},
  title     = {Fibonacci heaps and their uses in improved network optimization algorithms},
  journal   = {Journal of the {ACM}},
  volume    = {34},
  number    = {3},
  pages     = {596--615},
  year      = {1987},
  doi       = {10.1145/28869.28874},
  note      = {Announced at FOCS '84}
}

@article{Thorup99,
  author    = {Mikkel Thorup},
  title     = {Undirected Single-Source Shortest Paths with Positive Integer Weights
               in Linear Time},
  journal   = {Journal of the {ACM}},
  volume    = {46},
  number    = {3},
  pages     = {362--394},
  year      = {1999},
  doi       = {10.1145/316542.316548},
  note      = {Announced at FOCS '97}
}

@article{Cohen00,
  author    = {Edith Cohen},
  title     = {Polylog-Time and Near-Linear Work Approximation Scheme for Undirected Shortest Paths},
  journal   = {Journal of the ACM},
  volume    = {47},
  number    = {1},
  year      = {2000},
  pages     = {132--166},
  doi       = {10.1145/331605.331610},
  note      = {Announced at STOC '94}
}

@inproceedings{MillerPVX15,
  author    = {Gary L. Miller and
               Richard Peng and
               Adrian Vladu and
               Shen Chen Xu},
  title     = {Improved Parallel Algorithms for Spanners and Hopsets},
  booktitle = {Proc.\ of the {ACM} Symposium on Parallelism in Algorithms and Architectures (SPAA)},
  pages     = {192--201},
  year      = {2015},
  doi       = {10.1145/2755573.2755574},
  archivePrefix = {arXiv},
  eprint    = {1309.3545}
}

@inproceedings{ElkinN16,
  author    = {Michael Elkin and
               Ofer Neiman},
  title     = {Hopsets with Constant Hopbound, and Applications to Approximate Shortest
               Paths},
  booktitle = {Proc.\ of the Symposium on Foundations of Computer Science (FOCS)},
  pages     = {128--137},
  year      = {2016},
  doi       = {10.1109/FOCS.2016.22},
  archivePrefix = {arXiv},
  eprint    = {1605.04538}
}

@article{BaswanaS07,
  author    = {Surender Baswana and
               Sandeep Sen},
  title     = {A Simple and Linear Time Randomized Algorithm for Computing Sparse Spanners in Weighted Graphs},
  journal   = {Random Structures \& Algorithms},
  volume    = {30},
  number    = {4},
  pages     = {532--563},
  year      = {2007},
  doi       = {10.1002/rsa.20130},
  note      = {Announced at ICALP '03}
}

@article{DasSarmaHKKNPPW12,
  author    = {Atish {Das Sarma} and
               Stephan Holzer and
               Liah Kor and
               Amos Korman and
               Danupon Nanongkai and
               Gopal Pandurangan and
               David Peleg and
               Roger Wattenhofer},
  title     = {Distributed Verification and Hardness of Distributed Approximation},
  journal   = {{SIAM} Journal on Computing},
  volume    = {41},
  number    = {5},
  year      = {2012},
  pages     = {1235--1265},
  doi       = {10.1137/11085178X},
  note	    = {Announced at STOC '11},
  archivePrefix = {arXiv},
  eprint    = {1011.3049}
}

@article{Baswana08,
  author    = {Surender Baswana},
  title     = {Streaming algorithm for graph spanners - single pass and constant
               processing time per edge},
  journal   = {Information Processing Letters},
  volume    = {106},
  number    = {3},
  pages     = {110--114},
  year      = {2008},
  doi       = {10.1016/j.ipl.2007.11.001}
}

@inproceedings{Sherman17,
  author    = {Jonah Sherman},
  title     = {Generalized Preconditioning and Network Flow Problems},
  booktitle = {Proc.\ of the {ACM-SIAM} Symposium on Discrete Algorithms (SODA)},
  year      = {2017},
  pages     = {772--780},
  doi       = {10.1137/1.9781611974782.49},
  archivePrefix = {arXiv},
  eprint    = {1606.07425}
}

@article{Orlin93,
  author    = {James B. Orlin},
  title     = {A Faster Strongly Polynomial Minimum Cost Flow Algorithm},
  journal   = {Operations Research},
  volume    = {41},
  number    = {2},
  pages     = {338--350},
  year      = {1993},
  doi       = {10.1287/opre.41.2.338},
  note	    = {Announced at STOC '88}
}

@article{EdmondsK72,
  author    = {Jack Edmonds and
               Richard M. Karp},
  title     = {Theoretical Improvements in Algorithmic Efficiency for Network Flow
               Problems},
  journal   = {Journal of the {ACM}},
  volume    = {19},
  number    = {2},
  pages     = {248--264},
  year      = {1972},
  doi       = {10.1145/321694.321699}
}

@article{Bellman58,
  author    = {Richard Bellman},
  title     = {On a Routing Problem},
  journal   = {Quarterly of Applied Mathematics},
  volume    = {16},
  number    = {1},
  year      = {1958},
  pages     = {87--90}
}

@techreport{Ford56,
  author      = {Lester R. Ford},
  title       = {Network Flow Theory},
  institution = {The RAND Corporation},
  year        = {1956},
  number      = {P-923}
}

@inproceedings{LenzenPS13,
  author    = {Christoph Lenzen and
               Boaz Patt-Shamir},
  title     = {Fast Routing Table Construction Using Small Messages},
  booktitle = {Proc.\ of the {ACM} Symposium on Theory of Computing (STOC)},
  year      = {2013},
  pages     = {381--390},
  doi       = {10.1145/2488608.2488656},
  archivePrefix = {arXiv},
  eprint    = {1210.5774}
}

@inproceedings{Nanongkai14,
  author    = {Danupon Nanongkai},
  title     = {Distributed Approximation Algorithms for Weighted Shortest Paths},
  booktitle = {Proc.\ of the {ACM} Symposium on Theory of Computing (STOC)},
  year      = {2014},
  pages     = {565--573},
  doi       = {10.1145/2591796.2591850},
  archivePrefix = {arXiv},
  eprint    = {1403.5171}
}

@article{Elkin06,
  author    = {Michael Elkin},
  title     = {An Unconditional Lower Bound on the Time-Approximation Trade-off
               for the Distributed Minimum Spanning Tree Problem},
  journal   = {{SIAM} Journal on Computing},
  volume    = {36},
  number    = {2},
  year      = {2006},
  pages     = {433--456},
  doi       = {10.1137/S0097539704441058},
  note      = {Announced at STOC '04}
}

@article{GuruswamiO13,
  author    = {Venkatesan Guruswami and
               Krzysztof Onak},
  title     = {Superlinear Lower Bounds for Multipass Graph Processing},
  journal   = {Algorithmica},
  volume    = {76},
  number    = {3},
  pages     = {654--683},
  year      = {2016},
  doi       = {10.1007/s00453-016-0138-7},
  archivePrefix = {arXiv},
  eprint    = {1212.6925},
  note      = {Announced at CCC '13}
}

@article{FeigenbaumKMSZ05,
  author    = {Joan Feigenbaum and
	       Sampath Kannan and
	       Andrew McGregor and
	       Siddharth Suri and
	       Jian Zhang},
  title     = {On graph problems in a semi-streaming model},
  journal   = {Theoretical Computer Science},
  volume    = {348},
  number    = {2-3},
  pages     = {207--216},
  year      = {2005},
  doi       = {10.1016/j.tcs.2005.09.013},
  note	    = {Announced at ICALP '04}
}

@article{FeigenbaumKMSZ08,
  author    = {Joan Feigenbaum and
               Sampath Kannan and
               Andrew McGregor and
               Siddharth Suri and
               Jian Zhang},
  title     = {Graph Distances in the Data-Stream Model},
  journal   = {{SIAM} Journal on Computing},
  volume    = {38},
  number    = {5},
  pages     = {1709--1727},
  year      = {2008},
  doi       = {10.1137/070683155},
  note	    = {Announced at SODA '05}
}

@article{Elkin11,
  author    = {Michael Elkin},
  title     = {Streaming and Fully Dynamic Centralized Algorithms for Constructing
               and Maintaining Sparse Spanners},
  journal   = {{ACM} Transactions on Algorithms},
  volume    = {7},
  number    = {2},
  pages     = {20:1--20:17},
  year      = {2011},
  doi       = {10.1145/1921659.1921666},
  note      = {Announced at ICALP '07}
}

@article{ElkinZ06,
  author    = {Michael Elkin and
               Jian Zhang},
  title     = {Efficient algorithms for constructing {($1+\epsilon$, $\beta$)}-spanners in
               the distributed and streaming models},
  journal   = {Distributed Computing},
  volume    = {18},
  number    = {5},
  pages     = {375--385},
  year      = {2006},
  doi       = {10.1007/s00446-005-0147-2}
}

@inproceedings{LeeS14,
  author    = {Yin Tat Lee and
               Aaron Sidford},
  title     = {Path Finding Methods for Linear Programming: Solving Linear Programs
               in {$ \tilde O (\sqrt{} rank) $} Iterations and Faster Algorithms for Maximum Flow},
  booktitle = {Proc.\ of the {IEEE} Symposium on Foundations of Computer Science (FOCS)},
  pages     = {424--433},
  year      = {2014},
  doi       = {10.1109/FOCS.2014.52}
}

@article{CensorHillelKKLPS15,
  author    = {Keren Censor{-}Hillel and
               Petteri Kaski and
               Janne H. Korhonen and
               Christoph Lenzen and
               Ami Paz and
               Jukka Suomela},
  title     = {Algebraic methods in the congested clique},
  journal   = {Distributed Computing},
  volume    = {32},
  number    = {6},
  pages     = {461--478},
  year      = {2019},
  doi       = {10.1007/s00446-016-0270-2},
  archivePrefix = {arXiv},
  eprint    = {1503.04963},
  note      = {Announced at PODC '15}
}

@article{LotkerPPP03,
  author    = {Zvi Lotker and
               Boaz {Patt-Shamir} and
               Elan Pavlov and
               David Peleg},
  title     = {Minimum-Weight Spanning Tree Construction in {$ O(\log \log n) $} Communication Rounds},
  journal   = {{SIAM} Journal on Computing},
  volume    = {35},
  number    = {1},
  pages     = {120--131},
  year      = {2005},
  doi       = {10.1137/S0097539704441848},
  note      = {Announced at SPAA '03}
}

@inproceedings{ChristianoKMST11,
  author    = {Paul Christiano and
               Jonathan A. Kelner and
               Aleksander M\k{a}dry and
               Daniel A. Spielman and
               Shang{-}Hua Teng},
  title     = {Electrical flows, {Laplacian} systems, and faster approximation of maximum
               flow in undirected graphs},
  booktitle = {Proc.\ of the {ACM} Symposium on Theory of Computing (STOC)},
  pages     = {273--282},
  year      = {2011},
  doi       = {10.1145/1993636.1993674},
  archivePrefix = {arXiv},
  eprint    = {1010.2921}
}

@inproceedings{Madry13,
  shorthand = {M\k{a}d13},
  author    = {Aleksander M\k{a}dry},
  title     = {Navigating Central Path with Electrical Flows: From Flows to Matchings,
               and Back},
  booktitle = {Proc.\ of the {IEEE} Symposium on Foundations of Computer Science (FOCS)},
  pages     = {253--262},
  year      = {2013},
  doi       = {10.1109/FOCS.2013.35},
  archivePrefix = {arXiv},
  eprint    = {1307.2205}
}

@inproceedings{KelnerLOS14,
  author    = {Jonathan A. Kelner and
               Yin Tat Lee and
               Lorenzo Orecchia and
               Aaron Sidford},
  title     = {An Almost-Linear-Time Algorithm for Approximate Max Flow in Undirected
               Graphs, and its Multicommodity Generalizations},
  booktitle = {Proc.\ of the {ACM-SIAM} Symposium on Discrete Algorithms (SODA)},
  pages     = {217--226},
  year      = {2014},
  doi       = {10.1137/1.9781611973402.16},
  archivePrefix = {arXiv},
  eprint    = {1304.2338}
}

@inproceedings{DaitchS08,
  author    = {Samuel I. Daitch and
               Daniel A. Spielman},
  title     = {Faster approximate lossy generalized flow via interior point algorithms},
  booktitle = {Proc.\ of the {ACM} Symposium on Theory of Computing (STOC)},
  pages     = {451--460},
  year      = {2008},
  doi       = {10.1145/1374376.1374441},
  archivePrefix = {arXiv},
  eprint    = {0803.0988}
}

@inproceedings{RodittyTZ05,
  author    = {Liam Roditty and
               Mikkel Thorup and
               Uri Zwick},
  title     = {Deterministic Constructions of Approximate Distance Oracles and Spanners},
  booktitle = {Proc.\ of the International Colloquium on Automata, Languages and Programming (ICALP)},
  year      = {2005},
  pages     = {261--272},
  doi       = {10.1007/11523468_22}
}

@article{HansenKTZ15,
  author    = {Thomas Dueholm Hansen and
               Haim Kaplan and
               Robert E. Tarjan and
               Uri Zwick},
  title     = {Hollow Heaps},
  journal   = {{ACM} Trans. Algorithms},
  volume    = {13},
  number    = {3},
  pages     = {42:1--42:27},
  year      = {2017},
  doi       = {10.1145/3093240},
  archivePrefix = {arXiv},
  eprint    = {1510.06535},
  note      = {Announced at ICALP '15}
}

@book{Schrijver03,
  author    = {Alexander Schrijver},
  title     = {Combinatorial Optimization},
  publisher = {Springer},
  year      = {2003}
}

@book{KorteV00,
  author    = {Bernhard Korte and Jens Vygen},
  title     = {Combinatorial Optimization},
  publisher = {Springer},
  year      = {2000}
}

@article{Goldberg95,
  author    = {Andrew V. Goldberg},
  title     = {Scaling Algorithms for the Shortest Paths Problem},
  journal   = {{SIAM} Journal on Computing},
  volume    = {24},
  number    = {3},
  pages     = {494--504},
  year      = {1995},
  doi       = {10.1137/S0097539792231179},
  note	    = {Announced at SODA '93}
}

@inproceedings{Gall16,
  author    = {Le Gall, Fran{\c{c}}ois},
  title     = {Further Algebraic Algorithms in the Congested Clique Model and Applications
               to Graph-Theoretic Problems},
  booktitle = {Proc.\ of the International Symposium on Distributed Computing (DISC)},
  pages     = {57--70},
  year      = {2016},
  doi       = {10.1007/978-3-662-53426-7_5},
  archivePrefix = {arXiv},
  eprint    = {1608.02674}
}

@article{HenzingerKN18,
  author    = {Monika Henzinger and
               Sebastian Krinninger and
               Danupon Nanongkai},
  title     = {Decremental Single-Source Shortest Paths on Undirected Graphs in Near-Linear
               Total Update Time},
  journal   = {Journal of the {ACM}},
  volume    = {65},
  number    = {6},
  pages     = {36:1--36:40},
  year      = {2018},
  doi       = {10.1145/3218657},
  note      = {Announced at FOCS '14},
  archivePrefix = {arXiv},
  eprint    = {1512.08148}
}

@inproceedings{Bernstein09,
  author    = {Aaron Bernstein},
  title     = {Fully Dynamic $(2 + \epsilon)$ Approximate All-Pairs Shortest Paths
               with Fast Query and Close to Linear Update Time},
  booktitle = {Proc.\ of the {IEEE} Symposium on Foundations of Computer Science (FOCS)},
  pages     = {693--702},
  year      = {2009},
  doi       = {10.1109/FOCS.2009.16}
}

@inproceedings{CohenMSV16,
  author    = {Michael B. Cohen and
               Aleksander Madry and
               Piotr Sankowski and
               Adrian Vladu},
  title     = {Negative-Weight Shortest Paths and Unit Capacity Minimum Cost Flow
               in {$ \tilde O (m^{10/7} \log W) $} Time},
  booktitle = {Proc.\ of the {ACM-SIAM} Symposium on Discrete Algorithms (SODA)},
  year      = {2017},
  doi       = {10.1137/1.9781611974782.48},
  archivePrefix = {arXiv},
  eprint    = {1605.01717}
}

@article{AbboudBP18,
  author    = {Amir Abboud and
               Greg Bodwin and
               Seth Pettie},
  title     = {A Hierarchy of Lower Bounds for Sublinear Additive Spanners},
  journal   = {{SIAM} Journal on Computing},
  volume    = {47},
  number    = {6},
  pages     = {2203--2236},
  year      = {2018},
  doi       = {10.1137/16M1105815},
  archivePrefix = {arXiv},
  eprint    = {1607.07497}
}

@inproceedings{Bartal98,
  author    = {Yair Bartal},
  title     = {On Approximating Arbitrary Metrices by Tree Metrics},
  booktitle = {Proc.\ of the {ACM} Symposium on the Theory of Computing (STOC)},
  pages     = {161--168},
  year      = {1998},
  doi       = {10.1145/276698.276725}
}

@article{ElkinEST08,
  author    = {Michael Elkin and
               Yuval Emek and
               Daniel A. Spielman and
               Shang{-}Hua Teng},
  title     = {Lower-Stretch Spanning Trees},
  journal   = {{SIAM} Journal on Computing},
  volume    = {38},
  number    = {2},
  pages     = {608--628},
  year      = {2008},
  doi       = {10.1137/050641661},
  note      = {Announced at STOC '05},
  archivePrefix = {arXiv},
  eprint    = {cs/0411064}
}

@article{Elkin17,
  author    = {Michael Elkin},
  title     = {Distributed Exact Shortest Paths in Sublinear Time},
  journal   = {Journal of the {ACM}},
  volume    = {67},
  number    = {3},
  pages     = {15:1--15:36},
  year      = {2020},
  doi       = {10.1145/3387161},
  archivePrefix = {arXiv},
  eprint    = {1703.01939},
  note	    = {Announced at STOC '17} 
}

@article{Spencer97,
  author    = {Thomas H. Spencer},
  title     = {Time-work tradeoffs for parallel algorithms},
  journal   = {Journal of the {ACM}},
  volume    = {44},
  number    = {5},
  pages     = {742--778},
  year      = {1997},
  doi       = {10.1145/265910.265923},
  note	    = {Announced at SODA '91 and SPAA '91}
}

@article{KleinS97,
  author    = {Philip N. Klein and
               Sairam Subramanian},
  title     = {A Randomized Parallel Algorithm for Single-Source Shortest Paths},
  journal   = {Journal of Algorithms},
  volume    = {25},
  number    = {2},
  pages     = {205--220},
  year      = {1997},
  doi       = {10.1006/jagm.1997.0888},
  note	    = {Announced at STOC '92}
}

@article{BrodalTZ98,
  author    = {Gerth St{\o}lting Brodal and
               Jesper Larsson Tr{\"{a}}ff and
               Christos D. Zaroliagis},
  title     = {A Parallel Priority Queue with Constant Time Operations},
  journal   = {Journal of Parallel and Distributed Computing},
  volume    = {49},
  number    = {1},
  pages     = {4--21},
  year      = {1998},
  doi       = {10.1006/jpdc.1998.1425},
  note	    = {Announced at IPPS '97}
}

@article{ShiS99,
  author    = {Hanmao Shi and
               Thomas H. Spencer},
  title     = {Time-Work Tradeoffs of the Single-Source Shortest Paths Problem},
  journal   = {Journal of Algorithms},
  volume    = {30},
  number    = {1},
  pages     = {19--32},
  year      = {1999},
  doi       = {10.1006/jagm.1998.0968}
}

@article{Cohen97,
  author    = {Edith Cohen},
  title     = {Using Selective Path-Doubling for Parallel Shortest-Path Computations},
  journal   = {Journal of Algorithms},
  volume    = {22},
  number    = {1},
  pages     = {30--56},
  year      = {1997},
  doi       = {10.1006/jagm.1996.0813},
  note	    = {Announced at ISTCS '93}
}

@article{MeyerS03,
  author    = {Ulrich Meyer and
               Peter Sanders},
  title     = {{$\Delta$}-stepping: a parallelizable shortest path algorithm},
  journal   = {Journal of Algorithms},
  volume    = {49},
  number    = {1},
  pages     = {114--152},
  year      = {2003},
  doi       = {10.1016/S0196-6774(03)00076-2},
  note	    = {Announced at ESA '98}
}

@inproceedings{Blelloch0ST16,
  author    = {Guy E. Blelloch and
               Yan Gu and
               Yihan Sun and
               Kanat Tangwongsan},
  title     = {Parallel Shortest Paths Using Radius Stepping},
  booktitle = {Proc.\ of the {ACM} Symposium on Parallelism in Algorithms
               and Architectures (SPAA)},
  pages     = {443--454},
  year      = {2016},
  doi       = {10.1145/2935764.2935765},
  archivePrefix = {arXiv},
  eprint    = {1602.03881}
}

@inproceedings{ElkinN19,
  author    = {Michael Elkin and
               Ofer Neiman},
  title     = {Linear-Size Hopsets with Small Hopbound, and Constant-Hopbound Hopsets
               in {RNC}},
  booktitle = {Proc.\ of the {ACM} Symposium on Parallelism in Algorithms and Architectures (SPAA)},
  pages     = {333--341},
  year      = {2019},
  doi       = {10.1145/3323165.3323177}
}

@article{HuangP19,
  author    = {Shang{-}En Huang and
               Seth Pettie},
  title     = {{Thorup-Zwick} emulators are universally optimal hopsets},
  journal   = {Information Processing Letters},
  volume    = {142},
  pages     = {9--13},
  year      = {2019},
  doi       = {10.1016/j.ipl.2018.10.001},
  archivePrefix = {arXiv},
  eprint    = {1705.00327}
}

@book{AhujaMO93,
  author    = {Ravindra K. Ahuja and
               Thomas L. Magnanti and
               James B. Orlin},
  title     = {Network flows - theory, algorithms and applications},
  publisher = {Prentice Hall},
  year      = {1993},
  isbn      = {978-0-13-617549-0}
}
